\documentclass[journal,onecolumn,12pt]{IEEEtran}
\renewcommand{\baselinestretch}{1.5}
\usepackage{graphicx}                                      
\usepackage{amssymb,amsthm,bm}  
\usepackage[cmex10]{amsmath} 

\usepackage{multirow}   
\usepackage{tikz}                   
\newtheorem{theorem}{Theorem}
\newtheorem{lemma}{Lemma}

\newtheorem*{remark}{Remark}
\newtheorem{definition}{Definition}

\newtheorem{proposition}{Proposition}

\newcommand{\utag}[2]{\mathop{#2}\limits^{\text{(#1)}}}
\newcommand{\uref}[1]{(#1)}

\newcommand{\utagmodify}[2]{\mathop{#2}\limits^{(#1)}}

\long\def\symbolfootnote[#1]#2{\begingroup
\def\thefootnote{\fnsymbol{footnote}}\footnote[#1]{#2}\endgroup}

\usepackage{xcolor}
\definecolor{SkyBlue}{RGB}{240,255,255}
\definecolor{LavenderBlush}{RGB}{218,112,214}
\definecolor{BlanchedAlmond}{RGB}{255,235,205}

\newcommand{\cA}{\mathcal{A}}

\newcommand{\cD}{\mathcal{D}}

\newcommand{\cK}{\mathcal{K}}
\newcommand{\cL}{\mathcal{L}}
\newcommand{\cM}{\mathcal{M}}
\newcommand{\cN}{\mathcal{N}}

\newcommand{\cP}{\mathcal{P}}

\newcommand{\cR}{\mathcal{R}}
\newcommand{\cS}{\mathcal{S}}
\newcommand{\cT}{\mathcal{T}}

\newcommand{\cW}{\mathcal{W}}

\newcommand{\cZ}{\mathcal{Z}}
\usepackage{mathrsfs}

\newcommand{\sF}{\mathscr{F}}

\newcommand{\sS}{\mathscr{S}}

\newcommand{\sW}{\mathscr{W}}

\usepackage{tcolorbox}
\usepackage{amsmath}
\usepackage{amsthm}
\usepackage{graphicx}
\usepackage{float}
\usetikzlibrary{positioning}
\usetikzlibrary{math}
\usetikzlibrary{calc}
\usetikzlibrary{patterns}

\makeatletter
\newcommand{\gettikzxy}[3]{%
  \tikz@scan@one@point\pgfutil@firstofone#1\relax
  \edef#2{\the\pgf@x}%
  \edef#3{\the\pgf@y}%
}
\makeatother

\begin{document}
\title{On Secure Exact-repair Regenerating Codes with a Single Pareto Optimal Point}

\author{
\IEEEauthorblockN{Fangwei Ye, Shiqiu Liu, Kenneth W.~Shum, and Raymond W.~Yeung 
} \\
\IEEEauthorblockA{
Institute of Network Coding \& Department of Information Engineering \\
The Chinese University of Hong Kong \\
Shatin, N.T., Hong Kong \\
Email: \{fwye, sqliu, wkshum, whyeung\}@ie.cuhk.edu.hk}	
}

\maketitle

\begin{abstract}
The problem of exact-repair regenerating codes against eavesdropping attack is studied. The eavesdropping model we consider is that the eavesdropper has the capability to observe the data involved in the repair of a subset of $\ell$ nodes. An $(n,k,d,\ell)$ secure exact-repair regenerating code is an $(n,k,d)$ exact-repair regenerating code that is secure under this eavesdropping model. 
It has been shown that for some parameters $(n,k,d,\ell)$, the associated optimal storage-bandwidth tradeoff curve, which has one corner point, can be determined. The focus of this paper is on characterizing such  parameters. We establish a lower bound $\hat{\ell}$ on the number of wiretap nodes, and show that this bound is tight for the case $k = d = n-1$. 
\end{abstract}

\section{Introduction}
\label{Sec:intro}
Distributed storage systems (DSSs) have been widely researched because of the rapid growth in applications such as data center and cloud network. For data reliability, some redundancy must be added to the system. 
In the pioneering study \cite{Dimakis2010_DSS}, Dimakis \textit{et~al}.\ introduced a new class of codes called \emph{regenerating codes}, which substantially reduce the amount of data that need to be downloaded during the repair process.
In \cite{Dimakis2010_DSS}, a fundamental tradeoff between the amount of data stored in each node and the repair bandwidth was shown under the notion of functional repair, where the new replacement nodes only maintain the reconstruction property, that is, any $k$ out of $n$ nodes can reconstruct the file but do not maintain an exact copy of the failed node. 
On the other hand, under the notion of exact repair introduced in \cite{Rashmi2009}, the replacement node is required to recover exactly the same content that was stored in the failed node. However, a full characterization of the storage-bandwidth tradeoff curve of exact-repair regenerating codes appears to be more difficult and still remains open, and many attempts have been made along this line \cite{Tian2014_JSAC,Kumar2014_Outer,Duu2014,Tan2015_isit,Prakash2015,Elyasi2015,Tian2015_Layed,Rashmi2011_PM}.

In this paper, we consider the problem of exact-repair regenerating codes with an additional security requirement. Information-theoretically secure regenerating codes were first introduced by Pawar \textit{et~al.} \cite{Pawar2011}, in which they provided an upper bound on the maximum amount of information that can be securely stored in a system. Secure exact-repair regenerating codes at two extreme points, namely, the \emph{minimum bandwidth regenerating (MBR)} and \emph{minimum storage regenerating (MSR)} points, have been intensively studied in
\cite{Pawar2011,Shah2011_SecMBR,Goparaju2013,Rawat2014b}.
On the other hand, the optimal storage-bandwidth tradeoff curve under secure repair constraint has been studied in 
\cite{Tandon2014_allerton,Ye_isit,Tandon16,Ye,Shao}. 
In particular, the results in \cite{Tandon16,Ye} showed that the MBR point is the only corner point of the optimal storage-bandwidth tradeoff curve (or simply tradeoff curve) for some $(n,k,d,\ell)$, which contrasts sharply with the problem without the security constraint. 
Owing to a structural property of the tradeoff curve, if it has a single corner point, then it is completely characterized by that single point. Thus for the aforementioned cases investigated in  \cite{Tandon16,Ye}, the tradeoff curve is completely characterized by the MBR point.
Subsequently, 
Shao \textit{et~al.} \cite{Shao} found the first case where the optimal storage-bandwidth tradeoff curve has multiple corner points, and obtained a sufficient condition on the number of wiretap nodes where the rate region can be determined by a single corner point.
In this paper, we establish a lower bound $\hat{\ell}$ on the number of wiretap nodes, such that the optimal storage-bandwidth tradeoff curve has a single corner point if $\ell \geq \hat{\ell}$. In particular, the lower bound for the case $k=d=n-1$ is tight, which means that the optimal storage-bandwidth tradeoff curve has a single corner point if and only if $\ell \geq \hat{\ell}$. 

The remaining of this paper is organized as follows. In Section~\ref{Sec:formulation}, we describe the formulation of the problem. We give a threshold $\hat{\ell}$ for the number of wiretap nodes for the case $k = d$ in Section~\ref{Sec:k_equal_d}, and results for $k<d$ are stated in Section~\ref{Sec:k_less_d}. We conclude the paper in Section~\ref{Sec:conclusion}.

\section{Problem statement and notations}
\label{Sec:formulation}
Following the setting in \cite{Dimakis2010_DSS}, we assume that there is a secure distributed storage system consisting of $n$ active storage nodes $\cN:= \{1,2,\ldots,n\}$ for storing a file $\sF$ of $B_s$ message symbols, and each node can store $\alpha$ symbols. When a node fails, a new replacement node with the same storage capacity $\alpha$ connects to any $d \ (\geq k)$ nodes chosen from the remaining $n-1$ nodes arbitrarily and downloads $\beta$ symbols from each of them to regenerate the failed node.  Moreover, any legitimate data collector can reconstruct the original file by connecting to any $k$ of the $n$ active nodes. 
We assume that there exists an eavesdropper Eve who is able to observe the repair data for a subset of nodes with cardinality $\ell \ (< k)$. It not only can observe the information stored in node $i$ but also all the data transmitted from the other $d$ helper nodes to repair the node $i$ when it fails.  

Let $M$ be the uniformly distributed random variable representing the file to be stored in the system. The support set of $M$ is denoted by $\cM$, and $B_s$ is used to denote the entropy of the message variable, \textit{i.e.,} $B_s=H(M)$. Let $Z$ be a random variable independent on the message variable $M$, called the key, that takes value in an alphabet $\cZ$ according to the uniform distribution. 
\begin{figure}[htbp]
\centering
\includegraphics[width=8cm]{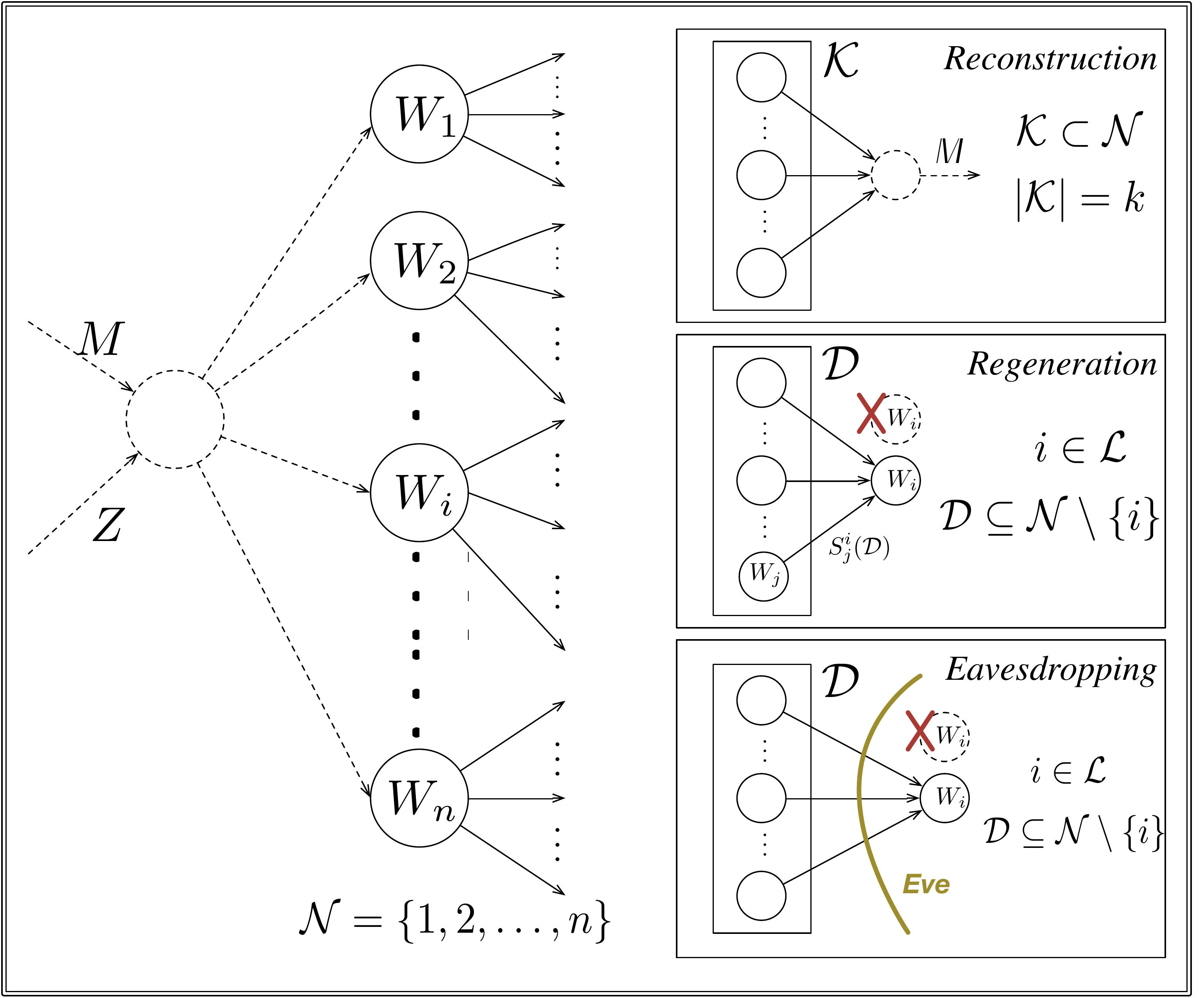}
\caption{System Model}
\label{pic1}
\end{figure}
As illustrated in Fig.~\ref{pic1}, we assume that the message and key are generated at an auxiliary source node and are directly available to all storage nodes in the system. 
For $i \in \cN$, let $W_i$ denote the data stored in the $i$-th node. $S_i^j(\cD)$ denotes the variable transmitted from node $i$ for repairing the node $j$ for a given set of helper nodes $\cD \subset \cN$, where $|\cD|= d$ and $i \in \cD$. 
$S_i^i(\cD)$ is defined as a constant for any possible $\cD$. Denote $\sW:=\{W_i:i \in \cN\}$ and $\sS:=\{S_i^j(\cD): j \in \cN, \cD \subseteq \cN \backslash \{j\}, |\cD|= d, i \in \cD\}$.
Each node has identical storage capacity $\alpha$ and limited transmission $\beta$ in repairing any single failure. Thus, assume without loss of generality that each $\cW_i$ takes value in a common alphabet $\cW$ and each $S_i^j(\cD)$ takes value in a common alphabet $\cS$, where $\alpha = \log |\cW|$ and $\beta = \log |\cS|$.   
For any set of wiretap nodes $\cL \subset \cN$ such that $|\cL|=\ell$, the information wiretapped by Eve is denoted by $Y_{\cL}$, where $Y_{\cL}$ is defined as $Y_{\cL}:=\{S_i^j(\cD): j \in \cL, \cD \subseteq \cN \backslash \{j\}, |\cD|= d, i \in \cD \}$. For any integer $i \leq j \leq n$, denote $[i]:=\{1,\ldots,i\}$, and $[i:j]:=\{i,\ldots,j\}$.

Next, we formally define a secure distributed storage system based on an exact-repair regenerating code formally. In the rest of the paper, when we refer to a secure distributed storage system, we always assume that it is based on an exact-repair regenerating code.
\begin{definition}
	An $(n,k,d,\ell)$ secure distributed storage system (SDSS) based on an exact-repair regenerating code consists of a set of encoding functions and decoding functions $(F,G,\Phi,\Psi)$, which can be described as follows.
\begin{itemize}
	\item Message encoding functions: $F =\{f_i:i \in \cN\}$ is a collection of message encoding functions, where  $f_i$ maps the message and key to the information  stored in the $i$-th node,
\[ f_i: \cM \times \cZ \rightarrow  \cW .\]
	\item Message decoding functions: $G =\{g_{\cK}: \cK \subset \cN, |\cK|= k \}$ consists of $\binom{n}{k}$ message decoding functions, where
\[g_{\cK}: \cW^{\cK}  \rightarrow \cM .\]
It maps the coded information stored in node $i$, $i \in \cK$.
	\item Repair encoding functions: $\Phi =\{\phi_{i,j,\cD}: j \in \cN, i \in \cD, \cD \subseteq \cN \backslash \{j\}, |\cD|= d \}$ consists of $n d \binom{n-1}{d} $ repair encoding functions, where
\[\phi_{i,j,\cD}: \cW \rightarrow \cS\]
maps the coded information in node $i$ to the information transmitted for repairing node $j$ for a given choice of helper nodes set $\cD$.
	\item Repair decoding functions: $\Psi = \{\psi_{j, \cD}: j \in \cN, \cD \subseteq \cN \backslash \{j\}, |\cD|= d\}$ consists of $n \binom{n-1}{d} $ repair decoding functions, where
\[\psi_{j, \cD}:   \cS^{\cD} \rightarrow \cW\]
maps the information from a set $\cD$ of help nodes to the information stored in the failed node.
\end{itemize}
\end{definition}

An $(n,k,d,\ell)$ secure distributed storage system is required to satisfy the following criteria:
\begin{itemize}
	\item (\emph{Reconstruction property}) the file can be retrieved from the contents stored in any $k$ out of $n$ storage nodes:
		\begin{equation}
		\label{eq:Reconstruction_condition}
			H(M|W_{\cK})=0, \forall \cK \subseteq \cN, |\cK|=k,
		\end{equation}
		where $W_{\cK}$ is defined as $W_{\cK}:=\{W_i : i \in \cK\}$.
	\item (\emph{Regeneration property}) any $d$ out of $n-1$ nodes can repair the failed $j$-th node:
		\begin{equation}
		\label{eq:temp_Regeneration_condition}
			H \left( W_j | S^j(\cD) \right) = 0, \forall \cD \subseteq \cN \backslash \{j\},  j \in \cN,
		\end{equation}
		where $S^j(\cD):=\{S_{i}^j(\cD): i \in \cD \}$. 
	\item (\emph{Security condition})	
		\begin{equation}
		\label{eq:temp_Secure_Condition}
			H(M|Y_{\cL})=H(M), \forall \cL \subseteq \cN.
		\end{equation} 
\end{itemize}

Any collection of encoding and decoding functions $(F,G,\Phi,\Psi)$ satisfying all these three criteria will naturally induce a \emph{secure exact-repair regenerating code} associated with the triple $(B_s, \alpha, \beta)$.  We can always assume that $B_s > 0$ because otherwise the code can not be used for storing any information.
Under this assumption, we can define the normalized pair $(\bar{\alpha},\bar{\beta})$ by
\[\bar{\alpha}:=\frac{\alpha}{B_s} ~~\text{and}~~  \bar{\beta}:=\frac{\beta}{B_s}.\]
A normalized pair $(\bar{\alpha},\bar{\beta})$ is also called an \emph{operating point}. We may use ``the pair'' or ``the point'' interchangeably in the following sections.  
With the normalized pair $(\bar{\alpha},\bar{\beta})$, we introduce the following definition.
\begin{definition}
 	A normalized pair $(\bar{\alpha},\bar{\beta})$ is achievable if there exists a secure exact-repair regenerating code that achieves $(\bar{\alpha},\bar{\beta})$. The collection of all achievable pairs $(\bar{\alpha},\bar{\beta})$ is referred to as the zero-error achievable region $ \cR_{n,k,d,\ell}$.
\end{definition}
It follows directly from the definition that if the pair $(\bar{\alpha},\bar{\beta})$ is achievable, then any pair $(\bar{\alpha}+\delta_1,\bar{\beta}+\delta_2)$ is also achievable, where $\delta_1, \delta_2 \geq 0$. Thus, the achievable region can be fully characterized if and only if the boundary is known. To be consistent with the terminology in the literature, we call the collection of points on the boundary the \emph{storage-bandwidth tradeoff curve}. 
For a given $(n,k,d,\ell)$ secure distributed storage system, its \emph{secrecy capacity} is defined as the maximum file size $C_s(\alpha,\beta)$ that can be stored in the system such that $(\alpha/B_s, \beta/B_s)$ is achievable, \textit{i.e.},
\begin{equation}
	C_s(\alpha,\beta):= \sup \left\{B_s: (\alpha/B_s,\beta/B_s) \in \cR_{n,k,d,\ell}\right\}.
\end{equation}
Clearly, determining the secrecy capacity for any given $\alpha$ and $\beta$ is equivalent to characterizing the storage-bandwidth tradeoff curve.

In \cite{Shah2011_SecMBR}, the following point is proved to be achievable for any $(n,k,d,\ell)$-SDSS:
\begin{equation}
 \left(   \frac{d}{\Gamma_{k,d,\ell} }, \frac{1}{\Gamma_{k,d,\ell}}   \right) \in  \cR_{n,k,d,\ell},
\end{equation}
where $\Gamma_{k,d,\ell}:=\sum_{i=\ell}^{k-1} (d-i)$.
For notational simplicity, denote
\begin{equation}
\label{eq:SRK}
 	(\hat{\alpha},\hat{\beta}): = \left(  \frac{d}{\Gamma_{k,d,\ell} }, \frac{1}{\Gamma_{k,d,\ell} }  \right).
\end{equation}

An interesting finding in \cite{Tandon16} and \cite{Ye} is that for some cases,
the storage-bandwidth tradeoff curve under the security condition is completely characterized by the single corner point specified in \eqref{eq:SRK}, \textit{i.e.}, the achievable rate region is given exactly by
\begin{equation}
\label{eq:rate_region}
	\cR_{n,k,d,\ell}= \left\{(\bar{\alpha},\bar{\beta}): \bar{\alpha} \geq \hat{\alpha}, \bar{\beta} \geq \hat{\beta} \right\} .
\end{equation}
\begin{remark}
We will prove in Appendix~\ref{app:optimality} that the point as defined in \eqref{eq:SRK} must be on the optimal tradeoff curve. Therefore, if the optimal tradeoff curve has only one corner point, then it must be $(\hat{\alpha},\hat{\beta})$.
\end{remark}
Subsequently in \cite{Shao}, the first case that the storage-bandwidth tradeoff curve has multiple corner points was found, and a sufficient condition for the number of wiretap nodes was given for the storage-bandwidth tradeoff curve of an SDSS to have a single corner point. In this paper, we will focus on finding parameters $(n,k,d,\ell)$ such that the tradeoff curve has this behavior.

In the remaining of this paper, we only consider the case that $d = n - 1$. Since any $(n' > d + 1,k,d,\ell)$ system has an $(n = d + 1,k, d, \ell)$ sub-system. If the sub-system satisfies that $\bar{\alpha} \geq \hat{\alpha}, \bar{\beta} \geq \hat{\beta}$, then the $(n'> d+1,k,d,\ell)$-SDSS must satisfy the same constraints. Moreover, $(\hat{\alpha},\hat{\beta})$ is also achievable for $(n'> d+1,k,d,\ell)$, and hence if the tradeoff curve for $(n = d + 1,k, d, \ell)$ has a single corner point, then the tradeoff curve for $(n' > d + 1,k, d, \ell)$ must also have this behavior. Therefore, all results obtained under the setting $d=n-1$ in this paper also hold for $n'>d+1$.
Under this setting, we can largely simplify our aforementioned notations. 
When repairing the failed node,
all the remaining nodes are helper nodes. Therefore we can drop $\cD$ in the notations $S_i^j(\cD)$ and $S^j(\cD)$. Specifically, we will write $S_i^j(\cD)$ as $S_i^j$ and write $S^j(\cD)$ as $S^j$ because $\cD=\cN \backslash \{j\}$ is implicit. Denote $S^{\cL}:=\{S^j:j \in \cL\}$, and obviously $S^{\cL}$ is identical to $Y_{\cL}$. 
Then, the regeneration property can be written as
\begin{equation}
\label{eq:Regeneration_condition}
	H \left( W_j | S^j \right) = 0, \forall   j \in \cN.
\end{equation}
Similarly, we can rewrite the security condition as	
\begin{equation}
\label{eq:Secure_Condition}
	H(M|S^{\cL})=H(M), \forall \cL \subseteq \cN.
\end{equation}

We follow the discussion for
symmetrical regenerating codes in \cite{Tian2014_JSAC}. 
A code is said to be a \emph{entropy-symmetrical regenerating code} (or simply symmetrical regenerating code) if for any $X_{\cA} \subseteq \sW \cup \sS$ and any permutation $\pi$ on $\cN$, we have $H(X_{\cA}) = H\left(\pi(X_{\cA})\right)$, where
\[\pi(X_{\cA}) := \{\pi(X_{i}): i \in \cA \},\]
and
\[
\pi(X_{i}) := 
\begin{cases} 
      W_{\pi(i)}, & \text{if} \ X_i=W_i, \\
      S_{\pi(i)}^{\pi(j)}, &  \text{if} \ X_i = S_i^j. 
 \end{cases} 
\]
It has been shown in \cite{Ye} that assuming that the secure exact-repair regenerating code is symmetrical
does not incur any loss of generality when we consider $\cR_{n,k,d,\ell}$. Therefore, we may invoke this symmetrical assumption in our argument without explicitly mentioning it. Under this setting, we can let $H(W_i) = \alpha$ and $H(S_i^j)=\beta$.  
For notational simplicity, let us define 
\begin{equation}
	\cP:=\left\{(k,d,\ell):\cR_{n=d+1,k,d,\ell}= \left\{(\bar{\alpha},\bar{\beta}): \bar{\alpha} \geq \hat{\alpha}, \bar{\beta} \geq \hat{\beta} \right\} \right\}.
\end{equation}
\begin{remark}
Since $(k, d, \ell = 0) \notin \cP$ for $k \geq 2$ and $(k=1, d, \ell = 0) \in \cP$ (which can be seen by considering the repetition code), we assume that $\ell \geq 1$ in this paper.
\end{remark}

Now, consider any subset $\cT$ of $\sW \cup \sS$ such that $H(W_{\cK} | \cT)=0$. Then by the reconstruction property \eqref{eq:Reconstruction_condition} and security constraint \eqref{eq:Secure_Condition}, we can obtain an upper bound on $B_s$ as follows:
\begin{align}
  B_s & =  H (M) \nonumber \\
      & =  H (M | S^{\cL} ) - H (M | \cT, S^{\cL} ) \nonumber \\
      & =  I (M ; \cT | S^{\cL}  ) \nonumber \\
      & \leq  H (\cT | S^{\cL} ).\label{eq:SC_Outer}
\end{align}
By letting $\cT=\left\{S_i^j: j < i \leq n, 1 \leq j \leq k \right\}$ and $\cL=\{1,\ldots,\ell\}$, we can obtain the upper bound in \cite{Pawar2011}:
\begin{equation}
\label{eq:betabound}
	B_s \leq  \sum_{i=\ell}^{k-1} (d-i) \beta,
\end{equation}
which can also be written as $\bar{\beta} \geq \hat{\beta}$.

Since $\bar{\beta} \geq \hat{\beta}$ and $(\hat{\alpha},\hat{\beta}) \in  \cR_{n,k,d,\ell}$ for any $(n,k,d,\ell)$-SDSS, the triple $(k,d,\ell) \in \cP$ if and only if $\bar{\alpha} \geq \hat{\alpha}$, or equivalently
\begin{equation}
\label{eq:alphabound}
	B_s \leq \frac{\Gamma_{k,d,\ell}}{d} \alpha.
\end{equation}
Therefore, we only need to prove that $B_s \leq \frac{\Gamma_{k,d,\ell}}{d} \alpha$ to conclude that
 $(k,d,\ell) \in \cP$.

\section{Threshold for $k = d$}
\label{Sec:k_equal_d}
We will establish in the next theorem a threshold $\hat{\ell}$ for the number of wiretap nodes for those systems whose optimal tradeoff curve has a single corner.
\begin{theorem}
\label{Thm:k_equal_d}
	 For any fixed $d$, the triple $(k=d, d, \ell) \in \cP$ if and only if $\ell \geq \hat{\ell} := \lceil\frac{1}{4}(d-1)\rceil$.
\end{theorem}
\begin{remark}
	It was shown in \cite{Shao} that if $\ell \geq \ell^{\star}:=\lceil (\sqrt{d}-1)^2\rceil$, then $(k=d,d,\ell) \in \cP$. When $d$ is large, $\hat{\ell} \approx \ell^{\star}/4$. Thus our bound not only is a significant improvement over the previous bound but also tight.
\end{remark}
In the remaining of this section, we will prove Theorem~\ref{Thm:k_equal_d}. We will invoke the setting $k = d = n - 1$ from time to time without explicitly mentioning it. Before presenting the details, we outline the proof here. 

In Subsection~\ref{subsec:k_eq_d_smaller}, we will show that if $\ell < \hat{\ell}$, then there exists one achievable point $(\bar{\alpha},\bar{\beta})$ such that $\bar{\alpha}<\hat{\alpha}$, which implies that $(k,d,\ell) \notin \cP$. 
The proof of the achievability of this point is largely borrowed from a code construction in \cite{Tian2015_Layed}. 

To prove that $(k,d,\ell) \in \cP$ for $\ell \geq \hat{\ell}$, we only need to show \eqref{eq:alphabound} for $\ell \geq \hat{\ell}$. By letting $\cT= W_{[k]}$ and $\cL=\{1,\ldots,\ell\}$ in \eqref{eq:SC_Outer}, we see that the secrecy capacity $B_s$ is upper bounded by  
\begin{align}
	B_s & \leq  H\left(W_{[k]}|S^{[\ell]} \right) = H\left(W_{[\ell+1:k]}|S^{[\ell]} \right). \label{eq:k_eq_d_SC} 
\end{align}
Thus, it is sufficient for us to prove that
\[H\left(W_{[\ell+1:k]}|S^{[\ell]} \right) \leq \frac{\Gamma_{k,d,\ell}}{d} \alpha,\] for $\ell \geq \hat{\ell}$. This will be proved by induction on $\ell$ in Subsection~\ref{subsec:k_eq_d_equal}.

\subsection{$\ell < \hat{\ell}$ implies that $(k,d,\ell) \notin \cP$}
\label{subsec:k_eq_d_smaller}
We first roughly review the code construction for $(n,k,d,\ell=0)$ exact-repair regenerating codes with $k=d=n-1$ in \cite{Tian2015_Layed}, where the code construction is based on duplicated combination block design. 
Considering a block design over the domain (node index) $\cN=\{1,\ldots,n\}$, the design there can be viewed as an exhaustive list of all $r$-combinations ($n \geq r$) of $\cN$. Each block forms a $(r,r-1)$ erasure code, and symbols in different blocks are independent. 

In particular, we consider block size $r=3$ in this subsection. We have a design $C(r,n)=\{B_1,\ldots,B_m\}$, where each block $B_i$ is a unique $3$-subset of $\cN$ and $m=\binom{n}{3}$.
For each $3$-subset $B_i=\{b_{i_1}, b_{i_2}, b_{i_3}\}$, let $X_i$ and $Y_i$ be independent random variables uniformly on a sufficient large field $\mathbb{F}$, and we consider a corresponding vector for each $B_i$ such that $\bm{b_i} = (b_{i_1}, b_{i_2}, b_{i_3})$ where $1 \leq  b_{i_1}< b_{i_2}< b_{i_3} \leq n$. Then, the encoding is as the following:
\begin{itemize}
	\item $X_{i}$ is stored in node $b_{i_1}$;
	\item $Y_{i}$ is stored in node $b_{i_2}$;
	\item $X_{i}+Y_{i}$ is stored in node $b_{i_3}$.
\end{itemize}

Let $X_i$ and $X_j$ ($Y_i$ and $Y_j$) be independent random variables for $i \neq j$. We can see that in this construction, 
\[\alpha = \binom{n-1}{2}, ~~ \beta = n-2, ~~ B_s = 2\binom{n}{3}, \]
and hence
\[\left( \bar{\alpha}, \bar{\beta} \right)= \left(\frac{\binom{n-1}{2}}{2\binom{n}{3}}, \frac{n-2}{2\binom{n}{3}} \right)  \in \cR_{d+1,d,d,0}.\]
See more details in \cite{Tian2015_Layed}.

Therefore, following the same argument in \cite{CY2011_Sec}, we know  
that there exists an $(n,k=n-1,d=n-1,\ell)$ secure exact-repair regenerating code with
$\alpha = \binom{n-1}{2}$, $\beta =n-2$ and $B_s = 2\binom{n-\ell}{3}$ if the field size is large enough, and so
\[\left( \bar{\alpha}, \bar{\beta} \right)= \left(\frac{\binom{n-1}{2}}{2\binom{n-\ell}{3}}, \frac{n-2}{2\binom{n-\ell}{3}} \right) \in \cR_{d+1,d,d,\ell}.\]

If an integer $\ell$ satisfying that $\ell < \hat{\ell} = \left\lceil  \frac{1}{4}(d-1) \right\rceil$, we have $\ell< \frac{1}{4}(d-1) = \frac{1}{4}(n-2)$. As such, we have
\[\bar{\alpha} - \hat{\alpha}  =  \frac{\binom{n-1}{2}}{2\binom{n-\ell}{3}} -  \frac{n-1}{\binom{n-\ell}{2}} 
= \frac{ (4\ell+2-n)(n-1)  }{2(n-\ell)(n-\ell-1)(n-\ell-2)} < 0.\]
Therefore, we know that if $\ell < \hat{\ell}$, there exists one achievable point $\left( \bar{\alpha}, \bar{\beta} \right)$ such that $\bar{\alpha} - \hat{\alpha} < 0$, which substantiates that if $\ell < \hat{\ell}$ then $(k,d, \ell) \notin \cP$.

\subsection{$\ell \geq \hat{\ell}$ implies that $(k,d,\ell) \in \cP$}
\label{subsec:k_eq_d_equal}
In this subsection, we will show that 
\[H \left(W_{[\ell+1:k]}|S^{[\ell]}\right) \leq \frac{\Gamma_{k,d,\ell}}{d}\alpha,\]
for $\ell \geq \hat{\ell}$ by induction.
For any subset $\cA \subseteq \cN$, denote $S_i^{\cA}:=\{S_i^j: j \in \cA\}$, $S^i_{\cA}:=\{S_j^i: j \in \cA\}$ and $S_{\cA}:=\{S_i^j:i,j \in \cA, i>j\}$.

\begin{proposition}
\label{pro:k_eq_d_1}
	For $k=d$, if $\cT \subseteq \sW \cup \sS$ satisfies $H\left(W_{[k]}|\cT\right) = 0$, then 
	\begin{equation}
		H\left(\cT \right) = H\left(W_{[k]} \right).
	\end{equation}
\end{proposition}
\begin{proof}
	Since $k=d$, $W_{[k]}$ can determine any subsets of $\sW \cup \sS$, and so $H\left(W_{[k]} \right) \geq H\left(\cT \right)$. From $H\left(W_{[k]}|\cT\right) = 0$, we have $H\left(W_{[k]} \right) \leq H\left(\cT \right)$, and hence $H\left(\cT \right) = H\left(W_{[k]} \right)$.
\end{proof}

The following lemma gives a class of upper bounds on $H \left(W_{[\ell+1:k]}|S^{[\ell]}\right)$.

\begin{lemma}
\label{lem:k_eq_d}
For any $(n=d+1,k = d, d,\ell)$ secure exact-repair regenerating codes, we have
\begin{equation}
\label{eq:lem_k_eq_d}
		H \left(W_{[\ell+1:k]}|S^{[\ell]}\right) \leq \frac{d+1 - t }{3} \alpha - \frac{d+1 - t }{3} H\left(S_n^{[t]}\right) + \frac{d+1 - t}{6} H\left(S^{t+1}|S^{[t]}\right)  - \sum_{i=t+1}^{\ell} H\left(S^i|S^{[i-1]}\right),
\end{equation}
for any $t = 0,\ldots,\ell-1$.
\end{lemma}
\begin{proof}
	See Appendix \ref{App:lem_k_eq_d}.
\end{proof}
Since $\ell \geq 1$, there always exists an upper bound on $H \left(W_{[\ell+1:k]}|S^{[\ell]}\right)$ for $t = 0$. When $t = 0$, $S_n^{[t]}$ is regarded as a constant.
For notational simplicity, denote the right-hand side of \eqref{eq:lem_k_eq_d} by $f(d,\ell,t)$, 
where $t = 0,\ldots,\ell-1$.
Then the following proposition is immediate.
\begin{proposition}
\label{Prop:k_eq_d}
For any $(n=d+1,k=d,d,\ell)$ secure exact-repair regenerating codes,
	\begin{equation}
	\label{eq:k_eq_d_convex}
		H \left(W_{[\ell+1:k]}|S^{[\ell]}\right) \leq  \sum_{t=0}^{\ell-1} \mu_t \ f(d,\ell,t), 
	\end{equation} 
for any $\bm{\mu} = (\mu_0,\ldots,\mu_{\ell-1})$ such that
\[\sum_{t=0}^{\ell-1} \mu_t=1,\] 
and
\[\mu_t \geq 0, ~ t=0,\ldots,\ell-1.\]
\end{proposition}

With these preparations, we start to prove that 
\begin{equation}
\label{eq:k_eq_d_induction}
	H \left(W_{[\ell+1:k]}|S^{[\ell]}\right) \leq \frac{\Gamma_{k,d,\ell}}{d}\alpha
\end{equation}
for $\ell \geq \hat{\ell}$ by induction on $\ell$.

First, for the base case $\ell = \hat{\ell}$, \eqref{eq:k_eq_d_induction} becomes
\begin{equation}
\label{eq:k_eq_d_base_case}
	H\left(W_{[\hat{\ell}+1:k]}|S^{[\hat{\ell}]}\right) \leq \frac{\Gamma_{k,d,\hat{\ell}}}{d}\alpha.
\end{equation}
From Proposition \ref{Prop:k_eq_d}, we know that 
\[H\left(W_{[\hat{\ell}+1:k]}|S^{[\hat{\ell}]}\right) \leq \sum_{t=0}^{\hat{\ell}-1} \mu_t \ f(d,\hat{\ell},t),\] for any $\bm{\mu}$ satisfying
\begin{equation}
\label{eq:k_eq_d_base_case_condition_1}
	\sum_{t=0}^{\hat{\ell}-1} \mu_t=1,
\end{equation}
and
\begin{equation}
\label{eq:k_eq_d_base_case_condition_2}
	\mu_t \geq 0,~ t=0,\ldots,\hat{\ell}-1.
\end{equation}
In particular, we can let
\begin{equation}
\label{eq:k_eq_d_base_coeff_1}
	\mu_t =
	\begin{cases}
	 	 \frac{1}{2}  \binom{n-\hat{\ell}}{2}   \frac{n-2\hat{\ell}-1+t}{\binom{n-t}{4}},  & 1\leq t \leq \hat{\ell}-3, \\
	 	 \frac{6(n-\hat{\ell}-3)}{(n-\hat{\ell}+1)(n-\hat{\ell}+2)}, &  t = \hat{\ell}-2, ~ \hat{\ell} \geq 3, \\
	 	 \frac{6}{n-\hat{\ell}+1}, &   t=\hat{\ell}-1, ~ \hat{\ell} \geq 2, 
	\end{cases}
\end{equation}
and 
\begin{equation}
\label{eq:k_eq_d_base_coeff_0}
	\mu_0 = 1 - \sum_{j=1}^{\hat{\ell}-1} \mu_{j}.
\end{equation}
For this choice of $\bm{\mu}$, \eqref{eq:k_eq_d_base_case_condition_1} is obvious satisfied, and we only need to verify that \eqref{eq:k_eq_d_base_case_condition_2} is also satisfied.
\begin{proposition}
\label{Prop:convex}
$\bm{\mu}=(\mu_0,\ldots,\mu_{\hat{\ell}-1})$ as defined in \eqref{eq:k_eq_d_base_coeff_1} and \eqref{eq:k_eq_d_base_coeff_0} satisfies
\[\mu_t \geq 0, ~ t=0,\ldots,\hat{\ell}-1.\]
\end{proposition}
\begin{proof}
 See Appendix \ref{App:convex}.
\end{proof}
It remains to show that 
\[\sum_{t=0}^{\hat{\ell}-1} \mu_t  f(d,\hat{\ell},t) \leq \frac{\Gamma_{k,d,\hat{\ell}}}{d}\alpha.\]
Towards this end, consider
\begin{align*}
	\sum_{t=0}^{\hat{\ell}-1} \mu_t  f(d,\hat{\ell},t) 
	& = \sum_{t=0}^{\hat{\ell}-1} \mu_t  \left( \frac{d+1 - t }{3} \alpha - \frac{d+1 - t }{3} H\left(S_n^{[t]}\right) + \frac{d+1 - t}{6} H\left(S^{t+1}|S^{[t]}\right)  - \sum_{i=t+1}^{\hat{\ell}} H\left(S^i|S^{[i-1]}\right)    \right) \\
	& = \left( \sum_{t=0}^{\hat{\ell}-1}  \frac{d+1-t}{3}\mu_t \right) \alpha  - \sum_{t=0}^{\hat{\ell}-1}  \frac{d+1-t}{3}\mu_t H\left(S_n^{[t]}\right) + \sum_{t=0}^{\hat{\ell}-1} \frac{d+1 - t}{6} \mu_t  H\left(S^{t+1}|S^{[t]}\right) \\
	& ~~ - \sum_{t=0}^{\hat{\ell}-1}  \sum_{i=t+1}^{\hat{\ell}}   \mu_t  H\left(S^i|S^{[i-1]}\right)     \\
	& = \left( \sum_{t=0}^{\hat{\ell}-1}  \frac{d+1-t}{3}\mu_t \right) \alpha  - \sum_{t=0}^{\hat{\ell}-1}  \frac{d+1-t}{3}\mu_t H\left(S_n^{[t]}\right) + \sum_{t=0}^{\hat{\ell}-1} \frac{d+1 - t}{6} \mu_t  H\left(S^{t+1}|S^{[t]}\right) \\
	& ~~ - \sum_{i=1}^{\hat{\ell}} \left( \sum_{t=0}^{i-1} \mu_t \right) H\left(S^i|S^{[i-1]}\right) \\
	& = \left( \sum_{t=0}^{\hat{\ell}-1}  \frac{d+1-t}{3}\mu_t \right) \alpha  - \sum_{t=0}^{\hat{\ell}-1}  \frac{d+1-t}{3}\mu_t H\left(S_n^{[t]}\right) + \sum_{t=0}^{\hat{\ell}-1} \frac{d+1 - t}{6} \mu_t  H\left(S^{t+1}|S^{[t]}\right) \\
	& ~~ - \sum_{t=0}^{\hat{\ell}-1} \left( \sum_{j=0}^{t} \mu_j \right) H\left(S^{t+1}|S^{[t]}\right),
\end{align*}
where in the last step we replace $i$ by $t+1$ and $t$ by $j$.

By letting 
\[b_t = \frac{n-t}{3} \mu_t,\] 
and
\[c_t = \frac{n-t}{6} \mu_{t} - \sum_{j=0}^{t} \mu_j,\]
we obtain
\begin{equation}
\label{eq:k_eq_d_base_convex}
	\sum_{t=0}^{\hat{\ell}-1} \mu_t  f(d,\hat{\ell},t) \leq \left( \sum_{t=0}^{\hat{\ell}-1}  b_t \right) \alpha  - \sum_{t=0}^{\hat{\ell}-1}  b_t H\left(S_n^{[t]}\right) + \sum_{t=0}^{\hat{\ell}-1} c_t  H\left(S^{t+1}|S^{[t]}\right).
\end{equation}

We separately discuss the case $\hat{\ell}=1$ here.
When $\hat{\ell}=1$,  clearly we have $\mu_0=1$, and then  \eqref{eq:k_eq_d_base_convex} becomes
\begin{equation*}
	 f(d,\hat{\ell},t=0) \leq b_0 \alpha  +  c_0  H\left(S^{1}\right)  = \frac{n}{3} \alpha + \left( \frac{n-6}{6}\right)  H\left(S^{1}\right).
\end{equation*}
Since $\hat{\ell} = \left \lceil \frac{1}{4}(d-1) \right \rceil = \left \lceil \frac{1}{4}(n-2) \right \rceil =1 $, we know that $n \leq 6$, and then we have
\[f(d,\hat{\ell},t=0) \utag{a}{\leq} \frac{n}{3} \alpha + \left(\frac{n-6}{6}\right) \alpha = \frac{1}{2}(n-2) \alpha = \frac{\Gamma_{d,d,1}}{d} \alpha, \]
where \uref{a} follows because $H(S^1) \geq H(W_1)=\alpha$. We have completed the proof for $\hat{\ell}=1$. 

For $\hat{\ell} \geq 2$, \eqref{eq:k_eq_d_base_convex} can be written as
\begin{align}
	\sum_{t=0}^{\hat{\ell}-1} \mu_t  f(d,\hat{\ell},t) 
	& \leq \left( \sum_{t=0}^{\hat{\ell}-1}  b_t \right) \alpha  - \sum_{t=0}^{\hat{\ell}-1}  b_t H\left(S_n^{[t]}\right) + \sum_{t=0}^{\hat{\ell}-1} c_t  H\left(S^{t+1}|S^{[t]}\right) \nonumber \\
	& = \left(\sum_{t=0}^{\hat{\ell}-1}  b_t \right) \alpha  - b_1 \beta - \sum_{t=2}^{\hat{\ell}-1}  b_t H\left(S_n^{[t]}\right) +c_0 H\left(S^{1}\right) + \sum_{t=1}^{\hat{\ell}-1} c_t  H\left(S^{t+1}|S^{[t]}\right). \label{eq:K_eq_d_last_temple1}
\end{align}
\begin{proposition}
\label{Prop:geq0}
For $\hat{\ell} \geq 2$, $c_t \geq 0$ for $t= 0, \ldots, \hat{\ell}-1$, and $c_{\hat{\ell}-1} = 0$.
\end{proposition}
\begin{proof}
	See Appendix \ref{App:c_geq_0}.
\end{proof}
Since
\begin{align}
	H\left(S^{t+1}|S^{[t]}\right) & = H\left(S_{[t]\cup[t+2:n]}^{t+1}|S^{[t]}\right) \nonumber \\
	& = H\left(S_{[t+2:n]}^{t+1}|S^{[t]}\right) \nonumber \\
	& = \sum_{j=t+2}^{n} H\left(S_j^{t+1}|S^{[t]}, S_{[t+2:j-1]}^{t+1}  \right) \nonumber \\
								  & \leq \sum_{j=t+2}^{n} H\left(S_j^{t+1}|S_j^{[t]} \right) \nonumber \\
								  & \utag{a}{=} (d-t) H\left(S_n^{t+1}|S_n^{[t]} \right), \label{eq:k_eq_d_bound_S} 
\end{align}
where \uref{a} follows from the symmetry,
we can further bound \eqref{eq:K_eq_d_last_temple1} as follows:
\begin{align}
	\sum_{t=0}^{\hat{\ell}-1} \mu_t  f(d,\hat{\ell},t) 
	& \leq \left(\sum_{t=0}^{\hat{\ell}-1}  b_t \right) \alpha  - b_1 \beta - \sum_{t=2}^{\hat{\ell}-1}  b_t H\left(S_n^{[t]}\right) +c_0 H\left(S^{1}\right) + \sum_{t=1}^{\hat{\ell}-1} c_t  H\left(S^{t+1}|S^{[t]}\right) \nonumber \\
	& \leq \left(\sum_{t=0}^{\hat{\ell}-1}  b_t \right) \alpha  - b_1 \beta - \sum_{t=2}^{\hat{\ell}-1}  b_t H\left(S_n^{[t]}\right) +c_0 H\left(S^{1}\right) + \sum_{t=1}^{\hat{\ell}-1} c_t (d-t) H\left(S_n^{t+1}|S_n^{[t]}\right) \nonumber \\
	& = \left(\alpha \sum_{t=0}^{\hat{\ell}-1}  b_t  - b_1 \beta +c_0 H\left(S^{1}\right) \right)  - \sum_{t=2}^{\hat{\ell}-1}  b_t H\left(S_n^{[t]}\right)  + \sum_{t=1}^{\hat{\ell}-1} c_{t} (d-t) \left(H\left(S_n^{[t+1]}\right) - H\left(S_n^{[t]}\right) \right) \nonumber \\
	& = \left(\alpha \sum_{t=0}^{\hat{\ell}-1}  b_t  - b_1 \beta +c_0 H\left(S^{1}\right) \right)  - \sum_{t=2}^{\hat{\ell}-1}  b_t H\left(S_n^{[t]}\right)  + \sum_{t=2}^{\hat{\ell}} c_{t-1} (d-t+1) H\left(S_n^{[t]}\right) \nonumber \\
	& ~~ - \sum_{t=1}^{\hat{\ell}-1} c_t (d-t) H\left(S_n^{[t]}\right) \nonumber \\
	& \utag{a}{=} \left(\alpha \sum_{t=0}^{\hat{\ell}-1}  b_t  - b_1 \beta +c_0 H\left(S^{1}\right)- c_1 (d-1) \beta \right)  + \sum_{t=2}^{\hat{\ell}-1} \left(c_{t-1} (d-t+1) -  c_t (d-t) - b_t \right)  H\left(S_n^{[t]}\right)  \nonumber \\
	& \utag{b}{\leq} \left(\alpha \sum_{t=0}^{\hat{\ell}-1}  b_t  - b_1 \beta +c_0 d \beta- c_1 (d-1) \beta \right)  + \sum_{t=2}^{\hat{\ell}-1} \left(c_{t-1} (d-t+1) -  c_t (d-t) - b_t \right)  H\left(S_n^{[t]}\right)  \nonumber,
\end{align}
where \uref{a} follows from $c_{\hat{\ell}-1}=0$, and \uref{b} follows because $c_0 \geq 0$ and $H\left(S^1\right) \leq d \beta$.
By letting 
\[T_1 = \sum_{t=0}^{\hat{\ell}-1} b_t,\]
\[T_2 = b_1  - c_0 d  +  c_1 (d-1)  ,\]
and 
\[\lambda_t = c_{t-1}  (d+1-t) - c_t  (d-t) - b_t, t=2,\ldots,\hat{\ell}-1,\]
we have
\[\sum_{t=0}^{\hat{\ell}-1} \mu_t  f(d,\hat{\ell},t)  \leq T_1 \alpha - T_2 \beta - \sum_{t=2}^{\hat{\ell}-1} \lambda_t H(S_n^{[t]}).\]
\begin{proposition}
\label{Prop:namda}
	 For $\hat{\ell} \geq 3$, $\lambda_t = 0$ for $t= 2, \ldots, \hat{\ell}-1$. 
\end{proposition}
\begin{proof}
	See Appendix \ref{APP:namba}.
\end{proof}

From Proposition \ref{Prop:namda}, we obtain 
\[\sum_{t=0}^{\hat{\ell}-1} \mu_t  f(d,\hat{\ell},t)  \leq T_1 \alpha - T_2 \beta.\]
\begin{proposition}
\label{prop:k_eq_d_final}
	$T_2 \geq 0$, $T_1  - \frac{T_2}{d} = \frac{\Gamma_{k,d,\hat{\ell}}}{d}$. 
\end{proposition}
\begin{proof}
	See Appendix \ref{App:k_eq_d_final}.
\end{proof}

Finally, we can substantiate that 
\[\sum_{t=0}^{\hat{\ell}-1} \mu_t  f(d,\hat{\ell},t)  \leq  T_1 \alpha - T_2 \beta \utag{a}{\leq} T_1 \alpha - \frac{T_2}{d} \alpha \leq \frac{\Gamma_{k,d,\hat{\ell}}}{d} \alpha,\]
where \uref{a} follows from $d \beta \geq \alpha$.
Therefore, the base case holds, that is,
\[H\left(W_{[\hat{\ell}+1:k]}|S^{[\hat{\ell}]} \right)  \leq \frac{\Gamma_{k,d,\hat{\ell}}}{d} \alpha.\]

Now, we start the inductive step to show that for any $\ell \geq \hat{\ell}+1$, if $H \left(W_{[\ell:k]}|S^{[\ell-1]} \right) \leq \frac{\Gamma_{k,d,\ell-1}}{d} \alpha$, 
then $H \left(W_{[\ell+1:k]}|S^{[\ell]} \right) \leq \frac{\Gamma_{k,d,\ell}}{d} \alpha$.

First, assume that 
\begin{equation}
\label{eq:k_eq_d_larger_assump}
	H \left(W_{[\ell:k]}|S^{[\ell-1]} \right) \leq \frac{\Gamma_{k,d,\ell-1}}{d} \alpha,
\end{equation}
for some $\ell \geq \hat{\ell}+1$.

Then, consider 
\begin{align}
	H \left(W_{[\ell+1:k]}|S^{[\ell]} \right) 
							& = H \left(W_{[\ell+1:k]},S^{[\ell]} \right) - H \left(S^{[\ell]} \right) \nonumber \\
							& \utag{a}{=} H \left(W_{[\ell:k]},S^{[\ell-1]} \right) - H \left(S^{[\ell-1]} \right) - H \left(S^{\ell}|S^{[\ell-1]} \right) \nonumber \\
							& = H \left(W_{[\ell:k]}|S^{[\ell-1]} \right) - H \left(S^{\ell}| S^{[\ell-1]}\right) \nonumber \\
							& \utag{b}{\leq} \frac{\Gamma_{k,d,\ell-1}}{d} \alpha - H \left(S^{\ell}| S^{[\ell-1]}\right), \label{eq:k_eq_d_larger:temple1}
\end{align}
where \uref{a} follows from Proposition \ref{pro:k_eq_d_1}, and \uref{b} follows from \eqref{eq:k_eq_d_larger_assump}.
Also, we have
\begin{align}
	 H \left(W_{[\ell+1:k]}|S^{[\ell]}\right) 
	& \leq H \left(S^{[\ell+1:k]}|S^{[\ell]}\right) \nonumber \\
	& = \sum_{i=\ell+1}^k H(S^i|S^{[i-1]}) \nonumber \\
	& \utag{a}{\leq} \sum_{i=\ell+1}^k (n-i) H(S_n^i|S_n^{[i-1]}) \nonumber \\
	& = (n-k) H(S_n^{[k]}) + \left( \sum_{i=\ell+1}^{k-1}  H(S_n^{[i]}) \right) - (n-\ell-1) H(S_n^{[\ell]})  \nonumber \\
	& =  \left( \sum_{i=\ell+1}^k  H(S_n^{[i]}) \right) - (n-\ell-1) H(S_n^{[\ell]}), \label{eq:k_eq_d_larger:temple3}
\end{align}
where \uref{a} follows from \eqref{eq:k_eq_d_bound_S}.
Moreover, the following lemma gives another upper bound on $H \left(W_{[\ell+1:k]}|S^{[\ell]}\right)$.
\begin{lemma}
\label{lem:k_eq_d_revised}
For any $(n=d+1,k = d, d,\ell)$ secure exact-repair regenerating codes, we have
	\begin{equation}
	\label{eq:k_eq_d_larger:temple2}
			H(W_{[\ell+1:k]}|S^{[\ell]}) \leq \frac{1}{2}(k-\ell+1) \alpha - \frac{1}{2} \sum_{i=\ell-1}^{k-1} H(S_n^{[i]}) +  \frac{1}{4}(k-\ell-2)  H(S^{\ell}|S^{[\ell-1]}).
	\end{equation}
\end{lemma}
\begin{proof}
	The lemma can be proved by modifying the proof of Lemma~\ref{lem:k_eq_d}. See details in Appendix \ref{App:k_eq_d_revised}. 
\end{proof}

We now have three upper bounds on $H(W_{[\ell+1:k]}|S^{[\ell]})$. Similar to what we did in the previous subsection, we will take a particular convex combination of \eqref{eq:k_eq_d_larger:temple1}, \eqref{eq:k_eq_d_larger:temple3} and \eqref{eq:k_eq_d_larger:temple2} to obtain the desired upper bound on $H \left(W_{[\ell+1:k]}|S^{[\ell]}\right)$. Denote the coefficients associated with \eqref{eq:k_eq_d_larger:temple1}, \eqref{eq:k_eq_d_larger:temple3} and \eqref{eq:k_eq_d_larger:temple2} by $v_1$, $v_2$ and $v_3$.


If $\ell=k-1$, from \eqref{eq:k_eq_d_larger:temple3}, we obtain
\[ H \left(W_{k}|S^{[k-1]}\right) \leq H(S_n^{k}|S_n^{[k-1]}) \leq \frac{1}{k} H(S_n^{[k]}) \leq \frac{1}{k} \alpha = \frac{\Gamma_{d,d,d-1}}{d} \alpha.\]
Hence, by letting $v_2=1$ and $v_1=v_3=0$, we obtain that 
\[H \left(W_{[\ell+1:k]}|S^{[\ell]}\right) \leq  \frac{\Gamma_{k,d,\ell}}{d} \alpha,\]
for $\ell = k-1$.

For $\ell \leq k-2$, let
\[v_1  =  \frac{(k-\ell-2)(n-\ell-1)}{4(n-1) + (n-\ell+1)(k-\ell-2)},\]
\[v_2  =  \frac{2(k+\ell-2)}{4(n-1) + (n-\ell+1)(k-\ell-2)  },\]
and 
\[v_3 =  \frac{4(n-\ell-1)}{4(n-1) + (n-\ell+1)(k-\ell-2)}.\]
Clearly, $v_1,v_2,v_3 \geq 0$ for $\ell \leq k-2$. Also, we have
\begin{align*}
	v_1 + v_2 +v_3  = \frac{(k-\ell-2)(n-\ell-1) + 2(k+\ell-2) + 4(n-\ell-1)}{4(n-1) + (n-\ell+1)(k-\ell-2)} 
				    = 1.
\end{align*}

Therefore, $H \left(W_{[\ell+1:k]}|S^{[\ell]}\right)$ is upper-bounded by $v_1 \eqref{eq:k_eq_d_larger:temple1}+  v_2\eqref{eq:k_eq_d_larger:temple3} + v_3 \eqref{eq:k_eq_d_larger:temple2}$ as follows:
\begin{align*}
	H \left(W_{[\ell+1:k]}|S^{[\ell]}\right)  
			&  \leq v_1 \left(\frac{\Gamma_{k,d,\ell-1}}{d} \alpha - H \left(S^{\ell}| S^{[\ell-1]}\right) \right) + v_2 \left( \left( \sum_{i=\ell+1}^k  H(S_n^{[i]}) \right) - (n-\ell-1) H(S_n^{[\ell]}) \right) \\
			& ~~ + v_3 \left( \frac{1}{2}(k-\ell+1) \alpha - \frac{1}{2} \sum_{i=\ell-1}^{k-1} H(S_n^{[i]}) +  \frac{1}{4}(k-\ell-2)  H(S^{\ell}|S^{[\ell-1]}) \right) \\
 			& = \left( v_1 \frac{\Gamma_{k,d,\ell-1}}{d} +   \frac{v_3}{2}(k-\ell+1 ) \right) \alpha +  \left(  \frac{v_3}{4}(k-2 - \ell) - v_1  \right) H \left(S^{\ell}| S^{[\ell-1]}\right)  \\
			& ~~  + v_2 \left( \left( \sum_{i=\ell+1}^k  H(S_n^{[i]}) \right) - (n-\ell-1) H(S_n^{[\ell]}) \right)  - \frac{v_3}{2}\sum_{i=\ell}^k H(S_n^{[i-1]}) \\
			& \utag{a}{=} \left( v_1 \frac{\Gamma_{k,d,\ell-1}}{d} +   \frac{v_3}{2}(k-\ell+1 ) \right) \alpha \\
			& ~~  + v_2 \left( \left( \sum_{i=\ell+1}^k  H(S_n^{[i]}) \right) - (n-\ell-1) H(S_n^{[\ell]}) \right)  - \frac{v_3}{2}\sum_{i=\ell}^k H(S_n^{[i-1]}) \\
			& \utag{b}{=}  \frac{\Gamma_{k,d,\ell}}{d}  \alpha+ v_2 \left( \left( \sum_{i=\ell+1}^k  H(S_n^{[i]}) \right) - (n-\ell-1) H(S_n^{[\ell]}) \right)  - \frac{v_3}{2}\sum_{i=\ell}^k H(S_n^{[i-1]}),
\end{align*}
where \uref{a} follows because  $v_1 = \frac{1}{4}(k-\ell-2)v_3$, and \uref{b} can be justified as follows:
\begin{align*}
	v_1 \frac{\Gamma_{k,d,\ell-1}}{d} +   \frac{v_3}{2}(k-\ell+1 ) 
	& =  \left( \frac{k-\ell-2}{4} \frac{\Gamma_{k,d,\ell-1}}{d} +   \frac{k-\ell+1}{2} \right)   v_3 \\
	& = \frac{(n-\ell+1)(n-\ell)(k-\ell-2)+4d(k-\ell+1)}{8d}  v_3 \\	
		& = \frac{(n-\ell) (n-\ell-1)}{2(n-1)} \\
		& = \frac{\Gamma_{k,d,\ell}}{d}.
\end{align*}

Finally, we claim that 
\begin{equation}
\label{eq:k_eq_d_larger_temple10}
	v_2 \left( \left( \sum_{i=\ell+1}^k  H(S_n^{[i]}) \right) - (n-\ell-1) H(S_n^{[\ell]}) \right)  - \frac{v_3}{2}\sum_{i=\ell}^k H(S_n^{[i-1]}) \leq 0.
\end{equation}
Towards this end, by re-arranging the left-hand side of \eqref{eq:k_eq_d_larger_temple10}, we have 
\begin{align*}
	& v_2 \left( \left( \sum_{i=\ell+1}^k  H(S_n^{[i]}) \right) - (n-\ell-1) H(S_n^{[\ell]}) \right)  - \frac{v_3}{2}\sum_{i=\ell}^k H(S_n^{[i-1]})  \\
	& = v_2  \left( \sum_{i=\ell+1}^k  H(S_n^{[i]}) \right) - v_2 (n-\ell-1) H(S_n^{[\ell]})  - \frac{v_3}{2}\sum_{i=\ell-1}^{k-1} H(S_n^{[i]}) \\
	& = v_2    H(S_n^{[k]}) +  \left(v_2 - \frac{v_3}{2} \right) \left( \sum_{i=\ell+1}^{k-1}  H(S_n^{[i]}) \right) - \left(v_2 (n-\ell-1) + \frac{v_3}{2} \right) H(S_n^{[\ell]})- \frac{v_3}{2} H(S_n^{[\ell-1]}).
\end{align*}

Since $v_2 - \frac{v_3}{2} \geq 0$ for $\ell \geq 1$, we have 
\begin{align*}
	& v_2 \left( \left( \sum_{i=\ell+1}^k  H(S_n^{[i]}) \right) - (n-\ell-1) H(S_n^{[\ell]}) \right)  - \frac{v_3}{2}\sum_{i=\ell}^k H(S_n^{[i-1]})  \\
	& = v_2    H(S_n^{[k]}) +  \left(v_2 - \frac{v_3}{2} \right) \left( \sum_{i=\ell+1}^{k-1}  H(S_n^{[i]}) \right) - \left(v_2 (n-\ell-1) + \frac{v_3}{2} \right) H(S_n^{[\ell]})- \frac{v_3}{2} H(S_n^{[\ell-1]}) \\
	& \utag{a}{\leq}  \left( \frac{k}{\ell} v_2 + \left(v_2 - \frac{v_3}{2} \right) \sum_{i=\ell+1}^{k-1} \frac{i}{\ell} - \left(v_2 (n-\ell-1) + \frac{v_3}{2} \right) - \frac{v_3}{2} \frac{\ell-1}{\ell}
	\right) H(S_n^{[\ell]}) \\
	& =\frac{1}{2\ell} \left( 2kv_2 + \left(2v_2 - v_3 \right) \left(\sum_{i=\ell+1}^{k-1} i\right) - \left( 2 v_2 (n-\ell-1) - v_3  \right)\ell - v_3(\ell-1) \right) H(S_n^{[\ell]}) \\
	& \utag{b}{=} 0,
\end{align*}
where \uref{a} follows because $\frac{1}{i} H(S_n^{[i]}) \leq \frac{1}{j} H(S_n^{[j]})$ for $n> i \geq j$, which is the consequence of Han's inequality and the symmetry of the problem, and 
\uref{b} can be justified by substituting $v_2$ and $v_3$.

\section{Sufficient condition of wiretap nodes for $k < d$}
\label{Sec:k_less_d}
In this section, we consider the general setting that $k < d$. We will provide a lower bound $\hat{\ell}$ on the number of wiretap nodes such that if $\ell \geq \hat{\ell}$, then $(k,d,\ell) \in \cP$.
Shao \textit{et~al.} \cite{Shao} showed that $(k,d,\ell) \in \cP$ for $\ell \geq \ell^{\star}$. 
It will be shown that $\hat{\ell} \leq \ell^{\star}$.

\subsection{Our approach}
\label{Subsec:approach}
By letting $\cK=[k]$ and $\cL=[\ell]$ in \eqref{eq:SC_Outer}, we obtain that for any given $d$, $k$ and $\ell$, the secrecy capacity $B_s$  is upper bounded by  
\[B_s \leq H\left(\cT|S^{[\ell]} \right),\]
for any $\cT$ such that $H\left(W_{[k]}|\cT \right) = 0$.

Similar to what we did in the last section, we will select $\cT$ in different ways to obtain a number of upper bounds on $B_s$, and then take a convex combination of them to derive an upper bound that depends only on $\alpha$. 
Consider any set of variables $\cT = \{S^{[\ell]}\}\cup\{X_y: y=\ell+1,\ldots,k\}$, where $X_y$ can either be $W_y$ or $S^y$. 
Then 
\begin{equation}
\label{eq:k_less_d_X}
 	B_s \leq  \sum_{y=\ell+1}^k H\left (X_y|S^{[\ell]}, X_{[\ell+1:y-1]} \right).
\end{equation} 
We can use a $(k-\ell)$-length binary vector $\mathbf{q} := (q_{\ell+1},\ldots,q_{k})$ to represent the choices of $X_y$, $\ell+1 \leq y \leq k$, where 
\begin{equation}
\label{eq:k_less_d_b_j}
	q_y = 
	\begin{cases} 
	      0, & \text{if} \ X_{y} = W_{y}, \\
	      1, & \text{if} \ X_{y} = S^{y}.
	\end{cases}
\end{equation}
Clearly, each possible $\mathbf{q}$ induces an upper bound on $B_s$.

By symmetry we know that $H\left (X_y|S^{[\ell]}, X_{[\ell+1:y-1]} \right)$ depends on $\{q_{\ell+1},\ldots,q_{y}\}$ only through $q_{y}$ and $\sum_{i=\ell+1}^{y-1} q_i$. 
Hence, we have
\begin{equation}
\label{eq:k_less_d_symmetry}
	H\left (X_y|S^{[\ell]}, X_{[\ell+1:y-1]} \right) = H\left (X_y| S^{[t_y]}, W_{[t_y+1:y-1]}  \right),
\end{equation}
where
\begin{equation}
\label{eq:k_less_d_t}
	t_y=\ell+ \sum_{i=\ell+1}^{y-1} q_i.
\end{equation}
The following lemma gives upper bounds on $H\left (W_y| S^{[t_y]}, W_{[t_y+1:y-1]} \right)$ and $H\left (S^y| S^{[t_y]}, W_{[t_y+1:y-1]} \right)$.
\begin{lemma}
\label{lem:k_less_d_tech}
For any $y=\ell+1,\ldots,k$,
\begin{equation}
\label{eq:k_less_d_S_bound}
	H\left (S^y| S^{[t_y]}, W_{[t_y+1:y-1]}\right) \leq \frac{d+1-y}{d-t_y} H\left(S^{t_y+1}|S^{[t_y]} \right), ~ t_y=\ell,\ldots,y-1,
\end{equation}
and
\begin{equation}
\label{eq:k_less_d_W_bound}
	H\left (W_y| S^{[t_y]}, W_{[t_y+1:y-1]}\right) \leq 
	\begin{cases}
		\alpha -  H \left (S^{[y-2]}_n \right)   +\frac{d+1-y}{d+1-t_y} H \left (S^{t_y}| S^{[t_y-1]} \right)-  H \left (S^{y-1}|S^{[y-2]} \right),  & \ell \leq t_y \leq y-2, \\
		\alpha -  H \left (S^{[y-1]}_n \right),  & t_y=y-1.
	\end{cases}
\end{equation}
\end{lemma}  
\begin{proof}
See Appendix \ref{App:k_less_d_tech}.
\end{proof}
By combining \eqref{eq:k_less_d_X}, \eqref{eq:k_less_d_symmetry}, \eqref{eq:k_less_d_S_bound} and \eqref{eq:k_less_d_W_bound}, we can obtain an upper bound on $B_s$ for any given $\mathbf{q}$.
By examining \eqref{eq:k_less_d_S_bound} and \eqref{eq:k_less_d_W_bound}, we see that the right-hand sides of them may contain the terms $\alpha$, $H\left (S^{j}| S^{[j-1]} \right)$ for $j=\ell,\ldots,k$, and $H\left(S^{[j]}_n \right)$ for $j=\ell,\ldots,k-1$. 
Hence, let us specify the mapping $f$ from any $(k-\ell)$-length binary vector to the corresponding upper bound, which can be written as
\begin{equation}
\label{eq:k_less_d_temple_fQ}
	f\left(\mathbf{q}\right) =  \left( \sum_{j=\ell}^{k-1} \nu_j \right) \alpha -  \sum_{j=\ell}^{k-1} \nu_j H\left(S_n^{[j]}\right) +  \sum_{j=\ell}^{k} \mu_j H\left(S^j|S^{[j-1]}\right),
\end{equation}
where $\nu_j$ and $\mu_j$ can be determined by the given $\mathbf{q}$. Note that from \eqref{eq:k_less_d_W_bound} we know that the coefficient of $\alpha$ can be determined by the sum of the coefficients of $H\left(S_n^{[j]}\right)$ for $j=\ell,\ldots,k-1$.

Furthermore, we consider an $m\times(k-\ell)$ binary matrix 
\begin{equation}
\label{eq:k_less_d_Q}
	Q = 
\begin{bmatrix}
    \mathbf{q}_1 \\
    \mathbf{q}_2 \\
    \vdots \\
    \mathbf{q}_m    
\end{bmatrix}	
= 
\begin{bmatrix}
    q_{1,\ell+1} & q_{1,\ell+2} & \cdots & q_{1,k}\\
    q_{2,\ell+1} & q_{2,\ell+2} & \cdots & q_{2,k}\\
    \vdots       & \vdots       & \ddots & \vdots\\
    q_{m,\ell+1} & q_{2,\ell+2} & \cdots & q_{m,k}   
\end{bmatrix},
\end{equation}
where each $\mathbf{q}_x, 1 \leq x \leq m$  is some binary row vector defined in \eqref{eq:k_less_d_b_j}, and the first column of $Q$ is labeled by the index $\ell+1$ for consistency.  The parameter $t_y$ (cf. \eqref{eq:k_less_d_t}) in Lemma~\ref{lem:k_less_d_tech} corresponding to the row $\mathbf{q}_x$, where $1 \leq x \leq m$, is given by
 \[t_{x,y}=\ell + \sum_{i=\ell+1}^{y-1} q_{x,i}.\]
For each $\mathbf{q}_x$, we can obtain from \eqref{eq:k_less_d_temple_fQ} the upper bound 
\begin{equation}
	f\left(\mathbf{q}_x\right) =  \alpha \sum_{j=\ell}^{k-1} \nu_{x,j}  -  \sum_{j=\ell}^{k-1} \nu_{x,j} H\left(S_n^{[j]}\right) +  \sum_{j=\ell}^{k} \mu_{x,j} H\left(S^j|S^{[j-1]}\right).
\end{equation}
With a slight abuse of notations, we write
\begin{align*}
	f(Q) 
	& = \sum_{x=1}^m f\left(\mathbf{q}_x\right) \\
	& = \sum_{x=1}^m \left( \alpha \sum_{j=\ell}^{k-1} \nu_{x,j}  -  \sum_{j=\ell}^{k-1} \nu_{x,j} H\left(S_n^{[j]}\right) +  \sum_{j=\ell}^{k} \mu_{x,j} H\left(S^j|S^{[j-1]}\right)  \right) \\
	& = \alpha \sum_{j=\ell}^{k-1} \sum_{x=1}^m \nu_{x,j}  -  \sum_{j=\ell}^{k-1} \sum_{x=1}^m \nu_{x,j} H\left(S_n^{[j]}\right) +  \sum_{j=\ell}^{k} \sum_{x=1}^m \mu_{x,j} H\left(S^j|S^{[j-1]}\right).  
\end{align*}
By denoting $\bar{\nu}_{j} = \frac{1}{m} \sum_{x=1}^m \nu_{x,j}$ and $\bar{\mu}_{j} = \frac{1}{m} \sum_{x=1}^m \bar{\mu}_{x,j}$, we have
\begin{equation}
\label{eq:k_less_d_fQ}
	f(Q)= m \left(\alpha \sum_{j=\ell}^{k-1} \bar{\nu}_j  -  \sum_{j=\ell}^{k-1} \bar{\nu}_j H\left(S_n^{[j]}\right) +  \sum_{j=\ell}^{k} \bar{\mu}_j H\left(S^j|S^{[j-1]}\right)\right).
\end{equation}
It is clear that $f(Q)$ is an upper bound on $m B_s$. By dividing $m$ on both sides of \eqref{eq:k_less_d_fQ}, we have
\begin{equation}
\label{eq:k_less_d_Qeff}
	B_s \leq \frac{1}{m} f(Q) =  \left( \sum_{j=\ell}^{k-1} \bar{\nu}_j \right)\alpha -  \sum_{j=\ell}^{k-1} \bar{\nu}_j H\left(S_n^{[j]}\right) + \sum_{j=\ell}^{k} \bar{\mu}_j H\left(S^j|S^{[j-1]}\right).
\end{equation}
Clearly, for any $(k,d,\ell)$, if there exists a $m \times(k-\ell)$ matrix $Q$ satisfying $\frac{1}{m} f(Q) \leq \frac{\Gamma_{k,d,\ell} }{d} \alpha$, then $(k,d,\ell) \in \cP$.

Now, we claim that if the conditions 
\begin{align}
	\sum_{j=\ell}^{k-1} \bar{\nu}_j & = \frac{\Gamma_{k,d,\ell} }{d}, \label{eq:k_less_d_Condition1} \\
	\bar{\mu}_j & \geq 0, j=\ell,\ldots,k, \label{eq:k_less_d_Condition2} \\
	\bar{\delta_j} & \geq 0, j=\ell+1,\ldots,k   \label{eq:k_less_d_Condition3}
\end{align}
are satisfied, where 
\begin{equation*}
\bar{\delta_j} = (d+1-j)\bar{\mu}_j - \sum_{i= j}^{k-1} \bar{\nu}_i,  j=\ell+1,\ldots,k,
\end{equation*}
then right hand side of \eqref{eq:k_less_d_Qeff} is upper bounded by $\frac{\Gamma_{k,d,\ell}}{d} \alpha$.

To see this, focus on the right hand side of \eqref{eq:k_less_d_Qeff}. By recalling from \eqref{eq:k_eq_d_bound_S} that 
\[H\left(S^j|S^{[j-1]}\right) \leq (d+1-j) H\left(S_n^j|S_n^{[j-1]}\right),\]
we have
\begin{align}
	\frac{1}{m} f(Q) & = \left( \sum_{j=\ell}^{k-1} \bar{\nu}_j \right)\alpha -  \sum_{j=\ell}^{k-1} \bar{\nu}_j H\left(S_n^{[j]}\right) + \sum_{j=\ell}^{k} \bar{\mu}_j H\left(S^j|S^{[j-1]}\right) \nonumber \\
	& \utag{a}{=} \frac{\Gamma_{k,d,\ell} }{d} \alpha -  \sum_{j=\ell}^{k-1} \bar{\nu}_j H\left(S_n^{[j]}\right) + \sum_{j=\ell}^{k} \bar{\mu}_j H\left(S^j|S^{[j-1]}\right) \nonumber \\
	& \utag{b}{\leq} \frac{\Gamma_{k,d,\ell}}{d} \alpha - \sum_{j=\ell}^{k-1} \bar{\nu}_j H\left(S_n^{[j]}\right) + \sum_{j=\ell}^{k} (d+1-j)\bar{\mu}_j H\left(S_n^j|S_n^{[j-1]}\right), \label{eq:k_less_d_new_temple2}
\end{align}
where \uref{a} follows from \eqref{eq:k_less_d_Condition1} and \uref{b} follows from \eqref{eq:k_less_d_Condition2}.

Since
\begin{align*}
	\sum_{j = \ell}^{k-1} \bar{\nu}_j H (S_n^{[j]}) 
	& =  \sum_{j = \ell}^{k-1} \bar{\nu}_j \left( H (S_n^{[\ell]}) + \sum_{i = \ell+1 }^{j} H (S_n^i | S_n^{[i - 1]}) \right)\\
  & =  \left( \sum_{j = \ell}^{k-1} \bar{\nu}_j \right) H (S_n^{[\ell]}) + \sum_{j = \ell}^{k-1} \sum_{i = \ell+1}^{j} \bar{\nu}_j H (S_n^i|S_n^{[i - 1]})   \\
  & =  \left( \sum_{j = \ell}^{k-1} \bar{\nu}_j \right) H (S_n^{[\ell]}) + \sum_{i = \ell+1}^{k-1} \sum_{j = i}^{k-1} \bar{\nu}_j H (S_n^i|S_n^{[i - 1]})   \\
  & \utag{c}{=}  \left( \sum_{j = \ell}^{k-1} \bar{\nu}_j \right) H (S_n^{[\ell]}) + \sum_{j = \ell+1}^{k-1} \left(\sum_{i = j}^{k-1} \bar{\nu}_i\right) H (S_n^j|S_n^{[j - 1]}),
\end{align*}
where in \uref{c}, the indices $i$ and $j$ in the double summation are renamed as $j$ and $i$, respectively,
we obtain 
\begin{align*}
	\frac{1}{m} f(Q)
	& \leq \frac{\Gamma_{k,d,\ell} }{d} \alpha - \sum_{j = \ell}^{k-1} \bar{\nu}_j H (S_n^{[j]}) + \sum_{j=\ell}^{k} (d+1-j)\bar{\mu}_j H\left(S_n^j|S_n^{[j-1]}\right) \\
	& = \frac{\Gamma_{k,d,\ell} }{d} \alpha - \left( \sum_{j = \ell}^{k-1} \bar{\nu}_j \right) H (S_n^{[\ell]}) - \sum_{j = \ell+1}^{k-1} \left(\sum_{i = j}^{k-1} \bar{\nu}_i\right) H (S_n^j|S_n^{[j - 1]}) + \sum_{j=\ell}^{k} (d+1-j)\bar{\mu}_j H\left(S_n^j|S_n^{[j-1]}\right) \\
	& = \frac{\Gamma_{k,d,\ell} }{d} \alpha -\frac{\Gamma_{k,d,\ell} }{d} H (S_n^{[\ell]}) - \sum_{j = \ell+1}^{k} \left(\sum_{i = j}^{k-1} \bar{\nu}_i\right) H (S_n^j|S_n^{[j - 1]}) + \sum_{j=\ell}^{k} (d+1-j)\bar{\mu}_j H\left(S_n^j|S_n^{[j-1]}\right) \\
	& \utag{d}{\leq} \frac{\Gamma_{k,d,\ell} }{d} \alpha -\frac{\Gamma_{k,d,\ell} }{d} H (S_n^{[\ell]}) - \sum_{j = \ell+1}^{k-1} \left(\sum_{i = j}^{k-1} \bar{\nu}_i\right) H (S_n^j|S_n^{[j - 1]}) + \sum_{j=\ell+1}^{k} (d+1-j)\bar{\mu}_j H\left(S_n^j|S_n^{[j-1]}\right)  \\
	& ~~ +  \frac{d+1-\ell}{\ell} \bar{\mu}_\ell  H\left(S_n^{[\ell]}\right)   \\
	& = \frac{\Gamma_{k,d,\ell} }{d} \alpha - \left( \frac{\Gamma_{k,d,\ell} }{d}- \frac{d+1-\ell}{\ell} \bar{\mu}_\ell \right) H (S_n^{[\ell]})  + \sum_{j=\ell+1}^{k} \left( (d+1-j)\bar{\mu}_j -\left( \sum_{i= j}^{k-1} \bar{\nu}_i  \right) \right) H\left(S_n^j|S_n^{[j-1]}\right)  \\
	& = \frac{\Gamma_{k,d,\ell} }{d} \alpha -  \left(\frac{\Gamma_{k,d,\ell} }{d}- \frac{d+1-\ell}{\ell} \bar{\mu}_\ell \right) H (S_n^{[\ell]})  + \sum_{j=\ell+1}^{k} \bar{\delta_j} H\left(S_n^j|S_n^{[j-1]}\right),
\end{align*}
where \uref{d} follows from Han's inequality.

Since $H\left(S_n^i|S_n^{[i-1]}\right) \leq H\left(S_n^j|S_n^{[j-1]}\right)$ for $i \geq j$,
we have
\begin{align}
	\frac{1}{m} f(Q)
	& \utag{e}{\leq}  \frac{\Gamma_{k,d,\ell} }{d} \alpha -  \left(\frac{\Gamma_{k,d,\ell} }{d}- \frac{d+1-\ell}{\ell} \bar{\mu}_\ell \right) H (S_n^{[\ell]})  + \sum_{j=\ell+1}^{k} \bar{\delta_j}  H\left(S_n^{\ell+1}|S_n^{[\ell]}\right)  \nonumber \\
	& \utag{f}{\leq}  \frac{\Gamma_{k,d,\ell} }{d} \alpha -  \left(\frac{\Gamma_{k,d,\ell} }{d}- \frac{d+1-\ell}{\ell} \bar{\mu}_\ell \right) H (S_n^{[\ell]})    +\frac{1}{\ell}\left( \sum_{j=\ell+1}^{k} \bar{\delta_j} \right) H\left(S_n^{[\ell]}\right) \nonumber \\ 
	& = \frac{\Gamma_{k,d,\ell} }{d} \alpha -  \left(\frac{\Gamma_{k,d,\ell} }{d}- \frac{d+1-\ell}{\ell} \bar{\mu}_\ell- \frac{1}{\ell}\left( \sum_{j=\ell+1}^{k} \bar{\delta_j} \right) \right) H (S_n^{[\ell]}), \label{eq:k_less_d_new_temple3} 
\end{align}
where \uref{e} follows from \eqref{eq:k_less_d_Condition3} and \uref{f} follows from Han's inequality. 

\begin{proposition}
\label{proposition:k_less_d:approach:final}
	$\frac{\Gamma_{k,d,\ell} }{d}- \frac{d+1-\ell}{\ell} \bar{\mu}_\ell- \frac{1}{\ell}\left( \sum_{j=\ell+1}^{k} \bar{\delta_j} \right) = 0$.
\end{proposition}
\begin{proof}
See Appendix~\ref{Appendix:k_less_d:approach:final}.
\end{proof}
We can see easily  that $\frac{1}{m}f(Q)$ is upper bounded by $\frac{\Gamma_{k,d,\ell} }{d} \alpha$ from \eqref{eq:k_less_d_new_temple3} and Proposition~\ref{proposition:k_less_d:approach:final}. Therefore, we have shown that for any $(k,d,\ell)$, if there exists a matrix $Q$ such that $f(Q)$ satisfies \eqref{eq:k_less_d_Condition1}, \eqref{eq:k_less_d_Condition2} and \eqref{eq:k_less_d_Condition3}, then $(k,d,\ell) \in \cP$.

\subsection{Main results}
\label{Sec:MainResults}
The following theorem gives the main result of this section.
\begin{theorem}
\label{thm:k_less_d}
	 The triple $(k, d, \ell) \in \cP$ if $\ell = k-1$ or 
	 \begin{equation}
	  \label{eq:k_less_d_threshold}
	  \begin{cases}
	  	d(d-\ell-1) - \frac{1}{2}(2d-k-\ell+1)(2d+k-3\ell-5) \geq 0, & \ell \leq k-4, \\
	  	k \geq \frac{1}{3}(d+8), & \ell= k - 3, \\
	  	k \geq \frac{1}{4}(d+7), & \ell= k - 2.			
	  \end{cases}
	 \end{equation} 
\end{theorem}
Before proving Theorem \ref{thm:k_less_d}, we first discuss some consequences of the theorem.

\begin{enumerate}
\item 
Let 
\[\cP_{s}:= \left\{(k,d,\ell): \ell=k-1 ~ \text{or} ~ \eqref{eq:k_less_d_threshold} ~ \text{is satisfied}\right\},\]
and for fixed $k$ and $d$ define 
\begin{equation}
\label{eq:k_less_d:def_hat_ell}
	\hat{\ell} :=  \min\left\{\ell \geq 1: (k,d,\ell) \in \cP_{s} \right\}.
\end{equation}
Note that $\hat{\ell}$ is well defined since $(k,d,\ell=k-1) \in \cP_{s}$ for any given $k$ and $d$.
Then, we claim that for fixed $k$ and $d$, $(k,d,\ell) \in \cP_{s}$ for $\ell \geq \hat{\ell}$. Clearly, to prove the claim, it is sufficient to show that if $(k,d,\ell) \in \cP_{s}$, then $(k,d,\ell+1) \in \cP_{s}$ for $\ell \leq k-2$.
Since the case $\ell=k-2$ is trivial, we consider $\ell \leq k-3$. By inspecting \eqref{eq:k_less_d_threshold}, we can easily see that if $(k,d, \ell = k-3) \in \cP_{s}$, then $(k,d, \ell = k-2) \in \cP_{s}$. Also, if $(k,d,\ell = k-4) \in \cP_{s}$, then the condition in the first line of \eqref{eq:k_less_d_threshold} is satisfied, which can be rewritten as
\[(2k-d-6)(d-k+3)+\frac{1}{2} \geq 0.\]
Since $k$ and $d$ are integers, we have $(2k-d-6)(d-k+3) \geq 0$, and hence
$k \geq \frac{1}{2}(d+6) \geq \frac{1}{3}(d+8)$, which implies that $(k,d,\ell = k-3) \in \cP_{s}$. Thus, it remains to show that if $(k,d,\ell) \in \cP_{s}$ for $\ell < k-4$, then $(k,d,\ell+1) \in \cP_{s}$. 

\hspace{1em} Towards this end, let
\begin{equation}
\label{eq:main_resulst:g}
	g(\ell) = d(d-\ell-1) - \frac{1}{2}(2d-k-\ell+1) (2d+k-3\ell-5), ~ 1 \leq \ell \leq k-4.
\end{equation}
Clearly $(k,d,\ell) \in \cP_{s}$ for $\ell \leq k-4$ if and only if $g(\ell) \geq 0$. Then we need to show that if $g(\ell) \geq 0$ for some $\ell < k-4$, then $g(\ell+1) \geq 0$.
For the quadratic equation $g(\ell) = 0$, the discriminant is $3(d-k)^2 + 12(d-4)+(k-8)^2$, which is nonnegative provided that $d \geq 4$. This condition is guaranteed because we have
\[d \geq k \geq \ell+4 \geq 5,\]
where the second inequality follows from the range of $\ell$ in \eqref{eq:main_resulst:g}. Thus the two roots of $g(\ell)=0$ are real and they are given by
\[\ell_1 = \frac{1}{3}\left(3d-k-1 - \sqrt{3(d-k)^2+12(d-4)+(k-8)^2} \right),\]
and
\[\ell_2 = \frac{1}{3}\left(3d-k-1 + \sqrt{3(d-k)^2+12(d-4)+(k-8)^2} \right).\]
Since the leading coefficient of $g(\ell)$ is negative, we see that $g(\ell) \geq 0$ if and only if $\ell_1 \leq \ell \leq \ell_2$.
Consider
\[\ell_2 \geq \frac{1}{3}\left(3d-k-1 + |k-8| \right) \geq \frac{1}{3}\left(3d-k-1 + k-8 \right) = d-3 、\geq k-3.\]
Then, if $g(\ell) \geq 0$ for some $\ell < k-4$, we have
\[\ell+1 < k-3 \leq \ell_2,\]
which implies that $g(\ell+1) \geq 0$, as is to be proved.
Thus we have shown that for fixed $k$ and $d$, $(k,d,\ell) \in \cP_{s}$ for $\ell \geq \hat{\ell}$. Since it is clear that $\cP_{s} \subseteq \cP$, we conclude that for fixed $k$ and $d$, $(k,d,\ell) \in \cP$ for $\ell \geq \hat{\ell}$.

\item
We claim that Theorem~\ref{thm:k_less_d} improves the existing result in Shao \textit{et~al.} \cite{Shao}, where they showed that $(k,d,\ell) \in \cP$ if 
\begin{equation}
\label{eq:Shao}
	\ell \geq \ell^{\star}:= \min\left\{\ell' \geq 1: \Gamma_{k,d,\ell'} \leq d+\sqrt{d\ell'}\right\}.
\end{equation}
Let
\begin{equation*}
	\cP_{r}:= \left\{(k,d,\ell): \ell \geq \ell^{\star} \right\}.
\end{equation*}
Recall that 
\[\Gamma_{k,d,\ell} =\sum_{i=\ell}^{k-1}(d-i) =\frac{1}{2}(k-\ell)(2d-k-\ell+1).\]
Evidently, $\Gamma_{k,d,\ell}$ is decreasing with $\ell$ while $d+\sqrt{d\ell}$ is increasing with $\ell$, and so we have
\begin{equation}
\label{eq:main_result:Pr}
	\cP_{r} = \left\{(k,d,\ell): \ell \geq \ell^{\star} \right\} = \left\{(k,d,\ell): \Gamma_{k,d,\ell} \leq d+\sqrt{d\ell} \right\}.
\end{equation}

We will justify our claim by first showing that $\cP_{r} \subseteq \cP_{s}$, or equivalently $\hat{\ell} \leq \ell^{\star}$.
For fixed $k$ and $d$, assume that $(k,d,\ell) \in \cP_{r}$, and we will prove that $(k,d,\ell) \in \cP_{s}$. It is trivial for the case $\ell = k-1$ because $(k,d,\ell) \in \cP_s$ always holds by definition. If $(k,d,\ell=k-2) \in \cP_{r}$, we can obtain from \eqref{eq:main_result:Pr} that 
\[k \geq \frac{1}{8}\left(5d - \sqrt{d(9d-8)} +12  \right) > \frac{1}{8}\left(5d - 3d +12  \right) = \frac{1}{4}(d+6).\]
Since $k$ and $d$ are integers, we must have $k \geq \frac{1}{4}(d+7)$, and hence $(k,d,\ell=k-2) \in \cP_{s}$. Similarly, if $(k,d,\ell=k-3) \in \cP_{r}$, we can obtain that
\[k \geq \frac{1}{18}\left(13d - \sqrt{d(25d-36)} +36  \right) > \frac{1}{18}\left(13d - 5d +36  \right) = \frac{1}{9}(4d +18).\]
Since $1 \leq \ell = k-3 \leq d -3$, we know that $d \geq 4$. If $d =4$, we have $k > \frac{34}{9}$. Since $k$ must be an integer, we have $k \geq 4 = \frac{1}{3}(d+8)$. If $d \geq 5$, since $k$ and $d$ are integers, we must have $k \geq \frac{1}{9}(4d+19) = \frac{1}{3}(d+8) + \frac{1}{9}(d-5) \geq \frac{1}{3}(d+8)$. Then we know that $(k,d,\ell=k-3) \in \cP_{s}$.

\hspace{1em} It remains to show that if $(k,d,\ell) \in \cP_{r}$ for $\ell \leq k-4$, then $(k,d,\ell) \in \cP_{s}$. 
Let
\[h(\ell) = d+\sqrt{d\ell} - \Gamma_{k,d,\ell}, ~ 1 \leq \ell \leq k-4.\]
Then $(k,d,\ell) \in \cP_{r}$ for some $\ell \leq k-4$ if and only if $h(\ell) \geq 0$ for some $\ell \leq k-4$. 
We claim that if $h(\ell) \geq 0$ for some $\ell \leq k-4$, then $\ell \geq \ell_0$, where $\ell_0 = \left\lceil \frac{1}{2}(d-1)\right\rceil$. This can be substantiated by contradiction. Assume the contrary that  $\ell \leq \ell_0 - 1$. Then we have
\begin{align}
	h\left(\ell \right)
	& = d+\sqrt{d \ell} - \Gamma_{k,d,\ell} \nonumber \\
	& = d+\sqrt{d \ell} - \frac{1}{2}(k-\ell)(2d-k-\ell+1) \nonumber \\
	& \utag{a}{\leq} d+\sqrt{d\ell} - 2(2d-2\ell-3) \nonumber\\
	& \utag{b}{\leq} d+\sqrt{d(\ell_0 - 1)} - 2(2d-2(\ell_0 - 1)-3) \nonumber\\
	& = d + \sqrt{ d \left(\left\lceil \frac{1}{2}(d-1)\right\rceil -1 \right)} - 2\left(2d-2\left\lceil \frac{1}{2}(d-1)\right\rceil -1\right)  \nonumber\\
	& \utag{c}{\leq} d + \sqrt{ d \left( \frac{1}{2}d -1 \right)} - 2\left(2d-d -1\right) \nonumber\\
	& = d + \sqrt{  \frac{1}{2}d(d-2)} - 2\left(d-1\right) \nonumber\\
	& < d + \sqrt{  \frac{1}{2}(d-1)^2} - 2\left(d-1\right) \nonumber\\
	& = \frac{\sqrt{2}-2}{2} (d-1) +1 \nonumber\\
	& \utag{d}{<} 0, \label{eq:k_less_d:discuss_h}
\end{align}
where \uref{a} follows from $k \geq \ell+4$ and $\frac{1}{2}(k-\ell)(2d-k-\ell+1)$ is increasing with $k$ when $k \leq d$; \uref{b} follows from the assumption $\ell \leq \ell_0 - 1$; \uref{c} follows from $\left\lceil \frac{1}{2}(d-1)\right\rceil \leq  \frac{1}{2}d$, and \uref{d} follows from $1 \leq \ell \leq k-4 \leq d-4$. Clearly, \eqref{eq:k_less_d:discuss_h} contradicts with the assumption that $h\left(\ell \right) \geq 0$, and hence we know that
if $h(\ell) \geq 0$ for some $\ell \leq k-4$, then $\ell \geq \ell_0$. 
Next, we will show that $g(\ell) \geq 0$ for $\ell_0 \leq \ell \leq k-4$.
Consider
\begin{align*}
	g(\ell) 
	& = d(d-\ell-1)-\frac{1}{2}(2d-k-\ell+1)(2d+k-3\ell-5) \\
	& \utag{e}{\geq} d(d-\ell-1)-\frac{1}{2}(2d-2\ell-3)(2d-2\ell-1)  \\
	& = d(d-\ell-1)-2(d-\ell-1)^2+\frac{1}{2} \\
	& = (d-\ell-1)(2\ell-d+2)+\frac{1}{2} \\
	& \utag{f}{\geq} (d-\ell-1)(2\ell_0-d+2)+\frac{1}{2} \\
	& \utag{g}{\geq} (d-\ell-1)+\frac{1}{2} \\
	& \geq 0, 
\end{align*}
where \uref{e} follows because $\frac{1}{2}(2d-k-\ell+1)(2d+k-3\ell-5)$ is decreasing with $k$ when $k \geq \ell+4$, \uref{f} follows from $\ell \geq \ell_0$, and \uref{g} follows from $2\ell_0-d+2 \geq 1$. 
Hence, we have shown that $g(\ell) \geq 0$ if  $h(\ell) \geq 0$ for some $\ell \leq k-4$. Therefore we can conclude that $\cP_{r} \subseteq \cP_{s}$, or $\hat{\ell} \leq \ell^{\star}$.

\hspace{1em} Finally, to see that there is a gap between $\ell^{\star}$ and $\hat{\ell}$, let us consider the example $d = 32$ and $k = 31$. For this case, we can easily check that the first case in \eqref{eq:k_less_d_threshold} is satisfied for $\ell = 12$, but is not satisfied for $\ell = 11$. Thus we obtain $\hat{\ell} = 12$. Also, by substituting $k=31$ and $d=32$ in \eqref{eq:Shao}, we have
\[\ell^{\star} = \min\left\{\ell \geq 1: \frac{1}{2}(31-\ell)(34-\ell) \leq 32+\sqrt{32\ell}\right\}.\]
Since the condition $\frac{1}{2}(31-\ell)(34-\ell) \leq 32+\sqrt{32\ell}$ is satisfied for $\ell =22$ but not for $\ell =21$, we obtain $\ell^{\star} = 22$. Therefore, there is a gap between $\hat{\ell}$ and $\ell^{\star}$, and this gap can be large.

\end{enumerate}


\subsection{Proof of Theorem \ref{thm:k_less_d}}
\label{subsec:thm1}
From the previous discussion, we know that for any $(k,d,\ell)$, if there exists a matrix $Q$ such that $f(Q)$ satisfies \eqref{eq:k_less_d_Condition1}, \eqref{eq:k_less_d_Condition2} and \eqref{eq:k_less_d_Condition3}, then $(k,d,\ell) \in \cP$. In this subsection, we will show the existence of a qualified matrix $Q$ for each $(k,d,\ell) \in \cP_s$. In particular, we consider $Q$ satisfying the following conditions:
\begin{enumerate}
\item If $q_{x,y}=0$, then $q_{x',y}=0$ for all $x' \leq x$;
\item If $q_{x,y}=0$, then $q_{x,y'}=0$ for all $y' \leq y$.
\end{enumerate}
These conditions say that the zeros and ones in the matrix $Q$ exhibit an echelon form, as depicted in Fig.~\ref{pic2}. 

\begin{figure}[htbp]
\centering
\begin{tikzpicture}[scale=0.5,yscale=1]
\def \parashort {0.5};
\def \slack {0.3};
\def \a {2};
\def \b {4};
\def \c {1};
\def \d {6};
\def \e {2};
\def \f {4};
\def \g {1};
\def \h {1};
\def \i {3};
\def \j {8};
\def \k {2};
\def \l {2};
\def \m {2};
\def \paralong  {13.5};
\coordinate (A) at       (\slack  ,  \slack+\a);
\coordinate (B) at ($(A) + (\b ,  0 )$);
\coordinate (C) at ($(B) + (0  ,  \c)$);
\coordinate (D) at ($(C) + (\d ,  0 )$);
\coordinate (E) at ($(D) + (0  ,  \e)$);
\coordinate (F) at ($(E) + (\f ,  0)$);
\coordinate (G) at ($(F) + (0  ,  \g)$);
\coordinate (H) at ($(G) + (\h ,  0)$);
\coordinate (I) at ($(H) + (0  ,  \i)$);
\coordinate (J) at ($(I) + (\j ,  0)$);
\coordinate (K) at ($(J) + (0  ,  \k)$);
\coordinate (L) at ($(K) + (\l ,  0)$);
\coordinate (M) at ($(L) + (0 ,  \m)$);
\gettikzxy{(L)}{\rightparax}{\ay}
\gettikzxy{(M)}{\ax}{\righttop}
\draw[line width=0.5mm] (\parashort , \paralong) -- (0, \paralong) -- (0,0) -- (\parashort, 0);
\draw[line width=0.5mm, xshift=\slack cm] ($(\rightparax, \paralong)-(\parashort,0)$) -- (\rightparax, \paralong) -- (\rightparax,0) -- ($(\rightparax, 0)-(\parashort,0)$);

\draw[dashed, line width=0.5mm] (A) -- (B) -- (C) -- (D) -- (E) -- (F) -- (G) -- (H) -- (I) -- (J) -- (K) -- (L);
\path[fill=black!20!white]  (\slack, \righttop) -- (A) -- (B) -- (C) -- (D) -- (E) -- (F) -- (G) -- (H) -- (I) -- (J) -- (K) -- (L) -- (\rightparax, \righttop) -- cycle;

\node[scale=3] (0) at (6,10) {$0$};
\node[scale=3] (1) at (21,4)  {$1$};
\end{tikzpicture}
\caption{An illustration of $Q$ in the proof of Theorem.~\ref{thm:k_less_d}.}
\label{pic2}
\end{figure}
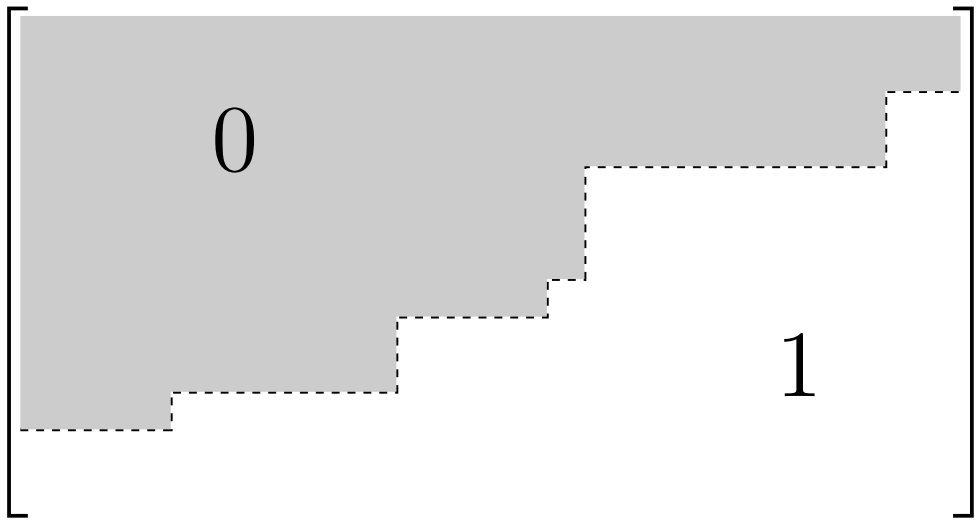

Any matrix illustrated in Fig.~\ref{pic2} can be uniquely represented by a set of rational numbers $\left\{z_j:j=\ell+1,\ldots,k\right\}$ such that $0 \leq z_j \leq 1$ and $z_i \leq z_j$ if $i \geq j$, where $m z_j$ corresponds to the number of zeros in the $j$-th column.

Now, for any $(k,d,\ell)$, let
\begin{equation}
\label{eq:k_less_d_zj}
	z_j=
	\begin{cases}
		\min\left\{\frac{\Gamma_{k,d,\ell}}{d}, \frac{2d-k-\ell+1}{d},1\right\},    & j=\ell+1,  \\
		\frac{2d-k-\ell+1}{2d},                   & j=\ell+2,\ldots,k-1, \\
		\max\left\{0, \frac{d-k-\ell+1}{d}\right\},                      & j=k ~ \text{and}~ \ell+1<k. 
	\end{cases}
\end{equation}
Note that when $\ell=k-1$, we have
\begin{equation}
\label{eq:Sec_less:ell_k1}
	z_{\ell+1} = z_k = \min\left\{\frac{\Gamma_{k,d,\ell}}{d}, \frac{2d-k-\ell+1}{d},1\right\} = \frac{\Gamma_{k,d,\ell}}{d}.
\end{equation}
It is easy to see that $0 \leq z_j \leq 1$ for all $j$, so we only need to verify that $z_i \leq z_j$ if $i \geq j$. Obviously, we only need to consider $\ell \leq k-2$. Then we have
\[\Gamma_{k,d,\ell} = \frac{1}{2}(k-\ell)(2d-k-\ell+1) \geq 2d-k-\ell+1.\]
Next, let us discuss the two cases $2d-k-\ell+1 \leq d$ and $2d-k-\ell+1 > d$ as follows. 
\begin{enumerate}
	\item If $2d-k-\ell+1 \leq d$, \eqref{eq:k_less_d_zj} can be written as
\begin{equation}
\label{eq:S3:leq_d}
	z_j=
	\begin{cases}
		\frac{2d-k-\ell+1}{d},    & j=\ell+1,  \\
		\frac{2d-k-\ell+1}{2d},                   & j=\ell+2,\ldots,k-1, \\
		0,                      & j=k.
	\end{cases}
\end{equation}
Since 
\[\frac{2d-k-\ell+1}{d} \geq \frac{2d-k-\ell+1}{2d} \geq 0,\]
we see that
$z_i \leq z_j$ if $i \geq j$.
\item If $2d-k-\ell+1 > d$, \eqref{eq:k_less_d_zj} can be written as
		\begin{equation}
		\label{eq:S3:geq_d}
			z_j=
			\begin{cases}
				1,    & j=\ell+1,  \\
				\frac{2d-k-\ell+1}{2d},                   & j=\ell+2,\ldots,k-1, \\
				\frac{d-k-\ell+1}{d},                      & j=k. 
			\end{cases}
		\end{equation}
Since
\[1 \geq \frac{2d-k-\ell+1}{2d} \geq \frac{d-k-\ell+1}{d},\]
we see that
$z_i \leq z_j$ if $i \geq j$.
\end{enumerate}
Therefore, the matrix specified by $z_j$ defined in \eqref{eq:k_less_d_zj} corresponds to the form depicted in Fig.~\ref{pic2}.

In the remaining of this subsection, we will verify that for any $(k,d,\ell) \in \cP_s$, $f(Q)$ satisfies the conditions \eqref{eq:k_less_d_Condition1}, \eqref{eq:k_less_d_Condition2} and \eqref{eq:k_less_d_Condition3}, where $Q$ is determined by \eqref{eq:k_less_d_zj}.
\begin{figure}[htbp]
\centering
\begin{tikzpicture}[scale=0.5,yscale=1]
\def \parashort {0.5};
\def \slack {0.3};
\def \a {2};
\def \b {4};
\def \c {1};
\def \d {6};
\def \e {2};
\def \f {4};
\def \g {1};
\def \h {1};
\def \i {3};
\def \j {8};
\def \k {2};
\def \l {2};
\def \m {2};
\def \paralong  {13.5};
\coordinate (A) at       (\slack  ,  \slack+\a);
\coordinate (B) at ($(A) + (\b ,  0 )$);
\coordinate (C) at ($(B) + (0  ,  \c)$);
\coordinate (D) at ($(C) + (\d ,  0 )$);
\coordinate (E) at ($(D) + (0  ,  \e)$);
\coordinate (F) at ($(E) + (\f ,  0)$);
\coordinate (G) at ($(F) + (0  ,  \g)$);
\coordinate (H) at ($(G) + (\h ,  0)$);
\coordinate (I) at ($(H) + (0  ,  \i)$);
\coordinate (J) at ($(I) + (\j ,  0)$);
\coordinate (K) at ($(J) + (0  ,  \k)$);
\coordinate (L) at ($(K) + (\l ,  0)$);
\coordinate (M) at ($(L) + (0 ,  \m)$);
\gettikzxy{(L)}{\rightparax}{\ay}
\gettikzxy{(M)}{\ax}{\righttop}
\draw[line width=0.5mm] (\parashort , \paralong) -- (0, \paralong) -- (0,0) -- (\parashort, 0);
\draw[line width=0.5mm, xshift=\slack cm] ($(\rightparax, \paralong)-(\parashort,0)$) -- (\rightparax, \paralong) -- (\rightparax,0) -- ($(\rightparax, 0)-(\parashort,0)$);

\path[fill=black!20!white]  (\slack, \righttop) -- (A) -- (B) -- (C) -- (D) -- (E) -- (F) -- (G) -- (H) -- (I) -- (J) -- (K) -- (L) -- (\rightparax, \righttop) -- cycle;

\path[pattern=dots, pattern color=blue] (\slack, \slack) -- (A) -- (\rightparax, \slack+\a) -- (\rightparax, \slack) -- cycle;

\gettikzxy{(B)}{\Bx}{\By}
\gettikzxy{(D)}{\Dx}{\Dy}
\gettikzxy{(E)}{\Ex}{\Ey}
\gettikzxy{(F)}{\Fx}{\Fy}
\gettikzxy{(H)}{\Hx}{\Hy}
\gettikzxy{(I)}{\Ix}{\Iy}

\draw[line width=0.7mm,color=black] ($(I) + (2 , 0)$) -- ($(I) + (2 , \By-\Iy)$);
\draw[line width=0.7mm,color=black] ($(I) + (1 , 0)$) -- ($(I) + (1 , \By-\Iy)$);
\draw[line width=0.7mm,color=black] ($(I) + (1 , \By-\Iy)$) -- ($(I) + (2 , \By-\Iy)$);
\draw[line width=0.7mm,color=black] ($(I) + (1 , 0)$) -- ($(I) + (2 , 0)$);
\path[pattern=north west lines, pattern color=red] ($(I) + (1 , 0)$) -- ($(I) + (2 , 0)$) -- ($(I) + (2 , \By-\Iy)$) -- ($(I) + (1 , \By-\Iy)$) -- cycle;

\draw[line width=0.5mm,color=black] ($(H) + (1 , 0)$) -- ($(H) + (2 , 0)$);
\draw[line width=0.5mm,color=black] ($(F) + (2 , 0)$) -- ($(F) + (3 , 0)$);
\draw[line width=0.5mm,color=black] ($(\Ix,\Dy) + (1 , 0)$) -- ($(\Ix,\Dy) + (2 , 0)$);

\path[pattern=crosshatch, pattern color=red]  ($(\Ix,\Dy) + (1 , 0)$) -- ($(F) + (2 , 0)$) -- ($(F) + (3 , 0)$) -- ($(\Ix,\Dy) + (2 , 0)$) -- cycle;

\draw[dashdotted, line width=0.1mm] (E) -- ($(F) + (2 , 0)$);
\draw[dashdotted, line width=0.1mm] ($(E) + (0 , \Dy-\Ey)$) -- ($(\Ix,\Dy) + (1 , 0)$);

\draw[loosely dashed, line width=0.3mm] ($(\Ix,\slack) + (1 , 0)$) -- ($(\Ix,\slack+\a) + (1 , 0)$);
\draw[loosely dashed, line width=0.3mm] ($(I) + (1 , 0)$) -- ($(\Ix,\righttop) + (1 , 0)$);
\draw[loosely dashed, line width=0.3mm] ($(\Ix,\slack) + (2 , 0)$) -- ($(\Ix,0) + (2 ,\slack+\a)$);
\draw[loosely dashed, line width=0.3mm] ($(I) + (2 , 0)$) -- ($(\Ix,0) + (2 ,\righttop)$);

\draw[loosely dashed, line width=0.3mm] ($(\Ex,\slack)+(-1,0)$) -- ($(E) + (-1 , \righttop-\Ey)$);
\draw[loosely dashed, line width=0.3mm] (\Ex,\slack) -- ($(E) + (0 , \righttop-\Ey)$);
\draw[loosely dashed, line width=0.3mm] ($(\Ex,\slack)+(1,0)$) -- ($(E) + (1 , \righttop-\Ey)$);

\node[scale=0.8] (i) at ($(\Ex,\righttop)+(-0.8,0.5)$) {$i$};
\node[scale=0.8] (i') at ($(\Ex,\righttop)+(0.8,0.5)$) {$i+1$};
\node[scale=0.8] (j) at ($(\Ix,\righttop)+(1.5,0.5)$) {$j$};

\node[scale=3] (AA) at (6,10) {$A$};
\node[scale=3] (BB) at (6,1.25) {$B$};
\node[scale=3] (CC) at (22,6) {$C$};
\end{tikzpicture}
\caption{Illustration of three regions $A$, $B$ and $C$ of the matrix $Q$.}
\label{pic3}
\end{figure}
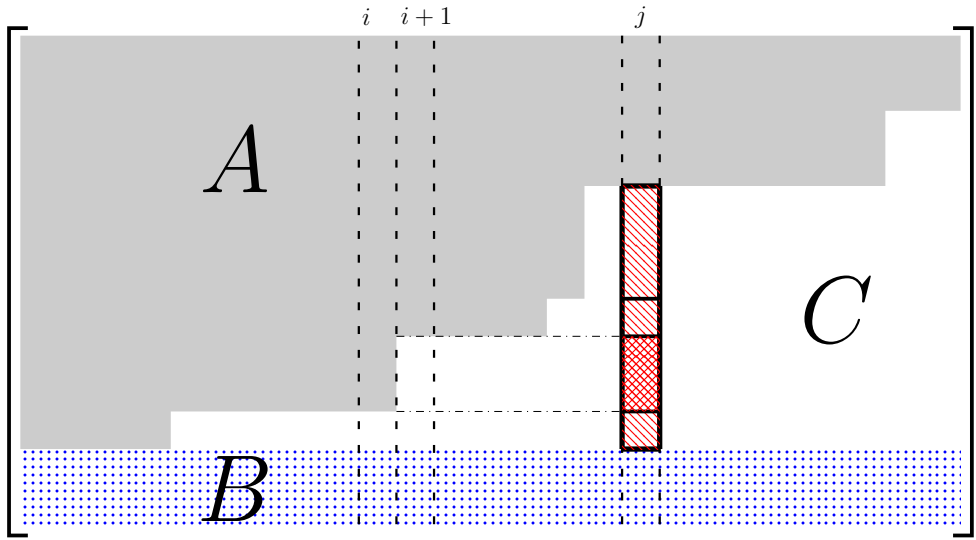
First, we need to write $f(Q)$ explicitly. To do this, we divide the matrix $Q$ into three regions, namely $A$, $B$ and $C$ as illustrated in Fig.~\ref{pic3}.

\begin{enumerate}
\item For the shaded gray region $A=\left\{q_{x,y}: x \leq m z_{y},\ell+1 \leq y \leq k\right\}$, we can easily see that $q_{x,y}=0$ and $t_{x,y}=\ell$. Then by checking the conditions in \eqref{eq:k_less_d_W_bound}, we see that only the elements in the first column, \textit{i.e.} $y=\ell+1$, belong to the second case, while all others belong to the first case. Hence, the total contribution of the region $A$ to $f(Q)$ is given by
	\begin{equation}
	\label{eq:k_less_d_zero}
		 m z_{\ell+1} \left( \alpha - H(S_n^{[\ell]}) \right) + \sum_{j =\ell+2}^k m z_j \left(\alpha-H(S_n^{[j-2]})-H(S^{j-1}|S^{[j-2]})+\frac{d+1-j}{d+1-\ell}H(S^{\ell}|S^{[\ell-1]} ) \right).
	\end{equation}
\item For the dotted area $B=\left\{ q_{x,y}: x > m z_{\ell+1}, \ell+1 \leq y \leq k \right\}$, we can easily see that $q_{x,y}=1$ and $t_{x,y}=y-1$. Hence, we can obtain from \eqref{eq:k_less_d_S_bound} that the total contribution of the region $B$ to $f(Q)$ is given by
	\begin{equation}
	\label{eq:k_less_d_all_one}
		 m (1 - z_{\ell + 1}) \sum_{j = \ell + 1}^k H \left(S^j | S^{[j - 1]} \right).
	\end{equation}

\item For the remaining region $C = \left\{q_{x,y}: m z_{y}< x \leq m z_{\ell+1},\ell+1 \leq y \leq k\right\}$, we consider its contribution to $f(Q)$ column by column. For the column $j$, let $C_j: = \left\{q_{x,j}: m z_{j}< x \leq m z_{\ell+1} \right\}$, which is illustrated as the vertical stripe in Fig.~\ref{pic3}. We further divide $C_j$ into $j-\ell-1$ segments.
Let $C_{j}^i:= \left\{q_{x,j}: m z_{i+1}< x \leq m z_{i} \right\}$ for $i = \ell+1,\ldots,j-1$, where 
\[\bigcup_{i=\ell+1,\ldots,j-1} C_j^i = C_{j}.\]
Note that for a fixed $j$, $C_j^i$ may be empty for some $i$.
Focus on a non-empty $C_j^i$, which is illustrated as the crosshatched segment in Fig.~\ref{pic3}. Then we have $q_{x,y} = 1$ and $t_{x,y} = \ell+j-i-1$. By invoking \eqref{eq:k_less_d_S_bound}, we obtain that the contribution of $C_j^i$ to $f(Q)$ is 
	\begin{equation*}
	    m ( z_i- z_{i + 1})\frac{d+1-j}{d+1-(\ell+j-i)} H(S^{\ell+j-i}|S^{[\ell+j-i-1]}).
	\end{equation*}
It follows that the contribution of $C_j$ to $f(Q)$ is given by 
	\begin{equation*}
	    \sum_{i=\ell+1}^{j-1}m ( z_i- z_{i + 1})\frac{d+1-j}{d+1-(\ell+j-i)} H(S^{\ell+j-i}|S^{[\ell+j-i-1]}).
	\end{equation*} 

Finally, the total contribution of the region $C$ to $f(Q)$ is given by	
	\begin{equation}
	\label{eq:k_less_d_one_zero_temple1}
  		\sum_{j=\ell+1}^k \sum_{i=\ell+1}^{j-1}m (z_i- z_{i + 1})\frac{d+1-j}{d+1-(\ell+j-i)} H(S^{\ell+j-i}|S^{[\ell+j-i-1]}).
	\end{equation}
\end{enumerate}
\vspace{1em}
For the ease of notation in the remaining parts, let us first simplify \eqref{eq:k_less_d_one_zero_temple1}. Consider
\begin{align}
	&  \sum_{j=\ell+1}^k \sum_{i=\ell+1}^{j-1} (z_i-z_{i + 1})\frac{d+1-j}{d+1-(\ell+j-i)} H(S^{\ell+j-i}|S^{[\ell+j-i-1]}) \nonumber\\
	& =  \sum_{j =\ell+1}^k \sum_{p=\ell+1}^{j-1}  (z_{\ell+j-p} - z_{\ell+j-p+1}) \frac{d+1-j}{d+1-p} H(S^p|S^{[p-1]})  \hspace{2em}(p=\ell+j-i)\nonumber\\
	& = \sum_{p =\ell+1}^{k-1} \sum_{j=p+1}^{k}  (z_{\ell+j-p} - z_{\ell+j-p+1}) \frac{d+1-j}{d+1-p} H(S^p|S^{[p-1]}) \nonumber\\
	& \utag{a}{=}  \sum_{j =\ell+1}^{k-1} \left( \sum_{i=j+1}^{k} (z_{\ell+i-j} - z_{\ell+i-j+1})(d+1-i)\right) \frac{1}{d+1-j} H(S^j|S^{[j-1]}), \nonumber
\end{align}
where \uref{a} follows from replacing the indices $p$ and $j$ by $j$ and $i$, respectively.

Let 
\[c_j  = \sum_{i=j+1}^k (z_{\ell+i-j} - z_{\ell+i-j+1}) (d+1-i),~ j=\ell+1,\ldots,k.\]
Then we have
\begin{align*}
	 & \sum_{j =\ell+1}^{k-1} \left( \sum_{i=j+1}^{k} (z_{\ell+i-j} - z_{\ell+i-j+1})(d+1-i)\right) \frac{1}{d+1-j} H(S^j|S^{[j-1]}) \\
	 & =  \sum_{j =\ell+1}^{k-1} \frac{c_j}{d+1-j} H(S^j|S^{[j-1]}) \\
	 & \utag{b}{=} \sum_{j =\ell+1}^{k} \frac{c_j}{d+1-j} H(S^j|S^{[j-1]}),
\end{align*}
where \uref{b} follows from $c_k=0$ by definition. Hence, \eqref{eq:k_less_d_one_zero_temple1} can be written as
\begin{equation}
\label{eq:k_less_d_one_zero_temple2}
	m \sum_{j =\ell+1}^{k} \frac{c_j}{d+1-j} H(S^j|S^{[j-1]}).
\end{equation}

Now, focus on $f(Q)$, which can be obtained by adding \eqref{eq:k_less_d_zero}, \eqref{eq:k_less_d_all_one}, and \eqref{eq:k_less_d_one_zero_temple2} as follows:
\begin{align*}
	f(Q) & =  m z_{\ell+1} \left( \alpha - H(S_n^{[\ell]}) \right) + \sum_{j =\ell+2}^k m z_j \left(\alpha-H(S_n^{[j-2]})-H(S^{j-1}|S^{[j-2]})+\frac{d+1-j}{d+1-\ell}H(S^{\ell}|S^{[\ell-1]} ) \right) \\
	& ~~ + m (1 - z_{\ell + 1}) \sum_{j = \ell + 1}^k H \left(S^j | S^{[j - 1]} \right) + m \sum_{j =\ell+1}^{k} \frac{c_j}{d+1-j} H \left(S^j | S^{[j - 1]} \right).
\end{align*}
By dividing $m$ on both sides, we have
\begin{align}
	\frac{1}{m}f(Q) & =   z_{\ell+1} \left( \alpha - H(S_n^{[\ell]}) \right) + \sum_{j =\ell+2}^k  z_j \left(\alpha-H(S_n^{[j-2]})-H(S^{j-1}|S^{[j-2]})+\frac{d+1-j}{d+1-\ell}H(S^{\ell}|S^{[\ell-1]} ) \right) \nonumber \\
	& ~~ +  (1 - z_{\ell + 1}) \sum_{j = \ell + 1}^k H \left(S^j | S^{[j - 1]} \right) +  \sum_{j =\ell+1}^{k} \frac{c_j}{d+1-j} H \left(S^j | S^{[j - 1]} \right)  \nonumber\\
	&   = \left(\sum_{j =\ell+1}^k z_j \right) \alpha - z_{\ell+1}  H(S_n^{[\ell]}) - \sum_{j =\ell+2}^k z_j H\left(S_n^{[j-2]}\right) 
	  -\sum_{j =\ell+2}^k z_j H\left(S^{j-1}|S^{[j-2]}\right)    \nonumber \\
	& ~~ + \frac{1}{d+1-\ell} \left(\sum_{j =\ell+2}^k z_j (d+1-j)\right) H\left(S^{\ell}|S^{[\ell-1]}\right)
	+ \sum_{j =\ell+1}^k \left(1 - z_{\ell + 1}  +\frac{c_j}{d+1-j}\right) H \left(S^j | S^{[j - 1]} \right) \nonumber \\
	& = \left(\sum_{j =\ell+1}^k z_j \right) \alpha  - z_{\ell+1}  H(S_n^{[\ell]})  - 
	\sum_{j =\ell}^{k-2} z_{j+2} H\left(S_n^{[j]}\right) 
	 - \sum_{j =\ell+1}^{k-1} z_{j+1} H\left(S^{j}|S^{[j-1]}\right)    \nonumber \\
	& ~~ + \frac{1}{d+1-\ell} \left(\sum_{j =\ell+2}^k z_j (d+1-j)\right) H\left(S^{\ell}|S^{[\ell-1]}\right)
	+ \sum_{j =\ell+1}^k \left(1 - z_{\ell + 1}  +\frac{c_j}{d+1-j}\right) H \left(S^j | S^{[j - 1]} \right). \label{eq:k_less_d_1127_temple1}    
\end{align}

For notational simplicity, let us separately discuss the case $\ell=k-1$. For $\ell=k-1$, \eqref{eq:k_less_d_1127_temple1} can be written as 
\begin{align*}
	\frac{1}{m} f(Q) & = z_{\ell+1} \alpha  - z_{\ell+1}  H(S_n^{[\ell]})   + \left(1 - z_{\ell + 1}  \right) H \left(S^{\ell+1} | S^{[\ell]} \right) \\
	&  \utag{c}{=} \frac{\Gamma_{k,d,k-1}}{d} \alpha  - \frac{\Gamma_{k,d,k-1}}{d}  H(S_n^{[\ell]})   + \left(1 - \frac{\Gamma_{k,d,k-1}}{d}  \right) H \left(S^{\ell+1} | S^{[\ell]} \right),
\end{align*}
where \uref{c} follows from \eqref{eq:Sec_less:ell_k1}.
By comparing the coefficients of $\alpha$, $H\left(S^{[j]}_n \right)$ and $H\left (S^{j}| S^{[j-1]} \right)$ with those in \eqref{eq:k_less_d_Qeff}, we have $\bar{\nu}_{\ell}= \frac{\Gamma_{k,d,k-1}}{d}$, $\bar{\mu}_{\ell} =0$ and $\bar{\mu}_{\ell+1} = 1 - \frac{\Gamma_{k,d,k-1}}{d}$, and we can easily check that these coefficients satisfy \eqref{eq:k_less_d_Condition1}, \eqref{eq:k_less_d_Condition2} and \eqref{eq:k_less_d_Condition3}, which implies that if $\ell=k-1$, then $(k,d,\ell) \in \cP$. This result has already been obtained in \cite{Ye} and \cite{Shao}, but the proof here is much shorter (if we confine our discussion to the case $\ell=k-1$).

For $\ell \leq k-2$, by collecting the terms in \eqref{eq:k_less_d_1127_temple1}, we obtain
\begin{equation}
\label{eq:k_less_d_Final_Form}
	\frac{1}{m} f(Q) = \left( \sum_{j=\ell}^{k-1} \bar{\nu}_j \right)\alpha -  \sum_{j=\ell}^{k-1} \bar{\nu}_j H\left(S_n^{[j]}\right) + \sum_{j=\ell}^{k} \bar{\mu}_j H\left(S^j|S^{[j-1]}\right),
\end{equation}
where 
\begin{equation}
\label{eq:mu}
\bar{\mu}_j = 
	\begin{cases}
	\frac{1}{d+1-\ell} \left(\sum_{j =\ell+2}^k z_j (d+1-j)\right),   & j=\ell, \\
	1 - z_{\ell + 1} -z_{j+1} +\frac{c_j}{d+1-j}  , & j=\ell+1,\ldots,k-1, \\
	1 - z_{\ell+1},          & j=k, 
\end{cases}
\end{equation}
and
\begin{equation}
\label{eq:v}
\bar{\nu}_j = 
	\begin{cases}
	z_{\ell+1}+z_{\ell+2},  & j=\ell, \\
	z_{j+2},   & j=\ell+1,\ldots,k-2, \\
	0,          & j=k-1.
\end{cases}
\end{equation}

Recall that we need to check the three conditions \eqref{eq:k_less_d_Condition1}, \eqref{eq:k_less_d_Condition2} and \eqref{eq:k_less_d_Condition3}. 
First, let us check that \eqref{eq:k_less_d_Condition1} is satisfied as follows. 
From \eqref{eq:S3:leq_d} and \eqref{eq:S3:geq_d}, we can see that in both cases, we have
\begin{equation}
\label{eq:S3:z_plus}
	z_{\ell+1} + z_k = \frac{2d-k-\ell+1}{d}.
\end{equation}
Hence we obtain
\[\sum_{j=\ell}^{k-1} \bar{\nu}_j  = \sum_{j=\ell+1}^{k} z_j =\frac{2d-k-\ell+1}{d} + \sum_{j=\ell+2}^{k-1} z_j \utag{d}{=} \frac{2d-k-\ell+1}{d} + (k-\ell-2) \frac{2d-k-\ell+1}{2d}  = \frac{\Gamma_{k,d,\ell}}{d},\]
where \uref{d} follows from \eqref{eq:k_less_d_zj}.

Now, let us verify the conditions \eqref{eq:k_less_d_Condition2} and \eqref{eq:k_less_d_Condition3}. From \eqref{eq:mu}, we see that $\bar{\mu}_{\ell} \geq 0$ and $\bar{\mu}_{k} \geq 0$. Since $\bar{\delta}_k = (d+1-k)\bar{\mu}_k$, we have $\bar{\delta}_{k} \geq 0$. Hence, it remains to show that $\bar{\mu}_{j} \geq 0$ and $\bar{\delta}_{j} \geq 0$ for $j = \ell+1,\ldots,k-1$. We know from \eqref{eq:v} that $\bar{\nu}_j \geq 0$ for all $j$, so we have
\[\bar{\delta}_j = (d+1-j)\bar{\mu}_j - \sum_{i= j}^{k-1} \bar{\nu}_i \leq (d+1-j)\bar{\mu}_j,\]
which implies that if $\bar{\delta}_j \geq 0$, then $\bar{\mu}_j \geq 0$. Therefore, it suffices to prove the following proposition.
\begin{proposition}
\label{Proposition:delta}
	For any $(k,d,\ell) \in \cP_s$, where $\ell \leq k-2$, $\bar{\delta}_j \geq 0$ for $j=\ell+1,\ldots,k-1$.
\end{proposition}
\begin{proof}
	See Appendix~\ref{Appendix:proposition:delta}.
\end{proof}


\section{Conclusion and discussion}
\label{Sec:conclusion}
In this paper, we study the problem of secure distributed storage systems where the eavesdropper has the capability to observe the data involved in the repair process. 
Our goal is to characterize parameters $(n,k,d,\ell)$ whose optimal storage-bandwidth tradeoff curve has one corner point and can be determined. Toward this end, we obtained a lower bound $\hat{\ell}$ on the number of wiretap nodes which is tight for $k=d=n-1$. Whether this bound is tight for other values of $n$, $k$ and $d$ is a problem for future research. Our results subsume all the previous related results \cite{Ye_isit,Tandon16,Ye,Shao}.

\appendices
\section{Proof of the optimality of $(\hat{\alpha},\hat{\beta})$}
\label{app:optimality}
We will prove that $(\hat{\alpha},\hat{\beta})$ is on the optimal tradeoff curve by establishing the following outer bound
\begin{equation}
\label{eq:Jan09_temple1}
	\begin{cases}
		\bar{\alpha} + \left(\Gamma_{k,d,\ell} - d\right) \bar{\beta} \geq 1, & \Gamma_{k,d,\ell} > d, \\
		\bar{\alpha} \geq \hat{\alpha}, & \Gamma_{k,d,\ell} \leq d.
	\end{cases}
\end{equation}
It has been shown that $\bar{\beta} \geq \hat{\beta}$ (cf.\eqref{eq:betabound}). As illustrated in Fig.~\ref{pic:opt_SRK}, we can see that the intersection of $\bar{\beta} \geq \hat{\beta}$ and \eqref{eq:Jan09_temple1} is given by $(\hat{\alpha},\hat{\beta})$, no matter whether $\Gamma_{k,d,\ell} > d$ or $\Gamma_{k,d,\ell} \leq d$. As $(\hat{\alpha},\hat{\beta})$ is achievable, we can conclude that $(\hat{\alpha},\hat{\beta})$ must be on the tradeoff curve if \eqref{eq:Jan09_temple1} holds. 
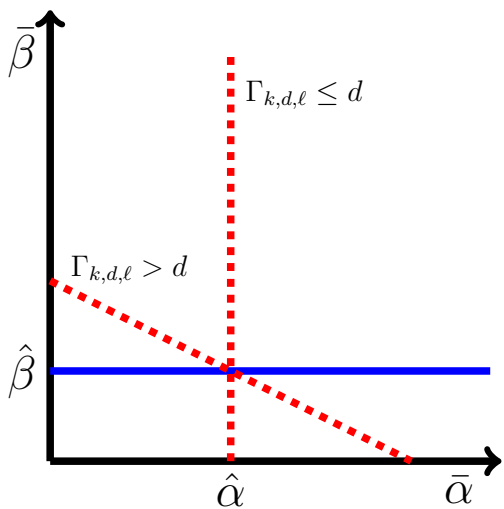
\begin{figure}[htbp]
\centering
\begin{tikzpicture}[scale=0.75,transform shape]
\coordinate (O) at (0,0);
\draw [->, line width=1mm] (O) -- (8,0);
\draw [->, line width=1mm] (O) -- (0,8);

\draw [line width=1mm, blue] (0,1.6) -- (7.8,1.6);
\draw [dashed, line width=1mm, red] (3.2,0) -- (3.2,7.2);
\draw [dashed, line width=1mm, red] (0,3.2) -- (6.4,0);

\node[scale=2] (1) at (7.25,-0.5)  {$\bar{\alpha}$};
\node[scale=2] (1) at (-0.5,7.25)  {$\bar{\beta}$};

\node[scale=2] (1) at (-0.5,1.6)  {$\hat{\beta}$};
\node[scale=2] (1) at (3.2,-0.5)  {$\hat{\alpha}$};

\node[scale=1.25] (1) at (1.4,3.4)  {$\Gamma_{k,d,\ell} > d$};
\node[scale=1.25] (1) at (4.5,6.5)  {$\Gamma_{k,d,\ell} \leq d$};

\end{tikzpicture}
\caption{Two outer bounds in \eqref{eq:Jan09_temple1}.}
\label{pic:opt_SRK}
\end{figure}

We now proceed to \eqref{eq:Jan09_temple1}. By letting $\cT=\left\{S^i:  i = 1, \ldots, k \right\}$ and $\cL=[\ell]$ in \eqref{eq:SC_Outer}, we have
\begin{align}
	B_s & \leq H\left(S^{[\ell+1:k]}|S^{[\ell]}\right) \nonumber \\
		& = \sum_{i=\ell+1}^{k} H\left(S^{i}|S^{[i-1]}\right) \nonumber \\
		& \utag{a}{\leq} \sum_{i=\ell+1}^{k} (d+1-i) H\left(S_n^{i}|S_n^{[i-1]}\right) \nonumber \\
		& \utag{b}{\leq} \sum_{i=\ell+1}^{k} \frac{d+1-i}{\ell} H\left(S_n^{[\ell]}\right) \nonumber \\
		& = \frac{\Gamma_{k,d,\ell}}{\ell} H\left(S_n^{[\ell]}\right), \label{eq:app:optimality_t1}
\end{align}
where \uref{a} follows from \eqref{eq:k_eq_d_bound_S} , and \uref{b} follows from Han's inequality.

Similarly, by letting $\cT=\left\{S^i:  i = 1, \ldots, k-1 \right\} \cup \left\{W_k\right\}$ and $\cL=[\ell]$ in \eqref{eq:SC_Outer}, we have 
\begin{align}
	B_s & \leq H\left(S^{[\ell+1:k-1]},W_k|S^{[\ell]}\right)  \nonumber \\
		& = \sum_{i=\ell+1}^{k-1} H\left(S^{i}|S^{[i-1]}\right) + H\left(W_{k}|S^{[k-1]}\right) \nonumber \\
		& \utag{c}{\leq} \sum_{i=\ell+1}^{k-1} (d+1-i) H\left(S_n^{i}|S_n^{[i-1]}\right) + H\left(W_{k}|S^{[k-1]}\right) \nonumber \\
		& \leq \sum_{i=\ell+1}^{k-1} (d+1-i) H\left(S_n^{i}|S_n^{[i-1]}\right) +  H\left(W_n|S_n^{[k-1]}\right)  \nonumber  \\
		& = \sum_{i=\ell+1}^{k-1} (d+1-i) H\left(S_n^{i}|S_n^{[i-1]}\right) + \alpha - H\left(S_n^{[k-1]}\right)  \nonumber  \\
		& = \alpha + \sum_{i=\ell+1}^{k-1} (d+1-i) H\left(S_n^{i}|S_n^{[i-1]}\right) - \left(H\left(S_n^{[\ell]}\right) + \sum_{i=\ell+1}^{k-1} H\left(S_n^{i}|S_n^{[i-1]}\right)   \right) \nonumber \\
		& = \alpha - H\left(S_n^{[\ell]}\right)  + \sum_{i=\ell+1}^{k-1} (d-i) H\left(S_n^{i}|S_n^{[i-1]}\right) \nonumber \\
		& \utag{d}{\leq} \alpha  - H\left(S_n^{[\ell]}\right) + \sum_{i=\ell+1}^{k-1}  \frac{d-i}{\ell} H\left(S_n^{[\ell]}\right)  \nonumber \\
		& =  \alpha + \frac{1}{\ell} \left(- \ell + \sum_{i=\ell+1}^{k-1} (d-i)  \right)  H\left(S_n^{[\ell]}\right)  \nonumber \\
		& = \alpha +  \frac{\Gamma_{k,d,\ell}-d}{\ell}    H\left(S_n^{[\ell]}\right), \label{eq:app:optimality_t2} 
\end{align}
where \uref{c} follows from \eqref{eq:k_eq_d_bound_S} and \uref{d} follows from Han's inequality.

When $\Gamma_{k,d,\ell} > d$, we know from \eqref{eq:app:optimality_t2} that 
\[B_s \leq \alpha +  \frac{\Gamma_{k,d,\ell}-d}{\ell}    H\left(S_n^{[\ell]}\right) \leq \alpha +  (\Gamma_{k,d,\ell}-d) \beta,\]
which is equivalent to 
\begin{equation*}
	\bar{\alpha} + (\Gamma_{k,d,\ell}-d) \bar{\beta} \geq 1, 
\end{equation*}
\textit{i.e.} the first case of \eqref{eq:Jan09_temple1}.

When $\Gamma_{k,d,\ell} \leq d$, by multiplying \eqref{eq:app:optimality_t1} and \eqref{eq:app:optimality_t2} by  $d-\Gamma_{k,d,\ell}$ and $\Gamma_{k,d,\ell}$ respectively, we have
\[d B_s \leq (d-\Gamma_{k,d,\ell}) \frac{\Gamma_{k,d,\ell}}{\ell} H\left(S_n^{[\ell]}\right) + \Gamma_{k,d,\ell} \left( \alpha +  \frac{\Gamma_{k,d,\ell}-d}{\ell}    H\left(S_n^{[\ell]}\right)\right) = \Gamma_{k,d,\ell} \alpha,\]
which is equivalent to $\bar{\alpha} \geq \hat{\alpha}$, \textit{i.e.} the second case of \eqref{eq:Jan09_temple1}.
Therefore, we have proved \eqref{eq:Jan09_temple1}.

\begin{remark}
Shao \textit{et~al.} \cite{Shao} have proved that $\bar{\alpha} \geq \hat{\alpha}$ if $\Gamma_{k,d,\ell} \leq d$, \textit{i.e.} the second case of \eqref{eq:Jan09_temple1}. The proof therein is very lengthy. The proof of \eqref{eq:Jan09_temple1} here is much simple, and the first case is a new result.
\end{remark}

\section{Proof of Lemma \ref{lem:k_eq_d}}
\label{App:lem_k_eq_d}
First, for $t=0,\ldots,k-2$, we consider
\begin{align}
H(W_{[t+1:k]}|S^{[t]})  & = \sum_{i=t+1}^k H\left(W_i|W_{[t+1:i-1]},S^{[t]} \right) \nonumber \\
						& = \sum_{i=t+1}^k H\left(W_i|S^{[t]}\right) - \sum_{i=t+2}^k I\left(W_i;W_{[t+1:i-1]}|S^{[t]}\right) \nonumber \\
						& \leq \sum_{i=t+1}^k H\left(W_i|S^{[t]}\right)- \sum_{i=t+2}^k I\left(S_i^{[t+1:i-1]};S_{[t+1:i-1]}^i|S^{[t]}\right), \label{eq:app_temple_1}
\end{align}
where \eqref{eq:app_temple_1} is justified because $S_i^{[t+1:i-1]}$ and $S_{[t+1:i-1]}^i$ are functions of $W_i$ and $W_{[t+1:i-1]}$ respectively.

The second term on the right-hand side of \eqref{eq:app_temple_1} can be further bounded as follows:
\begin{align}
 	& \sum_{i=t+2}^k I\left(S_i^{[t+1:i-1]};S_{[t+1:i-1]}^i|S^{[t]}\right)   \nonumber\\
 	& = \sum_{i=t+2}^k  H\left(S_{[t+1:i-1]}^i|S^{[t]}\right) + \sum_{i=t+2}^k H \left(S_i^{[t+1:i-1]}|S^{[t]} \right) - \sum_{i=t+2}^k  H \left(S_i^{[t+1:i-1]},S_{[t+1:i-1]}^i|S^{[t]}\right)  \nonumber\\
 	& \utag{a}{\geq}  \sum_{i=t+2}^k  \frac{i-1-t}{d-t} H\left(S_{\cN \backslash \{i\} \backslash [t]}^i|S^{[t]}\right)   + \sum_{i=t+2}^k H\left(S_i^{[t+1:i-1]}|S^{[t]}\right) - \sum_{i=t+2}^k  H \left(S_i^{[t+1:i-1]},S_{[t+1:i-1]}^i|S^{[t]}\right) \nonumber\\
 	& \utag{b}{\geq} \sum_{i=t+2}^k  \frac{i-1-t}{d-t} H\left(S_{\cN \backslash \{i\} \backslash [t]}^i|S^{[t]}\right) + \sum_{i=t+2}^k H\left(S_i^{[t+1:i-1]}|S^{[t]}\right) -  \sum_{i=t+2}^k  I \left(S^{[t+1:i-1]} ; S^i|S^{[t]}\right) \nonumber \\
 	& = \sum_{i=t+2}^k  \frac{i-1-t}{d-t} H\left(S_{\cN \backslash \{i\} \backslash [t]}^i|S^{[t]}\right) + \sum_{i=t+2}^k H\left(S_i^{[t+1:i-1]}|S^{[t]}\right) - \left( \sum_{i=t+2}^k  H \left(S^i|S^{[t]}\right) - \sum_{i=t+2}^k  H \left(S^i|S^{[i-1]}\right)  \right) \nonumber \\
 	& = \sum_{i=t+2}^k  \frac{i-1-t}{d-t} H\left(S_{\cN \backslash \{i\} \backslash [t]}^i|S^{[t]}\right) + \sum_{i=t+2}^k H\left(S_i^{[t+1:i-1]}|S^{[t]}\right) - \left( \sum_{i=t+1}^k  H \left(S^i|S^{[t]}\right) - \sum_{i=t+1}^k  H \left(S^i|S^{[i-1]}\right)  \right) \nonumber \\
 	& = \sum_{i=t+2}^k  \frac{i-1-t}{d-t} H\left(S_{\cN \backslash \{i\} \backslash [t]}^i|S^{[t]}\right) + \sum_{i=t+2}^k H\left(S_i^{[t+1:i-1]}|S^{[t]}\right) -  \left( \sum_{i=t+1}^k H\left(S^{i}|S^{[t]}\right) - H\left(S^{[t+1:k]}|S^{[t]}\right) \right),  \label{eq:Appendix-B:temple1}
\end{align} 
where \uref{a} follows from the well-known Han's inequality, and \uref{b} is justified because $\{S_i^{[t+1:i-1]},S_{[t+1:i-1]}^i\}$ is a function of $S^{[t+1:i-1]}$ and also a function of $S^i$ .

By symmetry, we know that
\[\sum_{i=t+2}^k  \frac{i-1-t}{d-t} H\left(S_{\cN \backslash \{i\} \backslash [t]}^i|S^{[t]}\right) = \left( \sum_{i=t+2}^k  \frac{i-1-t}{d-t} \right) H\left(S^{t+1}|S^{[t]}\right),\]
and
\[ \sum_{i=t+1}^k H\left(S^{i}|S^{[t]}\right) = (k-t) H\left(S^{t+1}|S^{[t]}\right).\]
Hence, \eqref{eq:Appendix-B:temple1} can be written as 
\begin{align}
	 	& \sum_{i=t+2}^k I\left(S_i^{[t+1:i-1]};S_{[t+1:i-1]}^i|S^{[t]}\right)   \nonumber\\
 	& \geq  \left( \sum_{i=t+2}^k  \frac{i-1-t}{d-t} \right) H\left(S^{t+1}|S^{[t]}\right) + \sum_{i=t+2}^k H\left(S_i^{[t+1:i-1]}|S^{[t]}\right) - (k-t) H\left(S^{t+1}|S^{[t]}\right) + H\left(S^{[t+1:k]}|S^{[t]}\right). \label{eq:Appendix-B:temple2} 
\end{align}

By substituting \eqref{eq:Appendix-B:temple2} in \eqref{eq:app_temple_1}, we have
\begin{align}
	H\left(W_{[t+1:k]}|S^{[t]}\right)  & \leq \sum_{i=t+1}^k H\left(W_i|S^{[t]}\right)  - \left( \sum_{i=t+2}^k  \frac{i-1-t}{d-t} \right) H\left(S^{t+1}|S^{[t]}\right) -\sum_{i=t+2}^k H\left(S_i^{[t+1:i-1]}|S^{[t]}\right) \nonumber \\
	& ~~ + (k-t) H\left(S^{t+1}|S^{[t]}\right) -H\left(S^{[t+1:k]}|S^{[t]}\right) \nonumber \\
	& = \sum_{i=t+1}^k H\left(W_i|S^{[t]}\right)  + \left((k-t)- \sum_{i=t+2}^k  \frac{i-1-t}{d-t} \right) H\left(S^{t+1}|S^{[t]}\right) -\sum_{i=t+2}^k H\left(S_i^{[t+1:i-1]}|S^{[t]}\right) \nonumber \\
	& ~~  -H\left(S^{[t+1:k]}|S^{[t]}\right) \nonumber \\
	& \utag{c}{=} \sum_{i=t+1}^k H\left(W_i|S^{[t]}\right)  + \frac{d-t+1}{2} H\left(S^{t+1}|S^{[t]}\right) -\sum_{i=t+2}^k H\left(S_i^{[t+1:i-1]}|S^{[t]}\right) - H\left(S^{[t+1:k]}|S^{[t]}\right), \label{eq:app_temple_revised}
\end{align}
where \uref{c} follows from $k=d$.

From the union bound, we know that
\[\sum_{i=t+2}^k H\left(S_i^{[t+1:i-1]}|S^{[t]}\right) \geq H\left(S_{[t+1:k]}|S^{[t]}\right),\]
and then we can further bound \eqref{eq:app_temple_revised} as follows:
\begin{align}
	H(W_{[t+1:k]}|S^{[t]}) 
	& \leq \sum_{i=t+1}^k H\left(W_i|S^{[t]}\right) + \frac{d-t+1}{2}  H\left(S^{t+1}|S^{[t]}\right) -H\left(S_{[t+1:k]}|S^{[t]}\right) - H\left(S^{[t+1:k]}|S^{[t]}\right)  \nonumber \\
	& = \sum_{i=t+1}^k H\left(W_i|S^{[t]}\right) + \frac{d-t+1}{2}  H\left(S^{t+1}|S^{[t]}\right) - H\left(S^{[t+1:k]}|S^{[t]}\right) \nonumber \\
	& ~~ -  H\left(S_{[t+1:n]}|S^{[t]}\right) +  H\left(S_{[t+1:n]}|S^{[t]}, S_{[t+1:k]}\right).    \label{eq:app_temple_3} 
\end{align}

By re-arranging \eqref{eq:app_temple_3}, we have
\begin{align*}
	& H(W_{[t+1:k]}|S^{[t]}) +  H\left(S^{[t+1:k]}|S^{[t]}\right) +  H\left(S_{[t+1:n]}|S^{[t]}\right) \\
	& \leq \sum_{i=t+1}^k H\left(W_i|S^{[t]}\right) + \frac{d-t+1}{2}  H\left(S^{t+1}|S^{[t]}\right) 
	 +  H\left(S_{[t+1:n]}|S^{[t]}, S_{[t+1:k]}\right) \\
	& \utag{d}{=} \sum_{i=t+1}^k H\left(W_i|S^{[t]}\right)+ \frac{d-t+1}{2}  H\left(S^{t+1}|S^{[t]}\right) +  H\left(S_n^{[t+1:k]}|S^{[t]}, S_{[t+1:k]}\right) \\
	& \leq \sum_{i=t+1}^k H\left(W_i|S^{[t]}\right)+ \frac{d-t+1}{2}  H\left(S^{t+1}|S^{[t]}\right) +  H\left(W_n|S_n^{[t]}\right) \\
	& \leq \sum_{i=t+1}^k H\left(W_i|S_i^{[t]}\right)+ \frac{d-t+1}{2}  H\left(S^{t+1}|S^{[t]}\right) +  H\left(W_n|S_n^{[t]}\right) \\
	& \utag{e}{=} (k-t+1) H\left(W_{n}|S_{n}^{[t]}\right)+ \frac{d-t+1}{2}  H\left(S^{t+1}|S^{[t]}\right) \\
	& \utag{f}{=} (d-t+1) \alpha - (d-t+1) H\left(S_{n}^{[t]}\right)+ \frac{d-t+1}{2}  H\left(S^{t+1}|S^{[t]}\right), 
\end{align*}
where  \uref{d} and \uref{f} follow from $k = d = n-1$, and \uref{e} follows from the symmetry.

From Proposition \ref{pro:k_eq_d_1}, we know that 
\[H(W_{[t+1:k]},S^{[t]}) = H\left(S^{[t+1:k]},S^{[t]}\right) = H\left(S_{[t+1:n]},S^{[t]}\right) = H(W_{[k]}),\]
so
\[H(W_{[t+1:k]}|S^{[t]}) = H\left(S^{[t+1:k]}|S^{[t]}\right) = H\left(S_{[t+1:n]}|S^{[t]}\right).\]
Hence we have
\begin{align*}
	3 H(W_{[t+1:k]}|S^{[t]}) 
		& \leq  (d-t+1) \alpha - (d-t+1) H\left(S_{n}^{[t]}\right)+ \frac{d-t+1}{2}  H\left(S^{t+1}|S^{[t]}\right),
\end{align*}
or
\begin{equation}
\label{eq:k_eq_d_Jan2_temple1}
	H(W_{[t+1:k]}|S^{[t]}) \leq \frac{d-t+1 }{3} \alpha - \frac{d-t+1}{3} H\left(S_{n}^{[t]}\right) + \frac{d-t+1}{6} H\left(S^{t+1}|S^{[t]}\right),
\end{equation}
for $t=0,\ldots,k-2$.

Now, consider $H\left(W_{[\ell+1:k]}|S^{[\ell]} \right)$.
For any $t < \ell$, we have 
\begin{align}
H\left(W_{[\ell+1:k]}|S^{[\ell]} \right)
& = H \left(W_{[\ell+1:k]},S^{[\ell]}\right) - H \left(S^{[t]}\right) - H \left(S^{[\ell]}|S^{[t]}\right) \nonumber \\
& \utag{g}{=} H \left(W_{[t+1:k]},S^{[t]}\right) - H \left(S^{[t]}\right) - H \left(S^{[\ell]}|S^{[t]}\right) \nonumber \\
& = H \left(W_{[t+1:k]}|S^{[t]}\right) - H \left(S^{[\ell]}|S^{[t]}\right) \nonumber \\
& = H \left(W_{[t+1:k]}|S^{[t]}\right) - \sum_{i=t+1}^{\ell} H\left(S^i|S^{[i-1]}\right), \label{eq:k_eq_d_wk}
\end{align}
where \uref{g} follows from Proposition \ref{pro:k_eq_d_1}. 

Since $t < \ell$ and $\ell \leq k-1$, we have $t \leq k-2$. Then by invoking the upper bound on $H \left(W_{[t+1:k]}|S^{[t]}\right)$ in \eqref{eq:k_eq_d_Jan2_temple1},
 we have
\begin{equation*}
H \left(W_{[\ell+1:k]}|S^{[\ell]}\right) \leq \frac{d+1 - t }{3} \alpha - \frac{d+1 - t }{3} H\left(S_{n}^{[t]}\right) + \frac{d+1 - t}{6} H\left(S^{t+1}|S^{[t]}\right)  - \sum_{i=t+1}^{\ell} H\left(S^i|S^{[i-1]}\right),
\end{equation*}
which completes the proof.

\section{Proof of propositions in Subsection \ref{subsec:k_eq_d_equal}}
\subsection{Proof of Proposition \ref{Prop:convex}}
\label{App:convex}
Recall that 
\begin{equation*}
	\mu_t =
	\begin{cases}
	 	 \frac{1}{2}  \binom{n-\hat{\ell}}{2}   \frac{n-2\hat{\ell}-1+t}{\binom{n-t}{4}},  & 1\leq t \leq \hat{\ell}-3, \\
	 	 \frac{6(n-\hat{\ell}-3)}{(n-\hat{\ell}+1)(n-\hat{\ell}+2)}, &  t = \hat{\ell}-2, ~ \hat{\ell} \geq 3, \\
	 	 \frac{6}{n-\hat{\ell}+1}, &   t=\hat{\ell}-1, ~ \hat{\ell} \geq 2, 
	\end{cases}
\end{equation*}
and 
\begin{equation*}
	\mu_0 = 1 - \sum_{j=1}^{\hat{\ell}-1} \mu_{j}.
\end{equation*}

Since $\hat{\ell} = \left \lceil \frac{1}{4}(n-2) \right \rceil$, if $\hat{\ell} \geq 3$, we have $n \geq 11$, and hence
\[n-\hat{\ell}-3 = n- \left \lceil \frac{1}{4}(n-2) \right \rceil-3 > n-  \frac{1}{4}(n-2) -4 = \frac{3}{4} \left( n- \frac{14}{3} \right) >0.\]
Also, when $1\leq t \leq \hat{\ell}-3$, we have
\[n-2\hat{\ell}-1+t \geq n-2 \hat{\ell} = n- 2 \left \lceil \frac{1}{4}(n-2) \right \rceil > \frac{1}{2}n-1 >0. \]
Hence, $\mu_t \geq 0$ for $t=1,\ldots,\hat{\ell}-1$ and so it remains to show that $\mu_0 \geq 0$.

For $\hat{\ell} = 1$, this is trivial as $\mu_0 = 1$. For $\hat{\ell} \geq 2$, we claim that
\begin{equation}
\label{eq:app_mu_t_temple1}
	\mu_0 = \frac{\left(n-\hat{\ell}\right)\left(n-\hat{\ell}-1\right) \left( n+1-4 \hat{\ell} \right) }{(n-1)(n-2)(n-3)}.	
\end{equation}
To see this, we first separately discuss the cases $\hat{\ell} = 2$ and $\hat{\ell} = 3$, where $\mu_t, t=0,\ldots,\hat{\ell}-1$ are as given as follows:
\begin{itemize}
	\item $\hat{\ell} = 2$
	\begin{equation}
	\label{eq:app_mu_ell_2}
	\mu_t =
		\begin{cases}
		  \frac{n-7}{n-1}, &  t=0, \\
		  \frac{6}{n-1}, &  t=1.
		\end{cases}
	\end{equation}
	\item $\hat{\ell} = 3$
	\begin{equation}
	\label{eq:app_mu_ell_3}
	\mu_t = 
		\begin{cases}
			\frac{(n-11)(n-4)}{(n-2)(n-1)}, & t=0, \\
			\frac{6(n-6)}{(n-2)(n-1)}, & t=1, \\
			\frac{6}{n-2}, & t=2. 
		\end{cases}	
	\end{equation}
\end{itemize}
Then we can easily verify that \eqref{eq:app_mu_t_temple1} holds.

For $\hat{\ell} \geq 4$ and any $j=1,\ldots,\hat{\ell}-3$, we have
\begin{align}
		\sum_{i=1}^{j} \mu_i   & = \sum_{i=1}^{j} \frac{1}{2}  \binom{n-\hat{\ell}}{2} \frac{n-2\hat{\ell}-1+i}{\binom{n-i}{4}} \nonumber \\
		&  \utagmodify{\ast}{=}  \frac{(n-\hat{\ell})(n-\hat{\ell}-1)(n-4\hat{\ell}+1 + 3j)}{(n-j-1)(n-j-2)(n-j-3)} +\frac{(n-\hat{\ell})(n-\hat{\ell}-1)(4\hat{\ell}-n-1)}{(n-1)(n-2)(n-3)}, \label{eq:k_eq_d_sum_mu_j}
\end{align}
where the above inequality and some other algebraic equalities in the sequel which are marked by an asterisk can be verified by symbolic computation application such as SageMath \cite{sage}. The steps are very lengthy and they are omitted here.

By substituting $j=\hat{\ell}-3$, we have
\[\sum_{i=1}^{\hat{\ell}-3} \mu_i  = \frac{(n-\hat{\ell})(n-\hat{\ell}-1)(n-\hat{\ell}-8)}{(n-\hat{\ell})(n-\hat{\ell}+1)(n-\hat{\ell}+2)} +\frac{(n-\hat{\ell})(n-\hat{\ell}-1)(4\hat{\ell}-n-1)}{(n-1)(n-2)(n-3)}.\]
Since 
\[ \mu_{\hat{\ell}-2}  = \frac{6(n-\hat{\ell}-3)}{(n-\hat{\ell}+1)(n-\hat{\ell}+2)},\]
and 
\[ \mu_{\hat{\ell}-1} = \frac{6}{n-\hat{\ell}+1},\]
we have 
\begin{align*}
	\sum_{i=1}^{\hat{\ell}-1} \mu_i & = 1  +\frac{(n-\hat{\ell})(n-\hat{\ell}-1)(4\hat{\ell}-n-1)}{(n-1)(n-2)(n-3)} \\
							  & = 1 - \frac{(n-\hat{\ell})(n-\hat{\ell}-1) \left( n+1-4 \hat{\ell} \right) }{(n-1)(n-2)(n-3)},
\end{align*}
and so
\[ \mu_0 = \frac{(n-\hat{\ell})(n-\hat{\ell}-1) \left( n+1-4 \hat{\ell} \right) }{(n-1)(n-2)(n-3)}.\]

Therefore, we obtain 
\begin{equation}
\label{eq:appendix_mu0}
\mu_0 =
	\begin{cases}
		 \frac{(n-\hat{\ell})(n-\hat{\ell}-1) \left( n+1-4 \hat{\ell} \right) }{(n-1)(n-2)(n-3)} & \hat{\ell} \geq 2, \\
		 1  & \hat{\ell} = 1.
	\end{cases}
\end{equation}
Since
\[n + 1- 4 \hat{\ell} = n + 1- 4 \left\lceil \frac{1}{4} (n-2) \right\rceil > n + 1- (n-2) -4 = -1, \]
and $n + 1- 4 \hat{\ell}$ is an integer, we have
\[n + 1- 4 \hat{\ell}  \geq 0,\]
and thus $\mu_0 \geq 0$ for $\hat{\ell} \geq 2$.

\subsection{Proof of Proposition \ref{Prop:geq0}}
\label{App:c_geq_0}
We need to prove that $c_{t} \geq 0$ for $t=0,\ldots,\hat{\ell}-1$ and $c_{\hat{\ell}-1} = 0$ when $\hat{\ell} \geq 2$.  
Recall that 
\[c_t = \frac{n-t}{6} \mu_{t} - \sum_{j=0}^{t} \mu_j,~ t=0,\ldots,\hat{\ell}-1.\]
First, we show that $c_{\hat{\ell}-1} = 0$ for $\hat{\ell} \geq 2$ as follows:
\[c_{\hat{\ell}-1} = \frac{n-\hat{\ell}+1}{6} \mu_{\hat{\ell}-1} - \sum_{j=0}^{\hat{\ell}-1} \mu_j = 1 -1 =0.\]
For $t=0$, it is easy to see that 
\[c_{0} = \frac{n}{6} \mu_{0} - \mu_0 = \frac{n-6}{6} \mu_0 \geq 0,\]
as $\hat{\ell} \geq 2$ implies that $n \geq 7$, and we know from Proposition \ref{Prop:convex} that $\mu_0 \geq 0$. Clearly, the proposition is proved for $\hat{\ell} = 2$, and it remains to verify that $c_t \geq 0$ for $t= 1,\dots,\hat{\ell}-2$ for $\hat{\ell} \geq 3$.

If $\hat{\ell} = 3$, we obtain from \eqref{eq:app_mu_ell_3} that
\begin{equation}
\label{eq:app_c_ell_3_t2}
c_1 = \frac{n-1}{6} \mu_{1} - \sum_{j=0}^{1} \mu_j =  \frac{n-6}{n-2}  - \left(1- \frac{6}{n-2} \right) = \frac{2}{n-2} \geq 0,
\end{equation}
which completes the proof for $\hat{\ell} = 3$.

For $\hat{\ell} \geq 4$ and any $t=1,\ldots,\hat{\ell}-3$, we have
\[\]
\begin{align}
	c_t  & = \frac{n-t}{6} \mu_{t} - \sum_{j=0}^{t} \mu_j  \nonumber \\
		 & \utagmodify{\ast}{=} \frac{2 (n-\hat{\ell})(n-\hat{\ell}-1)(\hat{\ell}- 1 - t)}{(n-t-1)(n-t-2)(n-t-3)}. \label{eq:k_eq_d_ct}
\end{align}
It is easy to see that $c_t \geq 0$ for any $t=1,\ldots,\hat{\ell}-3$ from \eqref{eq:k_eq_d_ct}.

If $t=\hat{\ell}-2$, we have 
\[c_{\hat{\ell}-2}  = \frac{n-\hat{\ell}+2}{6} \mu_{\hat{\ell}-2} - \sum_{j=0}^{\hat{\ell}-2} \mu_j = \frac{n-\hat{\ell}+2}{6} \mu_{\hat{\ell}-2} - ( 1 - \mu_{\hat{\ell}-1} )  = \frac{2}{n-\hat{\ell}+1} > 0.\]
Hence, we obtain that $c_t \geq 0, t=1,\ldots,\hat{\ell}-2$ for $\hat{\ell} \geq 4$.

In summary, for $\hat{\ell} \geq 2$, we have
\begin{equation}
\label{eq:app_c_t}
	c_t =  
	\begin{cases}
		\frac{n-6}{6} \mu_0, & t=0, \\ 
		\frac{2 (n-\hat{\ell})(n-\hat{\ell}-1)(\hat{\ell}- 1 - t)}{(n-t-1)(n-t-2)(n-t-3)},  & 1 \leq t \leq \hat{\ell}-3, \\
		\frac{2}{n-\hat{\ell}+1},  &  t = \hat{\ell}-2, ~ \hat{\ell} \geq 3, \\
		0,   &  t = \hat{\ell}-1,
	\end{cases}
\end{equation}
which substantiates that $c_{\hat{\ell}-1} = 0$ and $c_t \geq 0$ for all possible $t$ and $\hat{\ell} \geq 2$.

\subsection{Proof of Proposition \ref{Prop:namda}}
\label{APP:namba}
We need to verify that $\lambda_t = c_{t-1}  (d+1-t) - c_t  (d-t) - b_t =0$ for $t=2,\ldots, \hat{\ell}-1$ and $\hat{\ell} \geq 3$. Recall that \eqref{eq:app_c_t} and
\[b_t = \frac{n-t}{3} \mu_t.\]

If $t=\hat{\ell}-1$, we have
\[\lambda_t = c_{\hat{\ell}-2}  (d+2-\hat{\ell}) - c_{\hat{\ell}-1}  (d-\hat{\ell}+1) - b_{\hat{\ell}-1} = 2 -  \frac{(n-\hat{\ell}+1)}{3} \mu_{\hat{\ell}-1} = 0.\]
If $t=\hat{\ell}-2$ (implies that $\hat{\ell} \geq 4$), we have
\[\lambda_t = c_{\hat{\ell}-3}  (d+3-\hat{\ell}) - c_{\hat{\ell}-2}  (d+2-\hat{\ell}) - b_{\hat{\ell}-2} = \frac{4 (n-\hat{\ell}-1)}{n-\hat{\ell}+1} - 2 - \frac{n-\hat{\ell}+2}{3} \mu_{\hat{\ell}-2} = 0.\]
If $2 \leq t \leq \hat{\ell}-3$ (implies that $\hat{\ell} \geq 5$), we have
\begin{align*}
		\lambda_t & = c_{t-1}  (d+1-t) - c_t  (d-t) - b_t \\
		& = \frac{2 (n-\hat{\ell})(n-\hat{\ell}-1)(\hat{\ell} - t)}{(n-t-1)(n-t-2)} - \frac{2 (n-\hat{\ell})(n-\hat{\ell}-1)(\hat{\ell}- 1 - t)}{(n-t-2)(n-t-3)} - \frac{n-t}{6}   \binom{n-\hat{\ell}}{2}   \frac{n-2\hat{\ell}-1+t}{\binom{n-t}{4}} \\
		 		& \utagmodify{\ast}{=} 0.
\end{align*}
Therefore, we conclude that $\lambda_t = c_{t-1}  (d+1-t) - c_t  (d-t) - b_t =0$ for $t=2,\ldots, \hat{\ell}-1$.

\subsection{Proof of Proposition \ref{prop:k_eq_d_final}}
\label{App:k_eq_d_final}
We first prove that $T_2 \geq 0$ as follows:
\begin{align*}
	T_2  & =  b_1 -c_0 d +c_1(d-1) \\
		 & = \frac{n-1}{3} \mu_1 - \frac{(n-6)(n-1)}{6} \mu_0 + \left( \frac{n-1}{6}\mu_1 - \mu_0- \mu_1  \right)(n-2) \\
		 & = \left(\frac{1}{6}(n-7)(n-2) +  \frac{1}{3}(n-1) \right) \mu_1  - \left(  \frac{1}{6}(n-1)(n-6)+(n-2) \right) \mu_{0} \\
		 & \utagmodify{\ast}{=} \frac{(4\hat{\ell}+2-n)(n-\hat{\ell})(n-\hat{\ell}-1)}{6(n-2)} \\
		 & \utag{a}{\geq} 0,
\end{align*}
where \uref{a} is justified because $\hat{\ell}= \left\lceil \frac{1}{4}(n-2)  \right\rceil \geq \frac{1}{4}(n-2)$.

Now, we focus on $T_1$. For $\hat{\ell} = 2$, we have
\[T_1    =  \sum_{t=0}^{\hat{\ell}-1} \frac{n-t}{3} \mu_t 
		   =  \frac{n}{3} \mu_0 + \frac{n-1}{3} \mu_1 
		   \utag{a}{=}  \frac{n(n-7)}{3(n-1)} + 2 
		   = \frac{(n-\hat{\ell})(n-\hat{\ell}-1)}{(n-3)(n-2)}  \left( \frac{n ( n+1-4 \hat{\ell} ) }{3(n-1)} + 2(\hat{\ell}-1) \right),
\]
where \uref{a} follows from \eqref{eq:app_mu_ell_2}. 
For $\hat{\ell} = 3$, we have 
\begin{align*}
	T_1   & =  \sum_{t=0}^{\hat{\ell}-1} \frac{n-t}{3} \mu_t \\
		  & =  \frac{n}{3} \mu_0 + \frac{n-1}{3} \mu_1 + \frac{n-2}{3} \mu_2\\
		  & \utag{a}{=}  \frac{n(n-4)(n-11)}{3(n-1)(n-2)}  + \frac{2(n-6)}{n-2}  + 2 \\
		  & = \frac{(n-\hat{\ell})(n-\hat{\ell}-1)}{(n-3)(n-2)}  \left( \frac{n ( n+1-4 \hat{\ell} ) }{3(n-1)} + 2(\hat{\ell}-1) \right),
\end{align*}
where \uref{a} follows from \eqref{eq:app_mu_ell_3}.
For $\hat{\ell} \geq 4$, we have
\begin{align*}
	T_1   & =  \sum_{t=0}^{\hat{\ell}-1} \frac{n-t}{3} \mu_t   \\
		  & = \frac{n}{3} \mu_0 + \frac{n-\hat{\ell}+2}{3} \mu_{\hat{\ell}-2} + \frac{n-\hat{\ell}+1}{3} \mu_{\hat{\ell}-1}  +  \sum_{t=1}^{\hat{\ell}-3} \frac{n-t}{3} \mu_t \\
						 & = \frac{n(n-\hat{\ell})(n-\hat{\ell}-1) ( n+1-4 \hat{\ell} ) }{3(n-1)(n-2)(n-3)} + \frac{4(n-\hat{\ell}-1)}{n-\hat{\ell}+1} +  \sum_{t=1}^{\hat{\ell}-3}   \frac{2(n-2\hat{\ell}-1+t)(n-\hat{\ell})(n-\hat{\ell}+1)}{(n-t-1)(n-t-2)(n-t-3)} 	\\
						 & \utagmodify{\ast}{=} \frac{(n-\hat{\ell})(n-\hat{\ell}-1)}{(n-3)(n-2)}  \left( \frac{n ( n+1-4 \hat{\ell} ) }{3(n-1)} + 2(\hat{\ell}-1) \right).	 
\end{align*}

Therefore, for $\hat{\ell} \geq 2$, we obtain
\[T_1 = \frac{(n-\hat{\ell})(n-\hat{\ell}-1)}{(n-3)(n-2)}  \left( \frac{n ( n+1-4 \hat{\ell} ) }{3(n-1)} + 2(\hat{\ell}-1) \right),\]
and we can verify that
\begin{align*}
	T_1  - \frac{T_2}{d}  
	& = \frac{(n-\hat{\ell})(n-\hat{\ell}-1)}{(n-3)(n-2)}  \left( \frac{n ( n+1-4 \hat{\ell} ) }{3(n-1)} + 2(\hat{\ell}-1) \right) -  \frac{(4\hat{\ell}+2-n)(n-\hat{\ell})(n-\hat{\ell}-1)}{6(n-2)(n-1)} \\
		& \utagmodify{\ast}{=} \frac{(n-\hat{\ell})(n-\hat{\ell}-1)}{2(n-1)} \\
		& = \frac{\Gamma_{k,d,\hat{\ell}}}{d}.
\end{align*}

\section{Proof of Lemma \ref{lem:k_eq_d_revised}}
\label{App:k_eq_d_revised}
From \eqref{eq:app_temple_revised}, we know that for $t=0,\ldots,k-2$,
\begin{align*}
	H\left(W_{[t+1:k]}|S^{[t]}\right)  & \leq \sum_{i=t+1}^k H\left(W_i|S^{[t]}\right)  + \frac{d-t+1}{2} H\left(S^{t+1}|S^{[t]}\right) -\sum_{i=t+2}^k H\left(S_i^{[t+1:i-1]}|S^{[t]}\right)  -H\left(S^{[t+1:k]}|S^{[t]}\right).
\end{align*}
Then we have
\begin{align*}
	H\left(W_{[t+1:k]}|S^{[t]}\right)  
	& \leq \sum_{i=t+1}^k H\left(W_i|S^{[t]}\right)  + \frac{d-t+1}{2} H\left(S^{t+1}|S^{[t]}\right) -\sum_{i=t+2}^k H\left(S_i^{[t+1:i-1]}|S^{[t]}\right)  -H\left(S^{[t+1:k]}|S^{[t]}\right)  \\
	& = \sum_{i=t+1}^k \left(H\left(W_i|S^{[t]}\right) - H\left(S_i^{[t+1:i-1]}|S^{[t]}\right)\right)  + \frac{d-t+1}{2} H\left(S^{t+1}|S^{[t]}\right)    - H\left(S^{[t+1:k]}|S^{[t]}\right)  \\
	& = \sum_{i=t+1}^k H\left(W_i|S^{[t]}, S_i^{[t+1:i-1]} \right)  + \frac{d-t+1}{2} H\left(S^{t+1}|S^{[t]}\right)   -H\left(S^{[t+1:k]}|S^{[t]}\right) \\
	& \leq \sum_{i=t+1}^k H\left(W_i| S_i^{[i-1]} \right)  + \frac{d-t+1}{2} H\left(S^{t+1}|S^{[t]}\right)   -H\left(S^{[t+1:k]}|S^{[t]}\right) \\
	& = (k-t) \alpha - \sum_{i=t+1}^k H\left(S_i^{[i-1]} \right)  + \frac{d-t+1}{2} H\left(S^{t+1}|S^{[t]}\right)   -H\left(S^{[t+1:k]}|S^{[t]}\right) \\
	& \utag{a}{=} (k-t) \alpha  - \sum_{i=t}^{k-1} H\left( S_n^{[i]} \right) + \frac{d-t+1}{2}  H\left(S^{t+1}|S^{[t]}\right) -H\left(S^{[t+1:k]}|S^{[t]}\right),
\end{align*}
where \uref{a} follows from the symmetry.
Therefore, we have
\begin{align*}
	2 H(W_{[t+1:k]}|S^{[t]}) & \utag{b}{=} H(W_{[t+1:k]}|S^{[t]}) + H\left(S^{[t+1:k]}|S^{[t]}\right) \\
							 & \leq (k-t) \alpha  - \sum_{i=t}^{k-1} H\left( S_n^{[i]} \right) + \frac{d-t+1}{2}  H\left(S^{t+1}|S^{[t]}\right),
\end{align*}
where \uref{b} follows from Proposition~\ref{pro:k_eq_d_1}.

Upon dividing by $2$, we obtain 
\begin{equation}
\label{eq:Appendix-D:temple1}
	H(W_{[t+1:k]}|S^{[t]}) \leq \frac{1}{2}(k-t) \alpha - \frac{1}{2} \sum_{i=t}^{k-1} H(S_n^{[i]}) + \frac{1}{4}(d-t+1) H(S^{t+1}|S^{[t]}),
\end{equation}
for $t=0,\ldots,k-2$.

Since $\ell \leq k-1$, we know that $\ell -1 \leq k-2$. By letting $t=\ell-1$ in \eqref{eq:Appendix-D:temple1}, we have
\begin{equation}
\label{eq:Appendix-D:temple3}
	H(W_{[\ell:k]}|S^{[\ell-1]}) \leq \frac{1}{2}(k-\ell+1) \alpha - \frac{1}{2} \sum_{i=\ell-1}^{k-1} H(S_n^{[i]}) + \frac{1}{4}(d-\ell+2) H(S^{\ell}|S^{[\ell-1]}).
\end{equation}
Finally, consider
\begin{align}
	H(W_{[\ell+1:k]}|S^{[\ell]}) 
	& = H(W_{[\ell+1:k]},S^{[\ell]}) - H(S^{[\ell]}) \nonumber \\
	& \utag{c}{=} H(W_{[\ell:k]},S^{[\ell-1]}) - H(S^{[\ell]}) \nonumber \\
	& = H(W_{[\ell:k]}|S^{[\ell-1]}) - H(S^{[\ell]}|S^{[\ell-1]}), \label{eq:Appendix_D:temple11}
\end{align}
where \uref{c} follows from Proposition~\ref{pro:k_eq_d_1}. By substituting \eqref{eq:Appendix-D:temple3} into \eqref{eq:Appendix_D:temple11}, we obtain that
\begin{align*}
	H(W_{[\ell+1:k]}|S^{[\ell]}) 
	& \leq \frac{1}{2}(k-\ell+1) \alpha - \frac{1}{2} \sum_{i=\ell-1}^{k-1} H(S_n^{[i]}) + \frac{1}{4}(d-\ell+2) H(S^{\ell}|S^{[\ell-1]}) - H(S^{\ell}|S^{[\ell-1]})\\
	& = \frac{1}{2}(k-\ell+1) \alpha - \frac{1}{2} \sum_{i=\ell-1}^{k-1} H(S_n^{[i]}) + \frac{1}{4}(k-\ell-2) H(S^{\ell}|S^{[\ell-1]}),
\end{align*}
which completes the proof.

\section{Proof of Lemma \ref{lem:k_less_d_tech}}
\label{App:k_less_d_tech}
We first prove \eqref{eq:k_less_d_S_bound}. For any $y=\ell+1,\ldots,k$ and $\ell \leq t_y \leq y-1$, consider
\begin{align}
	H\left (S^y| S^{[t_y]}, W_{[t_y+1:y-1]}\right)  
	& = H\left (S_{[y-1]\cup[y+1:d+1]}^y| S^{[t_y]}, W_{[t_y+1:y-1]}  \right)  \nonumber \\
	& = H\left (S_{[y+1:d+1]}^y| S^{[t_y]}, W_{[t_y+1:y-1]}  \right)  \nonumber \\
							& \utag{a}{\leq} H\left(S_{[y+1:d+1]}^y| S^{[t_y]}, S_{[t_y+1:y-1]}^y \right)  \nonumber \\
							& \utag{b}{\leq} \frac{d+1-y}{d-t_y}   H\left(S_{[y+1:d+1]}^y, S_{[t_y+1:y-1]}^y|S^{[t_y]} \right)  \nonumber \\
							& = \frac{d+1-y}{d-t_y}   H\left(S^y|S^{[t_y]} \right)  \nonumber \\
							& \utag{c}{=} \frac{d+1-y}{d-t_y} H\left(S^{t_y+1}|S^{[t_y]} \right),  \label{eq:k_less_d_new_temple1}
\end{align}
where \uref{a} follows because $S_{[t_y+1:y-1]}^y$ is a function of $W_{[t_y+1:y-1]}$, \uref{b} follows  from Han's inequality, and \uref{c} is justified by invoking the symmetry.

Now, we focus on \eqref{eq:k_less_d_W_bound}. If $t_y=y-1$, we have
\[H\left (W_y| S^{[t_y]}, W_{[t_y+1:y-1]}\right) = H\left (W_y| S^{[y-1]}\right) \leq H\left (W_y| S_y^{[y-1]}\right) = \alpha - H\left (S_y^{[y-1]}\right) = \alpha - H\left (S_n^{[y-1]}\right),\]
where the last step follows from the symmetry. 
For $\ell \leq t_y \leq y-2$, consider
\begin{align}
	H\left (W_y| S^{[t_y]}, W_{[t_y+1:y-1]}\right) 
	& \utag{a}{=} H\left (W_y| S^{t_y}, S^{[t_y-1]}, W_{[t_y+1:y-1]}\right) \nonumber \\
	& \utag{b}{=} H\left (W_y| S^{t_y}, S^{[t_y-1]}, W_{t_y}, W_{[t_y+1:y-1]}\right) \nonumber \\
	& = H\left (W_y| S^{[t_y-1]},W_{t_y}, W_{[t_y+1:y-1]}\right) - I\left(W_y;S^{t_y}|S^{[t_y-1]},W_{t_y}, W_{[t_y+1:y-1]}\right) \nonumber\\
	& \utag{c}{=} H\left (S^y| S^{[t_y-1]},W_{t_y}, W_{[t_y+1:y-1]}\right)- H\left (S^y|W_y, S^{[t_y-1]},W_{t_y}, W_{[t_y+1:y-1]}\right) \nonumber\\
	& ~~ - I\left(W_y;S^{t_y}|S^{[t_y-1]},W_{t_y}, W_{[t_y+1:y-1]}\right) \nonumber\\
	& \utag{d}{=} H\left (S^y| S^{[t_y-1]},W_{t_y}, W_{[t_y+1:y-1]}\right)- H\left (S^{t_y}|W_{t_y}, S^{[t_y-1]},W_{y}, W_{[t_y+1:y-1]}\right) \nonumber\\
	& ~~ - I\left(W_y;S^{t_y}|S^{[t_y-1]},W_{t_y}, W_{[t_y+1:y-1]}\right) \nonumber\\
	& = H\left (S^y| S^{[t_y-1]},W_{t_y}, W_{[t_y+1:y-1]}\right)- H\left (S^{t_y}| S^{[t_y-1]},W_{t_y}, W_{[t_y+1:y-1]}\right)\nonumber \\
	& \utag{e}{=} H\left (S^y| S^{[t_y-1]},W_{t_y}, W_{[t_y+1:y-1]}\right)- H\left (S^{y-1}| S^{[t_y-1]},W_{y-1}, W_{[t_y:y-2]}\right) \nonumber \\
	& = H\left (S^y| S^{[t_y-1]}, W_{[t_y:y-1]}\right)- H\left (S^{y-1}| S^{[t_y-1]}, W_{[t_y:y-1]}\right), \label{eq:Appendix-E:temple1}
\end{align}
where \uref{a} follows because $t_y \geq \ell \geq 1$, so that $S^{t_y}$ is well defined,  \uref{b} follows because $W_{t_y}$ is a function of $S^{t_y}$, \uref{c} follows because $W_{y}$ is a function of $S^{y}$, \uref{d} is obtained from the symmetry by interchanging the indices $y$ and $t_y$ in the second term (while keeping the indices $[t_y -1]$ and $[t_y+1:y-1]$ fixed),
and \uref{e} is obtained from the symmetry by replacing the index $t_y$ by $y-1$ and the indices $[t_y+1:y-1]$ by $[t_y:y-2]$ (while keeping the indices $[t_y -1]$ fixed). 

On the right-hand side of \eqref{eq:Appendix-E:temple1}, we can upper bound the first term as
\begin{align}
	H\left (S^y| S^{[t_y-1]}, W_{[t_y:y-1]}\right)  
							& \leq \frac{d+1-y}{d+1-t_y} H\left(S^{t_y}|S^{[t_y-1]} \right); \label{eq:Appendix-E:temple2} 
\end{align}
this can be obtained by following the proof of \eqref{eq:k_less_d_new_temple1} step-by-step with $t_y$ replaced by $t_y-1$. For the second term on the right-hand side of \eqref{eq:Appendix-E:temple1}, we have
\begin{align}
	& H\left (S^{y-1}| S^{[t_y-1]}, W_{[t_y:y-2]}, W_{y-1}\right) \nonumber \\
	& \utag{a}{=} H\left (S^{y-1}| S^{[t_y-1]}, W_{[t_y:y-2]}, W_{y-1},  S_{y-1}^{[t_y:y-2]}\right) \nonumber \\
	& = H\left (S^{y-1}| S^{[t_y-1]}, W_{[t_y:y-2]}, S_{y-1}^{[t_y:y-2]}\right) - I \left(S^{y-1};W_{y-1}| S^{[t_y-1]}, W_{[t_y:y-2]}, S_{y-1}^{[t_y:y-2]}\right) \nonumber \\
	& \utag{b}{\geq} H\left (S^{y-1}| S^{[y-2]}\right) - I \left(S^{y-1};W_{y-1}| S^{[t_y-1]}, W_{[t_y:y-2]}, S_{y-1}^{[t_y:y-2]}\right) \nonumber \\
	& \utag{c}{=}  H\left (S^{y-1}| S^{[y-2]}\right) - H \left(W_{y-1}| S^{[t_y-1]}, W_{[t_y:y-2]}, S_{y-1}^{[t_y:y-2]}\right) \nonumber \\
	& \utag{d}{\geq}  H\left (S^{y-1}| S^{[y-2]}\right) - H \left(W_{y-1}|S_{y-1}^{[y-2]}\right) \nonumber \\
	& = H\left (S^{y-1}| S^{[y-2]}\right) - \alpha + H \left(S_{y-1}^{[y-2]}\right) \nonumber \\
	& \utag{e}{=}  H\left (S^{y-1}| S^{[y-2]}\right) - \alpha + H \left(S_{n}^{[y-2]}\right), \label{eq:Appendix-E:temple3}
\end{align}
where \uref{a} follows because $S_{y-1}^{[t_y:y-2]}$ is a function of $W_{y-1}$, \uref{b} follows because $ \left\{S^{[t_y-1]}, W_{[t_y:y-2]}, S_{y-1}^{[t_y:y-2]} \right\}$ is a function of $S^{[y-2]}$, \uref{c} follows because $W_{y-1}$ is a function of $S^{y-1}$, \uref{d} follows because $S_{y-1}^{[y-2]}$ is a subset of $\left\{ S^{[t_y-1]}, W_{[t_y:y-2]}, S_{y-1}^{[t_y:y-2]} \right\}$, and \uref{e} follows from the symmetry.

By substituting \eqref{eq:Appendix-E:temple2} and \eqref{eq:Appendix-E:temple3} in \eqref{eq:Appendix-E:temple1}, we finally obtain that
\begin{align*}
	H\left (W_y| S^{[t_y]}, W_{[t_y+1:y-1]}\right) 
	& = H\left (S^y| S^{[t_y-1]}, W_{[t_y:y-1]}\right)- H\left (S^{y-1}| S^{[t_y-1]}, W_{[t_y:y-1]}\right) \\
	& \leq \frac{d+1-y}{d+1-t_y} H\left (S^{t_y}| S^{[t_y-1]}\right) - \left(  H\left (S^{y-1}| S^{[y-2]}\right) - \alpha + H \left(S_{n}^{[y-2]}\right) \right) \\
	& = \alpha - H \left(S_{n}^{[y-2]}\right) + \frac{d+1-y}{d+1-t_y} H\left (S^{t_y}| S^{[t_y-1]}\right)- H\left (S^{y-1}| S^{[y-2]}\right).
\end{align*}

\section{Proof of Proposition~\ref{proposition:k_less_d:approach:final}}
\label{Appendix:k_less_d:approach:final}
First, let us write $\bar{\mu}_j$ and $\bar{\nu}_j$ explicitly. For $j=\ell,\ldots,k$, let
\[\Lambda_1(j) = \left\{(x,y): q_{x,y}=1, t_{x,y}=j-1  \right\},\]
\[\Lambda_2(j) = \left\{(x,y): q_{x,y}=0, t_{x,y} = j, y \neq j+1 \right\},\]
and
\[\Lambda_3(j) = \left\{(x,y): q_{x,y}=0, t_{x,y} \neq j, y=j+1 \right\}.\]
Here, $\Lambda_1(j)$ is the set of all $(x,y)$ that contributes to the coefficient of $H\left(S^j|S^{[j-1]} \right)$ via the upper bound in \eqref{eq:k_less_d_S_bound}, $\Lambda_2(j)$ is the set of all $(x,y)$ that contributes to the coefficient of $H\left(S^j|S^{[j-1]} \right)$ via the third term in the upper bound in the first line of \eqref{eq:k_less_d_W_bound}, and $\Lambda_3(j)$ is the set of all $(x,y)$ that contributes to the coefficient of $H\left(S^j|S^{[j-1]} \right)$ via the fourth term in the upper bound in the first line of \eqref{eq:k_less_d_W_bound}.
Since $\Lambda_1(j)$, $\Lambda_2(j)$ and $\Lambda_3(j)$ are disjoint, for the row $x$, $\mu_{x,j}$ is defined by
\begin{equation*}
	\mu_{x,j} =  \sum_{y} \left(\bm{1}_{\Lambda_1(j)}((x,y)) \frac{d+1-y}{d+1-j}  +  \bm{1}_{\Lambda_2(j)}((x,y)) \frac{d+1-y}{d+1-j} -  \bm{1}_{\Lambda_3(j)}((x,y))\right),
\end{equation*}
where
\[\bm{1}_{\cA}(a) =
\begin{cases}
	1, & \text{if} ~ a \in \cA, \\
	0, & \text{if} ~ a \notin \cA. \\ 
\end{cases}\]
Then we have
\begin{equation}
	\bar{\mu}_j = \frac{1}{m} \sum_{x} \mu_{x,j} =  \frac{1}{m} \sum_{x,y} \left(\bm{1}_{\Lambda_1(j)}((x,y)) \frac{d+1-y}{d+1-j}  +  \bm{1}_{\Lambda_2(j)}((x,y)) \frac{d+1-y}{d+1-j} -  \bm{1}_{\Lambda_3(j)}((x,y))\right).
\end{equation}

Similarly, for $j=\ell,\ldots,k-1$, let 
\[\Delta_1(j) = \left\{(x,y): q_{x,y}=0, t_{x,y} = j, y=j+1\right\},\]
and
\[\Delta_2(j) = \left\{(x,y): q_{x,y}=0, t_{x,y} \neq j+1, y= j+2\right\}.\]
Here, $\Delta_1(j)$ is the set of all $(x,y)$ that contributes to the coefficient of $H(S_n^{[j]})$ via the second term in the upper bound in the second line of \eqref{eq:k_less_d_W_bound}, and $\Lambda_2(j)$ is the set of all $(x,y)$ that contributes to the coefficient of $H(S_n^{[j]})$ via the second term in the upper bound in the first line of \eqref{eq:k_less_d_W_bound}.
Since $\Delta_1(j)$ and $\Delta_2(j)$ are disjoint, for the row $x$, $\nu_{x,j}$ is defined by
\begin{equation*}
	\nu_{x,j} =  \sum_{y} \left(\bm{1}_{\Delta_1(j)}((x,y))   +  \bm{1}_{\Delta_2(j)}((x,y)) \right),
\end{equation*}
and so we have
\begin{equation}
	\bar{\nu}_j = \frac{1}{m} \sum_{x} \nu_{x,j} =  \frac{1}{m} \sum_{x,y} \left(\bm{1}_{\Delta_1(j)}((x,y))   +  \bm{1}_{\Delta_2(j)}((x,y)) \right).
\end{equation}

Now, consider
\begin{align}
	& \frac{\Gamma_{k,d,\ell} }{d}- \frac{d+1-\ell}{\ell} \bar{\mu}_\ell- \frac{1}{\ell}\left( \sum_{j=\ell+1}^{k} \delta_j \right) \nonumber \\
	& = \frac{\Gamma_{k,d,\ell}}{d}- \frac{d+1-\ell}{\ell} \bar{\mu}_\ell- \frac{1}{\ell} \sum_{j=\ell+1}^{k}\left( (d+1-j)\bar{\mu}_j - \sum_{i= j}^{k-1} \bar{\nu}_i \right) \nonumber \\
	& = \frac{\Gamma_{k,d,\ell}}{d}-  \frac{1}{\ell} \left( \sum_{j=\ell}^{k-1} \bar{\nu}_j \right) - \frac{1}{\ell} \sum_{j=\ell}^{k}\left( (d+1-j)\bar{\mu}_j - \sum_{i= j}^{k-1} \bar{\nu}_i \right) \nonumber \\
	& = \frac{\Gamma_{k,d,\ell}}{d}-  \frac{1}{\ell} \left( \sum_{j=\ell}^{k-1} \bar{\nu}_j \right) - \frac{1}{\ell} \sum_{j=\ell}^{k} (d+1-j)\bar{\mu}_j + \frac{1}{\ell} \sum_{j=\ell}^{k}\sum_{i= j}^{k-1} \bar{\nu}_i. \label{eq:k_less_d_Jan30_temple_1}
\end{align}
First, focus on $\sum_{j=\ell}^{k} (d+1-j) \bar{\mu}_j$. Then we have
\begin{align}
	& \sum_{j=\ell}^{k} (d+1-j) \bar{\mu}_j \nonumber \\
	&  = \sum_{j=\ell}^{k} (d+1-j) \frac{1}{m} \sum_{x,y} \left(\bm{1}_{\Lambda_1(j)}((x,y)) \frac{d+1-y}{d+1-j}  +  \bm{1}_{\Lambda_2(j)}((x,y)) \frac{d+1-y}{d+1-j} -  \bm{1}_{\Lambda_3(j)}((x,y))\right) \nonumber \\
	& = \frac{1}{m} \sum_{j=\ell}^{k}  \sum_{x,y} \left(\bm{1}_{\Lambda_1(j)}((x,y))(d+1-y)  +  \bm{1}_{\Lambda_2(j)}((x,y)) (d+1-y) -  \bm{1}_{\Lambda_3(j)}((x,y)) (d+1-j)\right) \nonumber \\
	& \utag{a}{=} \frac{1}{m} \sum_{j=\ell}^{k}  \sum_{x,y} \left(\bm{1}_{\Lambda_1(j)}((x,y))(d+1-y)  +  \bm{1}_{\Lambda_2(j)}((x,y)) (d+1-y) -  \bm{1}_{\Lambda_3(j)}((x,y)) (d+2-y)\right) \nonumber \\
	& =  \frac{1}{m}   \sum_{x,y}   \left((d+1-y) \sum_{j=\ell}^{k} \bm{1}_{\Lambda_1(j)}((x,y))  +  (d+1-y)\sum_{j=\ell}^{k} \bm{1}_{\Lambda_2(j)}((x,y)) - (d+2-y)\sum_{j=\ell}^{k} \bm{1}_{\Lambda_3(j)}((x,y)) \right), \label{eq:AppendixF:temple1}
\end{align}
where \uref{a} follows because for fixed $x$ and $y$, $\bm{1}_{\Lambda_3(j)}((x,y)) =1$ only if $j=y-1$.

Since $\Lambda_1(j) \cap \Lambda_1(j') = \emptyset$, $\Lambda_2(j) \cap \Lambda_2(j') = \emptyset$ and $\Lambda_3(j) \cap \Lambda_3(j') = \emptyset$ for $j \neq j'$, we have
\[\sum_{j=\ell}^{k} \bm{1}_{\Lambda_1(j)}((x,y)) =  \bm{1}_{\cup_j \Lambda_{1}(j)}((x,y)),\]
\[\sum_{j=\ell}^{k} \bm{1}_{\Lambda_2(j)}((x,y)) =  \bm{1}_{\cup_j \Lambda_{2}(j)}((x,y)),\]
and
\[\sum_{j=\ell}^{k} \bm{1}_{\Lambda_3(j)}((x,y)) =  \bm{1}_{\cup_j \Lambda_{3}(j)}((x,y)).\]
By examining the set $\cup_j \Lambda_{1}(j)$, we have
\[\bigcup_j \Lambda_{1}(j) =\left\{(x,y): q_{x,y}=1, \ell-1 \leq t_{x,y} \leq k-1  \right\}.\]
Since $\ell-1 \leq t_{x,y} \leq k-1$ always holds, we have
\[\bigcup_j \Lambda_{1}(j) =\left\{(x,y): q_{x,y}=1  \right\},\]
and hence
\[ \bm{1}_{\cup_j \Lambda_{1}(j)}((x,y)) = 
\begin{cases}
	1, & \text{if}~ q_{x,y} =1, \\
	0, & \text{if}~ q_{x,y} =0,\\
\end{cases}
\]
which is equivalent to  
\begin{equation}
\label{eq:AppendixF:lambda_1}
	\bm{1}_{\cup_j \Lambda_{1}(j)}((x,y)) = q_{x,y}.
\end{equation}

Similarly, for the sets $\cup_j \Lambda_{2}(j)$ and $\cup_j \Lambda_{3}(j)$, we have
\[\bigcup_j \Lambda_{2}(j) = \left\{(x,y): q_{x,y}=0, y \neq t_{x,y} + 1 \right\},\]
and
\[\bigcup_j \Lambda_{3}(j) = \left\{(x,y): q_{x,y}=0, t_{x,y} \neq y - 1 \right\}.\]
Note that $\cup_j \Lambda_{2}(j) = \cup_j \Lambda_{3}(j)$. By letting 
\begin{equation}
\label{eq:AppendixF:delta}
	\Delta = \bigcup_j \Lambda_{2}(j) = \bigcup_j \Lambda_{3}(j)
\end{equation}
we have 
\begin{equation}
\label{eq:AppendixF:lambda_2}
	\bm{1}_{\cup_j \Lambda_{2}(j)}((x,y)) = \bm{1}_{\cup_j \Lambda_{3}(j)}((x,y)) = \bm{1}_{\Delta}((x,y)).
\end{equation}
Hence, \eqref{eq:AppendixF:temple1} can be written as
\begin{align}
	& \sum_{j=\ell}^{k} (d+1-j) \bar{\mu}_j  \nonumber \\
	& = \frac{1}{m}   \sum_{x,y}   \left((d+1-y) \bm{1}_{\cup_j \Lambda_{1}(j)}((x,y)) +  (d+1-y)\bm{1}_{\cup_j \Lambda_{2}(j)}((x,y)) - (d+2-y)\bm{1}_{\cup_j \Lambda_{3}(j)}((x,y)) \right) \nonumber \\
	& = \frac{1}{m}  \sum_{x,y}   \left((d+1-y) q_{x,y}  -  \bm{1}_{\Delta}((x,y)) \right). \label{eq:k_less_d_Jan30_temple_2}
\end{align}

Now, focus on $\sum_{j=\ell}^{k} \sum_{i= j}^{k-1} \bar{\nu}_i$ in \eqref{eq:k_less_d_Jan30_temple_1}, and we have
\begin{align}
\sum_{j=\ell}^{k} \sum_{i= j}^{k-1} \bar{\nu}_i 
	& =  \sum_{j=\ell}^{k-1} (j-\ell+1) \bar{\nu}_j \nonumber \\
	& =  \sum_{j=\ell}^{k-1} (j-\ell+1)  \frac{1}{m} \sum_{x,y} \left(\bm{1}_{\Delta_1(j)}((x,y))   +  \bm{1}_{\Delta_2(j)}((x,y)) \right)	\nonumber  \\
	& =  \frac{1}{m}    \sum_{x,y} \left( \sum_{j=\ell}^{k-1} (j-\ell+1)\bm{1}_{\Delta_1(j)}((x,y))   +  \sum_{j=\ell}^{k-1} (j-\ell+1)\bm{1}_{\Delta_2(j)}((x,y)) \right)	\nonumber \\
	& \utag{b}{=}  \frac{1}{m}    \sum_{x,y} \left( \sum_{j=\ell}^{k-1} (y-\ell)\bm{1}_{\Delta_1(j)}((x,y))   +  \sum_{j=\ell}^{k-1} (y-\ell-1)\bm{1}_{\Delta_2(j)}((x,y)) \right)	\nonumber \\
	& = \frac{1}{m}    \sum_{x,y} \left( (y-\ell)\sum_{j=\ell}^{k-1} \bm{1}_{\Delta_1(j)}((x,y))   + (y-\ell-1) \sum_{j=\ell}^{k-1} \bm{1}_{\Delta_2(j)}((x,y)) \right), \nonumber
\end{align}
where \uref{b} follows because for fixed $x$ and $y$, $\bm{1}_{\Delta_1(j)}((x,y)) =1$ only if $y=j+1$ and $\bm{1}_{\Delta_2(j)}((x,y)) =1$ only if $y=j+2$.

Since $\Delta_1(j) \cap \Delta_1(j') = \emptyset$ and $\Delta_2(j) \cap \Delta_2(j') = \emptyset$ for $j \neq j'$, we have
\[\sum_{j=\ell}^{k-1} \bm{1}_{\Delta_1(j)}((x,y)) =  \bm{1}_{\cup_j \Delta_{1}(j)}((x,y)),\]
and
\[\sum_{j=\ell}^{k-1} \bm{1}_{\Delta_2(j)}((x,y)) =  \bm{1}_{\cup_j \Delta_{2}(j)}((x,y)).\]
By examining the sets $\cup_j \Delta_{1}(j)$ and $\cup_j \Delta_{2}(j)$, we have
\[\bigcup_j \Delta_{1}(j) =\left\{(x,y): q_{x,y}=0, t_{x,y} = y-1 \right\},\]
and 
\[\bigcup_j \Delta_{2}(j) =\left\{(x,y): q_{x,y}=0, t_{x,y} \neq y-1,   \ell+2 \leq y \leq k+1 \right\}.\]
Since $\ell+1 \leq y \leq k$, $\cup_j \Delta_{2}(j)$ can be written as $\left\{(x,y): q_{x,y}=0, t_{x,y} \neq y-1, y \neq \ell+1 \right\}$. 
Note that if $y=\ell+1$, then $t_{x,y} = \ell = y-1$, so we know that $t_{x,y} \neq y-1$ implies that $y \neq \ell+1$. Hence, $\cup_j \Delta_{2}(j)$ can be written as
\[\bigcup_j \Delta_{2}(j) =\left\{(x,y): q_{x,y}=0, t_{x,y} \neq y-1 \right\}.\]
We can easily see that $\cup_j \Delta_{2}(j) = \Delta$, where $\Delta$ is defined in \eqref{eq:AppendixF:delta}.

By letting 
\[\Delta' = \bigcup_j \Delta_{1}(j),\]
we have
\[\sum_{j=\ell}^{k-1} \bm{1}_{\Delta_1(j)}((x,y)) =  \bm{1}_{\cup_j \Delta_{1}(j)}((x,y)) =\bm{1}_{ \Delta'}((x,y)).\]
Also,
\[\sum_{j=\ell}^{k-1} \bm{1}_{\Delta_2(j)}((x,y)) =  \bm{1}_{\cup_j \Delta_{2}(j)}((x,y)) =\bm{1}_{ \Delta}((x,y)).\]
Hence, we obtain that
\begin{align}
\sum_{j=\ell}^{k} \sum_{i= j}^{k-1} \bar{\nu}_i 
	& = \frac{1}{m}    \sum_{x,y} \left( (y-\ell)\sum_{j=\ell}^{k-1} \bm{1}_{\Delta_1(j)}((x,y))   + (y-\ell-1) \sum_{j=\ell}^{k-1} \bm{1}_{\Delta_2(j)}((x,y)) \right) \nonumber \\
	& = \frac{1}{m} \sum_{x,y} \left((y-\ell) \bm{1}_{\Delta'}((x,y)) +  (y-\ell-1) \bm{1}_{\Delta}((x,y)) \right). \label{eq:k_less_d_Jan30_temple_3}
\end{align}

By substituting \eqref{eq:k_less_d_Jan30_temple_2} and \eqref{eq:k_less_d_Jan30_temple_3} in \eqref{eq:k_less_d_Jan30_temple_1}, we obtain
\begin{align*}
	& \frac{\Gamma_{k,d,\ell} }{d}- \frac{d+1-\ell}{\ell} \bar{\mu}_\ell- \frac{1}{\ell}\left( \sum_{j=\ell+1}^{k} \delta_j \right) \\
	& = \frac{\Gamma_{k,d,\ell}}{d}-  \frac{1}{\ell} \left( \sum_{j=\ell}^{k-1} \bar{\nu}_j \right) - \frac{1}{\ell} \sum_{j=\ell}^{k} (d+1-j)\bar{\mu}_j + \frac{1}{\ell} \sum_{j=\ell}^{k}\sum_{i= j}^{k-1} \bar{\nu}_i \\
	& = \utag{c}{=} \frac{\Gamma_{k,d,\ell}}{d}-  \frac{1}{\ell} \left( \sum_{j=\ell}^{k-1} \bar{\nu}_j \right) -  \frac{1}{m\ell} \sum_{x,y}\left(  (d+1-y)q_{x,y}   -   \bm{1}_{\Delta}((x,y)) \right) \\
	& ~~ + \frac{1}{m\ell} \sum_{x,y} \left((y-\ell) \bm{1}_{\Delta'}((x,y)) +  (y-\ell-1) \bm{1}_{\Delta}((x,y)) \right) \\
	& = \frac{\Gamma_{k,d,\ell}}{d}-  \frac{1}{\ell} \left( \sum_{j=\ell}^{k-1} \bar{\nu}_j \right) -  \frac{1}{m\ell} \sum_{x,y}\left(  (d+1-y)q_{x,y} - (y-\ell) \left( \bm{1}_{\Delta}((x,y)) +  \bm{1}_{\Delta'}((x,y)) \right) \right) \\
	& \utag{d}{=} \frac{\Gamma_{k,d,\ell}}{d}-  \frac{1}{\ell} \left( \sum_{j=\ell}^{k-1} \bar{\nu}_j \right) -  \frac{1}{m\ell} \sum_{x,y}\left(  (d+1-y)q_{x,y} - (y-\ell) (1-q_{x,y}) \right) \\
	& = \frac{\Gamma_{k,d,\ell}}{d}-  \frac{1}{\ell} \left( \sum_{j=\ell}^{k-1} \bar{\nu}_j \right) -  \frac{d+1-\ell}{m\ell} \sum_{x,y}q_{x,y} + \frac{1}{\ell} \sum_{y=\ell+1}^{k} \left(y-\ell  \right).
\end{align*}
where \uref{c} follows from \eqref{eq:k_less_d_Jan30_temple_2} and \eqref{eq:k_less_d_Jan30_temple_3}, and \uref{d} is justified because $\Delta$ and $\Delta'$ are disjoint and $\Delta \cup \Delta' = \left\{(x,y): q_{x,y} =0 \right\}$.

Since we know from  \eqref{eq:k_less_d_W_bound} and \eqref{eq:k_less_d_fQ}
that $m  \sum_{j=\ell}^{k-1} \bar{\nu}_j$ corresponds to the total number of zeros in the matrix $Q$, we have
\[m  \sum_{j=\ell}^{k-1} \bar{\nu}_j = \sum_{x,y} (1-q_{x,y}),\]
and so
\[\sum_{x,y} q_{x,y} = m(k-\ell) - m  \sum_{j=\ell}^{k-1} \bar{\nu}_j \utag{e}{=} m(k-\ell) - m  \frac{\Gamma_{k,d,\ell} }{d},\]
where \uref{e} follows from \eqref{eq:k_less_d_Condition1}.
Finally, we obtain that
\begin{align*}
	& \frac{\Gamma_{k,d,\ell} }{d}- \frac{d+1-\ell}{\ell} \bar{\mu}_\ell- \frac{1}{\ell}\left( \sum_{j=\ell+1}^{k} \delta_j \right) \\
	& = \frac{\Gamma_{k,d,\ell}}{d}-  \frac{1}{\ell} \left( \sum_{j=\ell}^{k-1} \bar{\nu}_j \right) -  \frac{d+1-\ell}{m\ell} \sum_{x,y}q_{x,y} + \frac{1}{\ell} \sum_{y=\ell+1}^{k} \left(y-\ell  \right) \\
	& = \frac{\Gamma_{k,d,\ell}}{d}-  \frac{1}{\ell} \frac{\Gamma_{k,d,\ell} }{d} -  \frac{d+1-\ell}{m\ell} \left(m(k-\ell) - m  \frac{\Gamma_{k,d,\ell} }{d} \right) + \frac{1}{\ell} \sum_{y=\ell+1}^{k} \left(y-\ell  \right) \\
	& = \frac{\Gamma_{k,d,\ell}}{\ell} -  \frac{(d+1-\ell)(k-\ell)}{\ell}  + \frac{1}{\ell} \sum_{y=\ell+1}^{k} \left(y-\ell  \right) \\
	& = \frac{1}{\ell} \left( \sum_{i=\ell+1}^{k} (d+1-i)  - \sum_{i=\ell+1}^{k}(d+1-\ell) +  \sum_{i=\ell+1}^{k} \left(i-\ell  \right)    \right) \\
	& = 0,
\end{align*}
which completes the proof.

\section{Proof of Proposition~\ref{Proposition:delta}}
\label{Appendix:proposition:delta}
For $j=\ell+1,\ldots,k-1$, $\bar{\delta}_j$ can be written as
\begin{align}
	\bar{\delta}_j 
	& = (d+1-j)\bar{\mu}_j - \sum_{i= j}^{k-1} \bar{\nu}_i  \nonumber \\
	& = (d+1-j)\left( 1- z_{\ell+1} - z_{j+1} + \frac{c_j}{d+1-j}   \right) 
	- \sum_{i= j}^{k-2} z_{i+2} \nonumber \\
	& = (d+1-j)\left( 1- z_{\ell+1} - z_{j+1} \right)+ \sum_{i=j+1}^k (d+1-i) \left(z_{\ell+i-j} - z_{\ell+i-j+1}\right) - \sum_{i= j}^{k-2} z_{i+2}  \nonumber \\
	& = (d+1-j)\left( 1- z_{\ell+1} - z_{j+1} \right) - \sum_{i= j}^{k-2} z_{i+2} + (d-j) z_{\ell+1} - \sum_{i=j+2}^{k}  z_{\ell+i-j}    -  (d+1-k)  z_{\ell+k-j+1} \nonumber \\ 
	& = (d+1-j)\left(1 - z_{j+1} \right) -  (d+1-k)  z_{\ell+k-j+1} - \sum_{i= j+2}^{k} z_{i}  - \sum_{i=j+1}^{k}  z_{\ell+i-j} \nonumber \\ 
	& =  (d+1-j)\left(1 - z_{j+1} \right) -  (d+1-k)  z_{\ell+k-j+1} - \sum_{i= j+2}^{k} z_{i}  - \sum_{i=\ell+1}^{\ell+k-j}  z_{i}. \label{eq:k_less_d_delta}
\end{align}

Recall that we need to prove that for any $(k,d,\ell) \in P_s$, $\bar{\delta}_j \geq 0$ for $j=\ell+1,\ldots,k-1$. Since we consider $\ell \leq k-2$, $(k,d,\ell) \in P_s$ implies that 
\begin{equation*}
  \begin{cases}
  	d(d-\ell-1) - \frac{1}{2}(2d-k-\ell+1)(2d+k-3\ell-5) \geq 0, & \ell \leq k-4, \\
  	k \geq \frac{1}{3}(d+8), & \ell= k - 3, \\
  	k \geq \frac{1}{4}(d+7), & \ell= k - 2.			
  \end{cases}
\end{equation*} 
First, we discuss the cases $\ell= k - 2$ and $\ell= k - 3$. When $\ell= k - 2$, we only need to verify that $\bar{\delta}_{\ell+1} \geq 0$ provided that $k \geq \frac{1}{4}(d+7)$. From \eqref{eq:k_less_d_delta}, we have 
\begin{align*}
	\bar{\delta}_{\ell+1} 
	& = (d-\ell)\left(1 - z_k \right) -  (d+1-k)  z_{k}   -z_{\ell+1} \\
	& = (d-\ell)  -  (2d-k-\ell+1)  z_{k}   - z_{\ell+1} \\
	& = (d-\ell)  -  (2d-k-\ell)  z_{k}   - \left( z_{\ell+1} + z_{k} \right) \\
	& \utag{a}{=} (d-\ell)  -  (2d-k-\ell)  z_{k}   - \frac{2d-k-\ell+1}{d},
\end{align*}
where \uref{a} follows from \eqref{eq:S3:z_plus}. 

Since we know from \eqref{eq:k_less_d_zj} that
\begin{equation}
\label{eq:Appendix_G:z_k}
	z_{k} =
\begin{cases}
\frac{d-k-\ell+1}{d},  & \ell < d-k+1, \\
	0    ,             & \ell \geq d-k+1, 
\end{cases}
\end{equation}
we have
\begin{equation}
\label{eq:Appendix_G:minus2_temple1}
\bar{\delta}_{\ell+1} = 
	\begin{cases}
		(d-\ell)  -    \frac{2(d-k+1)(d-k-\ell+1)}{d}   - \frac{2d-k-\ell+1}{d} , &   \ell < d-k+1, \\
		(d-\ell)     - \frac{2d-k-\ell+1}{d}    ,                               &   \ell \geq d-k+1.
	\end{cases}
\end{equation}
By re-arranging \eqref{eq:Appendix_G:minus2_temple1} and substituting $\ell=k-2$, we obtain that
\begin{equation}
\bar{\delta}_{\ell+1} = 
	\begin{cases}
		\frac{d(d-k+2)-2(d-k+1)(d-2k+3)-(2d-2k+3)}{d}  , &   k < \frac{1}{2}(d+3), \\
		\frac{d(d-k+2)-(2d-2k+3)}{d}   ,                               &   k \geq \frac{1}{2}(d+3).
	\end{cases}
\end{equation}
We can see that 
\[\frac{d(d-k+2)-(2d-2k+3)}{d} =  \frac{(d-2)(d-k+2)+1}{d} \geq 0\]
because $d \geq k = \ell+2 \geq 3$.
Then we only need to consider $k < \frac{1}{2}(d+3)$, and it remains to show that $g_1(k) \geq 0$ provided that $k \geq \frac{1}{4}(d+7)$ and $k < \frac{1}{2}(d+3)$, 
where
\[g_1(k) = d(d-k+2)-2(d-k+1)(d-2k+3)-(2d-2k+3).\]

For the quadratic equation $g_1(k) =0$, the discriminant is $9d^2-8d \geq 0$, so the two roots are given by 
$k_1= \frac{1}{8}\left(5d-\sqrt{9d^2-8d}+12\right)$  and $ k_2=\frac{1}{8}\left(5d + \sqrt{9d^2-8d}+12\right)$.
Since the leading coefficient of $g_1(k)$ is negative, we see that $g_1(k) \geq 0$ if and only if $k_1 \leq k \leq k_2$. 
Hence, to prove that $g_1(k) \geq 0$ provided that $k \geq \frac{1}{4}(d+7)$ and $k < \frac{1}{2}(d+3)$, it suffices to have $k_1 \leq \frac{1}{4}(d+7)$ and $k_2 \geq \frac{1}{2}(d+3)$, which can be shown by considering
\[k_2=\frac{1}{8}\left(5d + \sqrt{9d^2-8d}+12\right) \geq \frac{1}{8}\left(5d +12\right) \geq \frac{1}{2}(d+3),\]
and
\begin{align*}
	k_1  
	& =\frac{1}{4}(d+7) + \frac{1}{8}\left(3d-2 -\sqrt{9d^2-8d}\right) \\
	& = \frac{1}{4}(d+7) + \frac{1}{8}\left( \sqrt{(3d-2)^2} -\sqrt{9d^2-8d}\right) \\
	& = \frac{1}{4}(d+7) + \frac{1}{8}\left( \sqrt{9d^2-8d - 4(d-1)} -\sqrt{9d^2-8d}\right) \\
	& \leq \frac{1}{4}(d+7).
\end{align*}
This completes the proof.

When $\ell=k-3$, we need to verify that $\bar{\delta}_{\ell+1} \geq 0$ and $\bar{\delta}_{\ell+2} \geq 0$ provided that $k \geq \frac{1}{3}(d+8)$. First, focus on $\bar{\delta}_{\ell+1}$. From \eqref{eq:k_less_d_delta}, we have
\begin{align*}
	\bar{\delta}_{\ell+1} 
	& = (d-\ell)\left(1 - z_{\ell+2} \right) -  (d+1-k)  z_{k} -  z_{k}  -  z_{\ell+1} - z_{\ell+2} \\
	& = (d-\ell) -  (d-\ell+1) z_{\ell+2} -  (d+1-k)  z_{k} -   \left(z_{\ell+1} + z_{k}\right) \\
	& \utag{b}{=}  (d-\ell) -  (d-\ell+1) z_{\ell+2} -  (d+1-k)  z_{k} -\frac{2d-k-\ell+1}{d} \\
	& \utag{c}{=}  (d-\ell) -   \frac{(d-\ell+1)(2d-k-\ell+1)}{2d} -  (d+1-k)  z_{k} -\frac{2d-k-\ell+1}{d},
\end{align*}
where \uref{b} follows from \eqref{eq:S3:z_plus} and \uref{c} follows from \eqref{eq:k_less_d_zj}. Then from \eqref{eq:Appendix_G:z_k}, we have
\begin{equation}
\label{eq:Appendix_G:minus3_temple1}
\bar{\delta}_{\ell+1} = 
	\begin{cases}
		(d-\ell) -   \frac{(d-\ell+1)(2d-k-\ell+1)}{2d} -    \frac{(d+1-k)(d-k-\ell+1)}{d} -\frac{2d-k-\ell+1}{d} , &   \ell < d-k+1, \\
		(d-\ell) -   \frac{(d-\ell+1)(2d-k-\ell+1)}{2d}  -\frac{2d-k-\ell+1}{d}    ,                               &   \ell \geq d-k+1.
	\end{cases}
\end{equation}
If $\ell < d-k+1$, we have
\begin{align*}
	\bar{\delta}_{\ell+1} & = (d-\ell) -   \frac{(d-\ell+1)(2d-k-\ell+1)}{2d} -    \frac{(d+1-k)(d-k-\ell+1)}{d} -\frac{2d-k-\ell+1}{d}  \\
	& \utag{d}{=} (d-k+3) - \frac{(d-k+4)(d-k+2)}{d} -    \frac{(d+1-k)(d-2k+4)}{d} -\frac{2(d-k+2)}{d} \\
	& = \frac{(3k-d-8)(d-k+2)}{d} \\
	& \utag{e}{\geq} 0,
\end{align*}
where \uref{d} follows from substituting $\ell= k-3$ and \uref{e} follows from our assumption that $k \geq \frac{1}{3}(d+8)$. If $\ell \geq d-k+1$, we have
\begin{align*}
	\bar{\delta}_{\ell+1} & = (d-\ell) -   \frac{(d-\ell+1)(2d-k-\ell+1)}{2d}  -\frac{2d-k-\ell+1}{d} \\
	& \utag{f}{\geq}  \frac{(d-\ell)(2d-k-\ell+1)}{d}-   \frac{(d-\ell+1)(2d-k-\ell+1)}{2d}  -\frac{2d-k-\ell+1}{d} \\
	& =  \frac{(d-\ell-3)  (2d-k-\ell+1)}{2d} \\
	& = \frac{(d-k)  (2d-k-\ell+1)}{2d} \\
	& \geq 0,
\end{align*}
where \uref{f} follows because $\ell \geq d-k+1$ implies $\frac{2d-k-\ell+1}{d} \leq 1$. Thus, we obtain that $\bar{\delta}_{\ell+1} \geq 0$ provided that $k \geq \frac{1}{3}(d+8)$.

For $\bar{\delta}_{\ell+2}$, we obtain from \eqref{eq:k_less_d_delta} that 
\begin{align*}
	\bar{\delta}_{\ell+2} 
	& = (d-\ell-1)\left(1 - z_{k} \right) -  (d+1-k)  z_{\ell+2}   - z_{\ell+1} \\
	& = (d-\ell-1) - (d-\ell-2)z_{k} -  (d+1-k)  z_{\ell+2}    - \left(z_{\ell+1} + z_k\right) \\
	& \utag{g}{=} (d-\ell-1) - (d-\ell-2)z_{k} -  (d+1-k)  z_{\ell+2}   - \frac{2d-k-\ell+1}{d} \\
	& \utag{h}{=} (d-\ell-1) - (d-\ell-2)z_{k}  -    \frac{(d+1-k)(2d-k-\ell+1)}{2d}   - \frac{2d-k-\ell+1}{d},
\end{align*}
where \uref{g} follows from \eqref{eq:S3:z_plus} and \uref{h} follows from \eqref{eq:k_less_d_zj}. Then from \eqref{eq:Appendix_G:z_k}, we have
\begin{equation}
\label{eq:Appendix_G:minus3_temple2}
\bar{\delta}_{\ell+2} = 
	\begin{cases}
		(d-\ell-1) -\frac{ (d-\ell-2)(d-k-\ell+1)}{d}  -    \frac{(d+1-k)(2d-k-\ell+1)}{2d}   - \frac{2d-k-\ell+1}{d}, &   \ell < d-k+1, \\
		(d-\ell-1)  -    \frac{(d+1-k)(2d-k-\ell+1)}{2d}   - \frac{2d-k-\ell+1}{d},                               &   \ell \geq d-k+1.
	\end{cases}
\end{equation}
If $\ell \geq d-k+1$, we have
\begin{align*}
	\bar{\delta}_{\ell+2} & = (d-\ell-1) - \frac{(d+1-k)(2d-k-\ell+1)}{2d}   - \frac{2d-k-\ell+1}{d} \\
	& \utag{i}{\geq} \frac{ (d-\ell-1)(2d-k-\ell+1)}{d} - \frac{(d+1-k)(2d-k-\ell+1)}{2d}   - \frac{2d-k-\ell+1}{d} \\
	& = \frac{(d+k-2\ell-5)  (2d-k-\ell+1)}{2d}  \\
	& = \frac{(d-k+1)  (2d-k-\ell+1)}{2d} \\
	& \geq 0,
\end{align*}
where \uref{i} follows because $\ell \geq d-k+1$ implies $\frac{2d-k-\ell+1}{d} \leq 1$. Hence it remains to show that $\bar{\delta}_{\ell+2} \geq 0$ provided that $ \ell < d-k+1$ and $k \geq \frac{1}{3}(d+8)$. To see this, by substituting $\ell = k-3$, the first case of \eqref{eq:Appendix_G:minus3_temple2} can be written as 
\[
\bar{\delta}_{\ell+2} = (d-k+2) -\frac{ (d-k+1)(d-2k+4)}{d}  -    \frac{(d+1-k)(d-k+2)}{d}   - \frac{2(d-k+2)}{d}, ~~ k < \frac{1}{2}(d+4).
\]
Let 
\[g_2(k) = (d-k+2) -\frac{ (d-k+1)(d-2k+4)}{d}  -    \frac{(d+1-k)(d-k+2)}{d}   - \frac{2(d-k+2)}{d}.\]
Then we need to prove that $g_2(k) \geq 0$ provided that $k \geq \frac{1}{3}(d+8)$ and $k < \frac{1}{2}(d+4)$.
By rearranging the terms in $g_2(k)$, we have 
\begin{align*}
	g_2(k)
	& =  \frac{(d-k+2)(3k-d-8)+(2d-3k+6)}{d}.
\end{align*}
For the quadratic equation $g_2(k) =0$, the discriminant is $4d^2-8d +1 \geq 0$, so the two roots of $g_2(k) =0$ are given by
\[k_1 = \frac{1}{6} \left(4d+11 - \sqrt{4d^2-8d+1} \right), \]
and
\[k_2 = \frac{1}{6} \left(4d+11 + \sqrt{4d^2-8d+1} \right). \]
Since the leading coefficient of $g_2(k)$ is negative, we see that $g_2(k) \geq 0$ if and only if $k_1 \leq k \leq k_2$. 
Hence, to prove that $g_2(k) \geq 0$ provided that $k \geq \frac{1}{3}(d+8)$ and $k < \frac{1}{2}(d+4)$, we only need to show that $k_1 \leq \frac{1}{3}(d+8)$ and $k_2 \geq \frac{1}{2}(d+4)$. Consider 
\[k_2 \geq \frac{1}{6} \left(4d+11 \right) = \frac{1}{2}(d+4) + \frac{1}{6}(d-4) \utag{j}{\geq} \frac{1}{2}(d+4),\]
and 
\begin{align*}
	k_1 & = \frac{1}{6} \left(4d+11 - \sqrt{4d^2-8d+1} \right) \\
		& = \frac{1}{3}(d+8)+ \frac{1}{6} \left(2d-11 - \sqrt{4d^2-8d+1} \right) \\
		& \leq  \frac{1}{3}(d+8)+ \frac{1}{6} \left(\sqrt{(2d-11)^2} - \sqrt{4d^2-8d+1} \right) \\
		& = \frac{1}{3}(d+8)+ \frac{1}{6} \left(\sqrt{4d^2-8d+1-12(3d-10)} - \sqrt{4d^2-8d+1} \right) \\
		& \utag{k}{\leq} \frac{1}{3}(d+8),
\end{align*}
where \uref{j} and \uref{k} are justified because we have $d \geq k \geq \ell+3 \geq 4$. The proof is completed.

Now, we consider the case $\ell \leq k-4$. We need to show that for any given $(k,d,\ell)$, where $\ell \leq k-4$, if $g(\ell) \geq 0$ (c.f.\eqref{eq:main_resulst:g}), then $\bar{\delta}_j \geq 0$ for $j=\ell+1,\ldots,k-1$.

First, we claim that if $g(\ell) \geq 0$, then $\ell \geq d-k+1$. To see this, recall that we know 
from the discussion in Section~\ref{Sec:MainResults} that $g(\ell) \geq 0$ if and only if $\ell_1 \leq \ell \leq \ell_2$, where
\[\ell_1 = \frac{1}{3}\left(3d-k-1 - \sqrt{3(d-k)^2+12(d-4)+(k-8)^2} \right),\]
and
\[\ell_2 = \frac{1}{3}\left(3d-k-1 + \sqrt{3(d-k)^2+12(d-4)+(k-8)^2} \right).\]
Clearly, to justify the claim, we only need to show that $\ell_1 \geq d-k+1$. Consider
\begin{align*}
	\ell_1 & = \frac{1}{3}\left(3d-k-1 - \sqrt{3(d-k)^2+12(d-4)+(k-8)^2} \right) \\
	& = (d-k+1) + \frac{1}{3} \left(2k-4 - \sqrt{3(d-k)^2+12(d-4)+(k-8)^2} \right) \\
	& = (d-k+1) + \frac{1}{3} \left(\sqrt{(2k-4)^2} - \sqrt{3(d-k)^2+12(d-4)+(k-8)^2} \right) \\
	& = (d-k+1) + \frac{1}{3}  \left(\sqrt{(2k-4)^2} - \sqrt{(2k-4)^2 + 3d(d-2k+4)} \right). 
\end{align*}
We see that $\ell_1 \geq d-k+1$ if and only if $d-2k+4 \leq 0$. Hence, it remains to show that if $g(\ell) \geq 0$ for some $\ell \leq k-4$, then $d-2k+4 \leq 0$. 
Since we know from Section~\ref{Sec:MainResults} that if $g(\ell) \geq 0$ for some $\ell \leq k-4$, then $g(\ell') \geq 0$ for any $\ell'$ such that $\ell \leq \ell' \leq k-4$. In particular, we have $g(k-4) \geq 0$. Hence, we have
\[g(k-4) = d(d-k+3) - \frac{1}{2}(2d-2k+5)(2d-2k+7) = (2k-d-6) (d-k+3)  + \frac{1}{2} \geq 0.\]
Since $k$ and $d$ are integers and $d \geq k$, we must have $2k-d-6 \geq 0$, which implies that $d-2k+4 \leq 0$. Thus, we have proved the claim that if $g(\ell) \geq 0$, then $\ell \geq d-k+1$.

Under the condition $\ell \geq d-k+1$, \eqref{eq:k_less_d_zj} can be written as
\begin{equation}
\label{eq:appendix:leq_d}
	z_j=
	\begin{cases}
		\frac{2d-k-\ell+1}{d},    & j=\ell+1,  \\
		\frac{2d-k-\ell+1}{2d},                   & j=\ell+2,\ldots,k-1, \\
		0,                      & j=k.
	\end{cases}
\end{equation}
Now, we write $\bar{\delta}_j$ explicitly for all values of $j$. Recall from \eqref{eq:k_less_d_delta} that
\begin{align*}
	\bar{\delta}_j 
	& =  (d+1-j)\left(1 - z_{j+1} \right) -  (d+1-k)  z_{\ell+k-j+1} - \sum_{i= j+2}^{k} z_{i}  - \sum_{i=\ell+1}^{\ell+k-j}  z_{i}. 
\end{align*}
If $j=\ell+1$, we have 
\begin{align*}
	\bar{\delta}_{\ell+1} 
	&  =  (d-\ell)\left(1 - z_{\ell+2} \right) -  (d+1-k)  z_{k} - \sum_{i= \ell+3}^{k} z_{i}  - \sum_{i=\ell+1}^{k-1}  z_{i}  \\
	&  =  (d-\ell)\left(1 - z_{\ell+2} \right) -  (d+1-k)  z_{k} - \sum_{i= \ell+3}^{k-1} z_{i}  - z_{\ell+1} -  \sum_{i=\ell+2}^{k-1}  z_{i}  \\
	& \utag{l}{=} (d-\ell)\left(1 - \frac{2d-k-\ell+1}{2d} \right) - 0 -  \frac{(k-\ell-3)(2d-k-\ell+1)}{2d} \\
	& ~~   - \frac{2d-k-\ell+1}{d} - \frac{(k-\ell-2)(2d-k-\ell+1)}{2d}\\
	& =  (d-\ell) - \frac{(d+2k-3\ell-3)(2d-k-\ell+1)}{2d} \\
	& =  (d-\ell-1)  - \frac{(2d+k-3\ell-5)(2d-k-\ell+1)}{2d} +
	\frac{(d-k-2)(2d-k-\ell+1)}{2d} +1 \\
	& = \frac{g(\ell)}{d} +
	\frac{(d-k-2)(2d-k-\ell+1)}{2d} +1 \\
	& \utag{m}{\geq} \frac{g(\ell)}{d} +
	\frac{(d-k-2)(2d-k-\ell+1)}{2d} + \frac{2d-k-\ell+1}{d} \\
	& = \frac{g(\ell)}{d} +
	\frac{(d-k)(2d-k-\ell+1)}{2d}  \\
	& \geq \frac{g(\ell)}{d} \\
	& \utag{n}{\geq} 0,
\end{align*}
where  \uref{l} follows from \eqref{eq:appendix:leq_d}, \uref{m} follows because $\ell \geq d-k+1$ implies $\frac{2d-k-\ell+1}{d} \leq 1$, and  \uref{n} follows from the assumption that $g(\ell) \geq 0$. 

If $j=k-1$, we have
\begin{align*}
	\bar{\delta}_{k-1} 
	& =  (d -k +2)\left(1 - z_{k} \right) -  (d+1-k)  z_{\ell+2}  -  z_{\ell+1}  \\
	& \utag{o}{=}  (d -k +2) - \frac{(d+1-k)(2d-k-\ell+1)}{2d}    - \frac{2d-k-\ell+1}{d} \\
	& \utag{p}{\geq}  \frac{(2d-k-\ell+1)(d-k +2)}{d} - \frac{(d+1-k)(2d-k-\ell+1)}{2d}    - \frac{2d-k-\ell+1}{d}  \\
	& = \frac{(2d-k-\ell+1)(d-k +1)}{2d} \\
	& \geq 0,
\end{align*}
where  \uref{o} follows from \eqref{eq:appendix:leq_d}, and \uref{p} follows because $\ell \geq d-k+1$ implies $\frac{2d-k-\ell+1}{d} \leq 1$.

For $j=\ell+2,\ldots,k-2$, we have
\begin{align*}
	\bar{\delta}_j 
	& =  (d+1-j)\left(1 - z_{j+1} \right) -  (d+1-k)  z_{\ell+k-j+1} - \sum_{i= j+2}^{k} z_{i}  - \sum_{i=\ell+1}^{\ell+k-j}  z_{i}  \\
	& = (d+1-j)\left(1 - z_{j+1} \right) -  (d+1-k)  z_{\ell+k-j+1} - \sum_{i= j+2}^{k-1} z_{i}  - z_{\ell+1} - \sum_{i=\ell+2}^{\ell+k-j}  z_{i} \\
	& = (d+1-j)\left(1 - \frac{2d-k-\ell+1}{2d} \right) -    \frac{(d+1-k)(2d-k-\ell+1)}{2d} \\
	& ~~ -  \frac{(k-j-2)(2d-k-\ell+1)}{2d}  - \frac{2d-k-\ell+1}{d}-  \frac{(k-j-1)(2d-k-\ell+1)}{2d} \\
	& = (d+1-j) - \frac{(2d+k-3j+1)(2d-k-\ell+1)}{2d}.
\end{align*}
Then we consider the following two cases. 
\begin{itemize}
	\item If $d \geq \frac{3}{2}(2d-k-\ell+1)$, we have
		\begin{align*}
			\bar{\delta}_j 
			& = (d+1-j) - \frac{(2d+k-3j+1)(2d-k-\ell+1)}{2d} \\
			& = \frac{3(2d-k-\ell+1) - 2d }{2d} j + (d+1) - \frac{(2d+k+1)(2d-k-\ell+1)}{2d}  \\
			& \geq \frac{3(2d-k-\ell+1) - 2d}{2d} (k-2) + (d+1) - \frac{(2d+k+1)(2d-k-\ell+1)}{2d} \\
			& = (d-k+3) - \frac{1}{2d}(2d-k-\ell+1)(2d-2k+7) \\
			& \utag{q}{\geq} \frac{3}{2d}(2d-k-\ell+1)(d-k+3) - \frac{1}{2d}(2d-k-\ell+1)(2d-2k+7) \\
			& = \frac{1}{2d}(2d-k-\ell+1)(d-k+2)  \\
			& \geq 0,
		\end{align*}
		where \uref{q} follows because $d \geq \frac{3}{2}(2d-k-\ell+1)$ implies that $\frac{3}{2d}(2d-k-\ell+1) \leq 1$.
	\item If $d < \frac{3}{2}(2d-k-\ell+1)$, we have 
			\begin{align*}
			\bar{\delta}_j 
			& = (d+1-j) - \frac{(2d+k-3j+1)(2d-k-\ell+1)}{2d} \\
			& = \frac{3(2d-k-\ell+1) - 2d }{2d} j + (d+1) - \frac{(2d+k+1)(2d-k-\ell+1)}{2d}  \\
			& \geq \frac{3(2d-k-\ell+1) - 2d}{2d} (\ell+2) + (d+1) - \frac{(2d+k+1)(2d-k-\ell+1)}{2d} \\
			& = (d-\ell-1) - \frac{(2d+k-3\ell-5)(2d-k-\ell+1)}{2d} \\
			& = \frac{g(\ell)}{d}  \\
			& \geq 0.
		\end{align*}
\end{itemize}
Combining two two cases, we have $\bar{\delta}_j \geq 0$ for $j=\ell+2,\ldots,k-2$. Therefore, we have shown that $\bar{\delta}_j \geq 0$ for $j=\ell+1,\ldots,k-1$.


\bibliographystyle{IEEEtran}
\bibliography{ref}
\end{document}